\documentclass{elsarticle-modified}
\pdfoutput=1

\usepackage{microtype}
\usepackage{graphicx}

\usepackage{mathtools}
\usepackage{amsthm}
\usepackage{amssymb}
\usepackage{amsmath}
\usepackage{stmaryrd}

\usepackage[utf8]{inputenc}
\usepackage[T1]{fontenc}
\usepackage{lmodern}

\usepackage{pifont}

\usepackage[dvipsnames]{xcolor}
\usepackage{tikz-cd}
\usepackage{enumerate}
\renewcommand{\eqref}[1]{(\ref{#1})}

\usepackage{booktabs}

\usepackage{xurl}
\usepackage[colorlinks=true]{hyperref}
\usepackage[capitalise,noabbrev]{cleveref}

\NewCommandCopy{\saveciteauthor}{\citeauthor}
\renewcommand{\citeauthor}[1]{\begin{NoHyper}\saveciteauthor{#1}\end{NoHyper}}

\usepackage{etoolbox}

\newcommand{\mkTTurl}[1]{\href{https://www.cs.bham.ac.uk/~mhe/TypeTopology/Published.#1.html}{\texttt{#1}}}

\newcommand{\ALF}{\textsc{ALF}}
\newcommand{\Agda}{\textsc{Agda}}
\newcommand{\Coq}{\textsc{Coq}}
\newcommand{\TypeTopology}{\textsc{TypeTopology}}

\newcommand{\colonequiv}{\mathrel{\vcentcolon\mspace{-1.2mu}\equiv}}

\DeclarePairedDelimiter{\pa}{(}{)}
\DeclarePairedDelimiter{\set}{\{}{\}}
\DeclarePairedDelimiter{\squash}{\|}{\|}
\DeclarePairedDelimiter{\tosquash}{|}{|}
\DeclareMathOperator{\powerset}{\mathcal P}
\DeclareMathOperator{\lifting}{\mathcal L}

\DeclareMathOperator{\id}{id}
\DeclareMathOperator{\Id}{Id}
\DeclareMathOperator{\ev}{ev}
\DeclareMathOperator{\List}{List}
\DeclareMathOperator{\inl}{inl}
\DeclareMathOperator{\inr}{inr}
\DeclareMathOperator{\isprop}{is-prop}
\DeclareMathOperator{\totalspace}{\mathbb T}
\DeclareMathOperator{\Fin}{Fin}
\DeclarePairedDelimiter{\toquotient}{[}{]}

\newcommand{\lambdadot}[2]{\mathop{\lambda}{#1}\mathrel{.}#2}
\newcommand{\issmall}[1]{\operatorname{is}\;\!\mathcal{#1}\!\operatorname{-small}}

\newcommand{\surj}{\twoheadrightarrow}

\newcommand{\U}{\mathcal U}
\newcommand{\V}{\mathcal V}
\newcommand{\W}{\mathcal W}
\newcommand{\T}{\mathcal T}

\newcommand{\Zero}{\mathbf{0}}
\newcommand{\Nat}{\mathbb{N}}
\newcommand{\One}{\mathbf{1}}
\newcommand{\Two}{\mathbf{2}}

\newcommand{\fst}{{\operatorname{pr_1}}}

\newcommand{\below}{\mathrel\sqsubseteq}
\newcommand{\aboveorder}{\mathrel\sqsupseteq}

\newcommand{\DCPOnum}[3]{{\U_{#1}}\!\operatorname{-DCPO}_{\U_{#2},\U_{#3}}}
\newcommand{\DCPO}[3]{\mathcal{#1}\!\operatorname{-DCPO}_{\mathcal{#2},\mathcal{#3}}}

\newcommand{\retract}[2]{%
\!\!\begin{tikzcd}[ampersand replacement=\&]
{#1}\ar[r,"s",shift left,hookrightarrow] %
\& {#2}\ar[l,"r",shift left,two heads]\end{tikzcd}\!\!}

\newcommand{\retractalt}[4]{%
\!\!\begin{tikzcd}[ampersand replacement=\&]
{#1}\ar[r,"{#3}",shift left,hookrightarrow] %
\& {#2}\ar[l,"{#4}",shift left,two heads]\end{tikzcd}\!\!}

\newcommand{\Ind}[1]{\mathcal{#1}\!\operatorname{-Ind}}
\newcommand{\cof}{\mathrel{\lesssim}}
\newcommand{\ddset}{\mathop{\rotatebox[origin=c]{270}{\(\twoheadrightarrow\)}}}
\newcommand{\dset}{\mathop{\downarrow}}
\newcommand{\ideal}[1]{\(\mathcal{#1}\)-ideal}

\newcommand{\Idl}[1]{\mathcal{#1}\!\operatorname{-Idl}}
\newcommand{\Idlnum}[1]{\mathcal{U}_{#1}\!\operatorname{-Idl}}
\newcommand{\dyadics}{\mathbb D}
\DeclareMathOperator{\dyadicmiddle}{m}
\DeclareMathOperator{\dyadicleft}{l}
\DeclareMathOperator{\dyadicright}{r}
\DeclarePairedDelimiter{\ssfbrackets}{\llparenthesis}{\rrparenthesis}
\newcommand{\ssf}[2]{\ssfbrackets{#1 \Rightarrow #2}}

\newcommand{\initseg}{\mathbin\downarrow}

\theoremstyle{plain}
\newtheorem{theorem}{Theorem}
\numberwithin{theorem}{section}
\newtheorem{lemma}[theorem]{Lemma}
\newtheorem{proposition}[theorem]{Proposition}
\newtheorem{corollary}[theorem]{Corollary}
\theoremstyle{definition}
\newtheorem{definition}[theorem]{Definition}
\newtheorem{example}[theorem]{Example}
\theoremstyle{remark}
\newtheorem{remark}[theorem]{Remark}

\AtBeginEnvironment{definition}{\DeclareEmphSequence{\bfseries}}

\begin{document}

\title{Continuous and algebraic domains \\ in univalent foundations}

\author[1,2]{\texorpdfstring{Tom de Jong\corref{cor1}}{Tom de Jong}}
\ead{tom.dejong@nottingham.ac.uk}
\ead[url]{https://www.tdejong.com}

\author[1]{Mart\'in H\"otzel Escard\'o}
\ead{m.escardo@bham.ac.uk}
\ead[url]{https://www.cs.bham.ac.uk/~mhe}

\cortext[cor1]{Corresponding author}

\affiliation[1]{organization={School of Computer Science, University of Birmingham},
  city={Birmingham},
  country={United Kingdom}
}

\affiliation[2]{organization={School of Computer Science, University of Nottingham},
  city={Nottingham},
  country={United Kingdom}
}

\begin{keyword}
  Constructivity and predicativity \sep Propositional resizing \sep Univalent
  foundations \sep Homotopy type theory \sep Domain theory \sep Continuous and
  algebraic domain
  \MSC[2020]{06B35, 03B38, 03F65, 03B70, 18B35, 68Q55}
\end{keyword}

\begin{abstract}
  We develop the theory of continuous and algebraic domains in constructive and
  predicative univalent foundations, building upon our earlier work on basic
  domain theory in this setting.
  That we work predicatively means that we do not assume Voevodsky's
  propositional resizing axioms. Our work is constructive in the sense that we
  do not rely on excluded middle or the axiom of (countable) choice.
  To deal with size issues and give a predicatively suitable definition of
  continuity of a dcpo, we follow Johnstone and Joyal's work on continuous
  categories.
  Adhering to the univalent perspective, we explicitly distinguish between
  data and property.
  To ensure that being continuous is a property of a dcpo, we turn to the
  propositional truncation, although we explain that some care is needed to
  avoid needing the axiom of choice.
  We also adapt the notion of a domain-theoretic basis to the predicative
  setting by imposing suitable smallness conditions, analogous to the
  categorical concept of an accessible category.
  All our running examples of continuous dcpos are then actually examples of
  dcpos with small bases which we show to be well behaved predicatively. In
  particular, such dcpos are exactly those presented by small ideals.
  As an application of the theory, we show that Scott's \(D_\infty\) model of
  the untyped \(\lambda\)-calculus is an example of an algebraic dcpo with a
  small basis.
  Our work is formalised in the \Agda\ proof assistant and its ability to infer
  universe levels has been invaluable for our purposes.
\end{abstract}

\maketitle

\section{Introduction}

Domain theory~\cite{AbramskyJung1994} is a well-established subject in
mathematics and theoretical computer science with applications to programming
language semantics~\cite{Scott1972,Scott1993,Plotkin1977}, higher-type
computability~\cite{LongleyNormann2015}, topology, and
more~\cite{GierzEtAl2003}.
We explore the development of domain theory from the univalent point of
view~\cite{Voevodsky2015,HoTTBook}. This means that we work with the
stratification of types as singletons, propositions, sets, 1-groupoids, etc.
Our work does not require any higher inductive types other than the
propositional truncation, and the only consequences of univalence needed here
are function extensionality and propositional extensionality.
Additionally, we work constructively and predicatively, as described below.

\subsection{Constructivity}
That we work constructively means that we do not assume excluded middle,
or weaker variants, such as Bishop's LPO~\cite{Bishop1967}, or the axiom of
choice (which implies excluded middle), or its weaker variants, such as the
axiom of countable choice.
An~advantage of working constructively and
not relying on these additional
logical axioms is that our development is valid in every
\((\infty,1)\)-topos~\cite{Shulman2019} and not just those in which the logic is
classical.

Our commitment to constructivity has the particular consequence that we cannot
simply add a least element to a set to obtain the free pointed dcpo. Instead of
adding a single least element representing an undefined value, we must work with
a more complex type of partial elements~(\cref{sec:lifting}).
Similarly, the booleans under the natural ordering fail to be a dcpo, so we use
the type of (small) propositions, ordered by implication, instead.
Finally, we mention two other domain-theoretic aspects in this work
that require particular attention when working constructively.
Firstly, it is well-known that the are several inequivalent notions of a finite
subset in constructive mathematics and to characterise the compact elements of a
powerset we need to use the \emph{Kuratowski} finite
subsets~\cite{Kuratowski1920,Johnstone2002,CoquandSpiwack2010,FruminEtAl2018}.
Secondly, \emph{single step functions} are classically defined by a case
distinction (using excluded middle) on the ordering of elements. Constructively,
we cannot, in general, make this case distinction, so we use subsingleton
suprema to define single step functions instead.

\subsection{Predicativity}
Our work is predicative in the sense that we do not assume Voevodsky's
\emph{resizing} rules~\cite{Voevodsky2011,Voevodsky2015} or axioms. In
particular, powersets of small types are large.

There are several (philosophical, model-theoretic, proof-theoretic, etc.)
arguments for keeping the type theory predicative, see for
instance~\cite{Uemura2019,Swan2019a,Swan2019b},
and~\cite[Section~1.1]{deJongEscardo2023} for a brief overview, but here we only
mention one that we consider to be amongst the most interesting. Namely, the
existence of a computational interpretation of propositional impredicativity
axioms for univalent foundations is an open problem.

A common approach to deal with domain-theoretic size issues in a predicative
foundation is to work with information systems~\cite{Scott1982a,Scott1982b},
abstract bases~\cite{AbramskyJung1994} or formal
topologies~\cite{Sambin1987,Sambin2003,CoquandEtAl2003} rather than dcpos, and
approximable relations rather than \emph{Scott continuous functions}.
Instead, we work directly with dcpos and Scott continuous functions. In dealing
with size issues, we draw inspiration from category theory and make crucial use
of type universes and type equivalences to capture \emph{smallness}.
For example, in our development of the Scott model of
PCF~\cite{deJong2021a,Hart2020}, the dcpos have carriers in the second universe
\(\U_1\) and least upper bounds for directed families indexed by types in the
first universe \(\U_0\).
Moreover, up to equivalence of types, the order relation of the dcpos takes
values in the lowest universe \(\U_0\).
Seeing a poset as a category in the usual way, we can say that these dcpos are
large, but locally small, and have small filtered colimits.
The fact that the dcpos have large carriers is in fact unavoidable and
characteristic of predicative settings, as proved in~\cite{deJongEscardo2023}.

Because the dcpos have large carriers it is a priori not clear that complex
constructions of dcpos, involving countably infinite iterations of exponentials
for example, do not result in a need for ever-increasing universes and are
predicatively possible. We show that they are possible through a careful
tracking of type universe parameters, and this is illustrated by the
construction of Scott's \(D_\infty\).

Since keeping track of these universes is prone to mistakes, we have formalised
our work in \Agda~(see~\cref{sec:formalisation}); its ability to infer and keep
track of universe levels has been invaluable.

\subsection{Contributions}
In previous work~\cite{deJong2021a} we developed domain theory
in constructive and predicative univalent foundations and considered basic
applications in the semantics of programming languages, such as the Scott model
of PCF~\cite{Plotkin1977,Scott1993}.
However, we did not discuss a rich and deep topic in domain theory:
\emph{algebraic} and \emph{continuous} dcpos~\cite{GierzEtAl2003}.
We present a treatment of their theory including several examples in our
constructive and predicative approach, where we deal with size issues by taking
direct inspiration from category theory and the work of Johnstone and Joyal on
continuous categories~\cite{JohnstoneJoyal1982} in particular.

Classically, a dcpo \(D\) is said to be \emph{continuous} if for every element
\(x\) of \(D\) the set of elements \emph{way below} it is directed and has
supremum \(x\).
The problem with this definition in our foundational setup is that the type of
elements way below \(x\) is not necessarily small. Although this does not stop
us from asking it to be directed and having supremum \(x\), this still poses a
problem: for example, there would be no guarantee that its supremum is preserved
by a Scott continuous function, as it is only required to preserve suprema of
directed families indexed by small types.

Our solution is
to use the ind-completion to give a predicatively suitable definition of
continuity of a dcpo, following the category theoretic work by
\citeauthor{JohnstoneJoyal1982}~\cite{JohnstoneJoyal1982}.
Some care is needed to ensure that the resulting definition expresses a property
of a dcpo, rather than additional structure. This is of course
where the propositional truncation comes in useful, but there are two natural
ways of using the truncation. We show that one of them yields a well behaved
notion that serves as our definition of continuity, while the other, which we
call \emph{pseudocontinuity}, is problematic in a constructive context. In a
classical setting where the axiom of choice is assumed, the two notions
(continuity and pseudocontinuity) are equivalent.

Another approach is to turn to the notion of a
\emph{basis}~\cite[Section~2.2.2]{AbramskyJung1994}, but to include smallness
conditions. While we cannot expect the type of elements way below an element
\(x\) to be small, in many examples it is the case that the type of {basis}
elements way below \(x\) is small.
We show that if a dcpo has a small basis, then it is continuous. In fact, all
our running examples of continuous dcpos are actually examples of dcpos with
small bases. Moreover, dcpos with small bases are better behaved. For example,
they are locally small and so are their exponentials.
Furthermore, we show that having a small basis is equivalent to being presented
by ideals.
For algebraic dcpos, bases work especially well constructively, at least in the
presence of set quotients and univalence, as explained in
\cref{sec:basis-of-compact-elements}.
In particular, we show that Scott's \(D_\infty\), as originally conceived
in~\cite{Scott1972}, and recalled in the setting of predicative univalent
foundations in~\cref{sec:Scott-D-infty}, is algebraic and that it has a small
compact basis.

\subsection{Related work}

In short, the distinguishing features of our work are: (i) the adoption of
homotopy type theory as a foundation, (ii) a commitment to predicatively and
constructively valid reasoning, (iii) the use of type universes to avoid size
issues concerning large posets.

The standard works on domain theory, e.g.~\cite{AbramskyJung1994,GierzEtAl2003},
are based on traditional impredicative set theory with classical logic.
A constructive, topos valid, and hence impredicative, treatment of some domain
theory can be found in~\cite[Chapter~III]{Taylor1999}.

Domain theory has been studied predicatively in the setting of formal topology
\cite{Sambin1987,Sambin2003,CoquandEtAl2003} in
\cite{MaiettiValentini2004,Negri2002,SambinValentiniVirgili1996} and the more
recent categorical papers~\cite{Kawai2017,Kawai2021}. In this predicative
setting, one avoids size issues by working with information
systems~\cite{Scott1982a,Scott1982b}, abstract bases~\cite{AbramskyJung1994} or
formal topologies, rather than dcpos, and approximable relations rather than
Scott continuous functions.
Hedberg~\cite{Hedberg1996} presented some of these ideas in Martin-L\"of Type
Theory and formalised them in the proof assistant \ALF~\cite{Magnusson1995}, a
precursor to \Agda. A~modern formalisation in \Agda\ based on Hedberg's work was
recently carried out in Lidell's master thesis~\cite{Lidell2020}.

Our development differs from the above line of work in that it studies
dcpos directly and uses type universes to account for the fact that
dcpos may be large.\index{dcpo}
An advantage of this approach is that we can work with (Scott continuous)
functions rather than the arguably more involved (approximable) relations.

Another approach to formalising domain theory in type theory can be found
in~\cite{BentonKennedyVarming2009,Dockins2014}. Both formalisations study
\(\omega\)-chain complete preorders, work with setoids, and make use of \Coq's
impredicative sort~\texttt{Prop}.
A setoid is a type equipped with an equivalence relation that must be respected
by all functions. The particular equivalence relation given by equality is
automatically respected of course, but for general equivalence relations this
must be proved explicitly.
The aforementioned formalisations work with preorders, rather than posets,
because they are setoids where two elements \(x\) and \(y\) are related if
\(x \leq y\) and \(y \leq x\).
Our~development avoids the use of setoids thanks to the adoption of the
univalent point of view. Moreover, we work predicatively and we work with the
more general directed families rather than \(\omega\)-chains, as we intend the
theory to also be applicable to topology and algebra~\cite{GierzEtAl2003}.

There are also constructive accounts of domain theory aimed at program
extraction~\cite{BauerKavkler2009,PattinsonMohammadian2021}.
Both these works study \(\omega\)-chain complete posets (\(\omega\)-cpos) and
define notions of \(\omega\)-continuity for them.
The former~\cite{BauerKavkler2009} is notably predicative, but makes use of
additional logical axioms: countable choice, dependent choice and Markov's
Principle, which are validated by a realisability interpretation.
The latter~\cite{PattinsonMohammadian2021} uses constructive logic to extract
witnesses but employs classical logic in the proofs of correctness by phrasing
them in the double negation fragment of constructive logic.
By~contrast, we study (continuous) dcpos rather than (\(\omega\)-continuous)
\(\omega\)-cpos and is fully constructive without relying on additional
principles such as countable choice or Markov's Principle.

Yet another approach is the field of \emph{synthetic domain
  theory}~\cite{Rosolini1986,Rosolini1987,Hyland1991,Reus1999,ReusStreicher1999}.
Although the work in this area is constructive, it is still impredicative, as it
is based on topos logic; but more importantly it has a focus different from that
of regular domain theory. The aim is to isolate a few basic axioms and find
models in (realisability) toposes where every object is a domain and every
morphism is continuous. These models often validate additional axioms, such as
Markov's Principle and countable choice, both of which are crucially used in the
theory, as well as anti-classical axioms which contradict excluded middle. We have
a different goal, namely to develop regular domain theory constructively and
predicatively, but in a foundation compatible with excluded middle and choice,
while not relying on them.

Our treatment of continuous (and algebraic) dcpos is based on the work
of~\citeauthor{JohnstoneJoyal1982}~\cite{JohnstoneJoyal1982} which is situated
in category theory where attention must be paid to size issues even in an
impredicative setting.
In the categorical context, a smallness criterion similar to our notion of
having a small basis appears in \cite[Proposition~2.16]{JohnstoneJoyal1982}.
In contrast to Johnstone and Joyal~\cite{JohnstoneJoyal1982}, we use the
propositional truncation to ensure that the type of continuous dcpos is a
subtype of the type of dcpos. This, together with the related notion of
pseudocontinuity, are discussed in \cref{sec:pseudocontinuity}.
The particular case of a dcpo with a small compact basis is analogous to the
notion of an accessible category~\cite{MakkaiPare1989}.

In constructive set theory, our approach corresponds to working with partially
ordered classes as opposed to sets~\cite{Aczel2006}. Our notion of a small basis
for a dcpo (\cref{sec:small-bases}) is similar, but different from
\citeauthor{Aczel2006}'s notion of a set-generated
dcpo~\cite[Section~6.4]{Aczel2006}.
While Aczel requires the set \(\{b \in B \mid b \below x\}\) to be directed, we
instead require the set of elements in \(B\) that are way-below \(x\) to be
directed in line with the usual definition of a
basis~\cite[Section~2.2.6]{AbramskyJung1994}.

Finally, abstract bases were introduced by \citeauthor{Smyth1977} as
``R\nobreakdash-structures''~\cite{Smyth1977}. Our treatment of them and the
round ideal completion is closer to that of \citeauthor{AbramskyJung1994} in
the aforementioned \cite[Section~2.2.6]{AbramskyJung1994}, although ours is
based on families and avoids impredicative constructions.

\subsection{Departures from our previous work}

This paper presents a revised and expanded treatment of continuous and algebraic
domains compared to our conference paper~\cite{deJongEscardo2021a}. The
presentation of basic domain theory (\cref{sec:basic-domain-theory}) and the
construction of Scott's \(D_\infty\) in particular has been abridged for
brevity, but full details can be found in Chapter~3 of the first author's PhD
thesis~\cite{deJongThesis} as well as the accompanying formalisation
(\cref{sec:formalisation}).
In~\cite{deJongEscardo2021a} (and also \cite{deJong2021a}) the definition of a
poset included the requirement that the carrier is a set, because we only
realised later that this was redundant~(\cref{posets-are-sets}).

With the notable exception of
\cref{ordinals-structurally-continuous-but-no-small-basis,sec:basis-of-compact-elements},
the results of this paper can all be found in the aforementioned PhD
thesis~\cite{deJongThesis}.
Compared to the thesis, we also mention two terminological changes:
\begin{itemize}
\item Instead of writing ``\(\alpha\) is cofinal in \(\beta\)'', we now say that
  ``\(\beta\) exceeds \(\alpha\)'' (see \cref{def:exceeds}).
  The issue with saying ``cofinal'' is that this word is ordinarily used for
  two subsets where one is already contained in the other.
  In particular, two cofinal subsets have the same least upper bound (if it
  exists) and this was not the case with our usage of the word ``cofinal''.
\item We have swapped ``continuity structure'' for ``continuity data'' when
  discussing continuous dcpos to reflect that the morphisms do not preserve the
  data (cf.~\cref{continuity-prop-vs-data}).
  Accordingly, we no longer say that a dcpo is ``structurally continuous'';
  instead writing that a dcpo is ``equipped with continuity data''.
  Similar terminological changes apply to the algebraic~case.
\end{itemize}

The present treatment of continuous and algebraic dcpos and small (compact)
bases is significantly different from that of our earlier
work~\cite{deJongEscardo2021a}. There, the definition of continuous dcpo was an
amalgamation of pseudocontinuity and having a small basis, although it did not
imply local smallness.
In this work we have disentangled the two notions and based our definition of
continuity on Johnstone and Joyal's notion of a continuous
category~\cite{JohnstoneJoyal1982} without making any reference to a basis. The
current notion of a small basis is simpler and slightly stronger than that of
our conference paper~\cite{deJongEscardo2021a}, which allows us to prove that
having a small basis is equivalent to being presented by ideals.

\subsection{Formalisation}\label{sec:formalisation}

All of our results are formalised in \Agda,
building on Escard\'o's \TypeTopology\ development~\cite{TypeTopology}.
Hart's previously cited work~\cite{Hart2020} was also ported to the current
\TypeTopology\ development by Escard\'o~\cite{TypeTopologyHart}.
The reference~\cite{TypeTopologyPaper} precisely links each numbered environment
(including definitions, examples and remarks) in this paper to its
implementation.
The HTML rendering has clickable links and so is particularly suitable
for exploring the development.
But this paper is self-contained and can be read independently from
the formalisation.

\subsection{Organisation}
\begin{description}
\item[\normalfont\cref{sec:foundations}:] A brief introduction to univalent
  foundations with a particular focus on type universes and the propositional
  truncation, as well as a discussion of impredicativity in the form of
  Voevodsky's propositional resizing axioms.
\item[\normalfont\cref{sec:basic-domain-theory}:] An abridged overview of basic
  domain theory in constructive and predicative univalent foundations, including
  directed complete posets (dcpos), Scott continuous maps, the lifting of a set,
  exponentials and bilimits of dcpos, and Scott's \(D_\infty\) model of the
  untyped \(\lambda\)-calculus.
  \item[\normalfont\cref{sec:way-below}:] Definition and examples of the way-below
    relation and compact elements.
  \item[\normalfont\cref{sec:ind-completion}:] The ind-completion of a preorder:
    a tool used to discuss continuity and pseudocontinuity of dcpos.
  \item[\normalfont\cref{sec:continuous-and-algebraic-dcpos}:] Definitions of
    continuous and algebraic dcpos accompanied by a discussion on
    pseudocontinuity and issues concerning the axiom of choice.
  \item[\normalfont\cref{sec:small-bases}:] The notion of a small (compact)
    basis: strengthening continuity (resp.~algebraicity) by imposing smallness
    conditions.
  \item[\normalfont\cref{sec:round-ideal-completion}:] The (round) ideal
    completion of an abstract basis as a continuous dcpo with a small basis.
  \item[\normalfont\cref{sec:continuous-bilimits-and-exponentials}:]
    Bilimits and exponentials of structurally continuous (or algebraic) dcpos
    (with small bases), including a proof that Scott's \(D_\infty\) is algebraic
    with a small compact basis.
\end{description}

\section{Foundations}\label{sec:foundations}
We work within intensional Martin-L\"of Type Theory and we include \(+\)~(binary
sum), \(\Pi\)~(dependent product), \(\Sigma\)~(dependent sum), \(\Id\)
(identity type), and inductive types, including~\(\Zero\)~(empty type),
\(\One\)~(type with exactly one element \(\star : \One\)) and \(\Nat\)~(natural
numbers).
In general we adopt the same conventions of~\cite{HoTTBook}.  In particular, we
simply write \(x=y\) for the identity type \(\Id_{X}(x,y)\) and use \(\equiv\)
for the judgemental equality, and for dependent functions
\(f,g : \Pi_{x : X}A(x)\), we write \(f \sim g\) for the pointwise equality
\(\Pi_{x : X} f(x) = g(x)\).
\subsection{Universes}\label{sec:universes}
We assume a universe \(\U_0\) and two operations: for every universe \(\U\), a
successor universe \(\U^+\) with \(\U : \U^+\), and for every two universes
\(\U\) and \(\V\) another universe \(\U \sqcup \V\) such that for any
universe~\(\U\), we have \(\U_0 \sqcup \U \equiv \U\) and
\(\U \sqcup \U^+ \equiv \U^+\). Moreover, \((-)\sqcup(-)\) is idempotent,
commutative, associative, and \((-)^+\) distributes over \((-)\sqcup(-)\). We
write \(\U_1 \colonequiv \U_0^+\), \(\U_2 \colonequiv \U_1^+, \dots\) and so on.
If \(X : \U\) and \(Y : \V\), then \({X + Y} : \U \sqcup \V\) and if \(X : \U\)
and \(Y : X \to \V\), then the types \(\Sigma_{x : X} Y(x)\) and
\(\Pi_{x : X} Y(x)\) live in the universe \(\U \sqcup \V\); finally,
if~\(X : \U\) and \(x,y : X\), then \(\Id_{X}(x,y) : \U\). The type of natural
numbers \(\Nat\) is assumed to be in \(\U_0\) and we postulate that we have
copies \(\Zero_{\U}\) and \(\One_{\U}\) in every universe \(\U\).
This has the useful consequence that while we do not assume cumulativity of
universes, embeddings that lift types to higher universes are definable. For
example, the map \((-) \times \One_{\V}\) takes a type in any universe \(\U\) to
an equivalent type in the higher universe \(\U \sqcup \V\).
All our examples go through with just
two universes \(\U_0\) and \(\U_1\), but the theory is more easily developed in
a general setting.

\subsection{The univalent point of view}
Within this type theory, we adopt the univalent point of view~\cite{HoTTBook}.
A type \(X\) is a \emph{proposition} (or \emph{truth value} or
\emph{subsingleton}) if it has at most one element, i.e.\ we have an element of the type
\(\isprop(X) \colonequiv \prod_{x,y : X} x = y\).
A major difference between univalent foundations and other foundational systems
is that we \emph{prove} that types are propositions or properties. For~instance,
we can show (using function extensionality) that the axioms of directed complete
poset form a proposition.
A type \(X\) is a \emph{set} if any two elements can be identified in at most
one way, i.e.\ we have an element of the type
\(\prod_{x,y : X} \isprop(x = y)\).

\subsection{Extensionality axioms}
The univalence axiom~\cite{HoTTBook} is not needed for our development, although
we do pause to point out its consequences in a few places, namely in
\cref{sec:impredicativity,sec:basis-of-compact-elements,equality-of-continuous-dcpos}.

We assume function extensionality and propositional extensionality, often
tacitly:
\begin{enumerate}[(i)]
\item \emph{Propositional extensionality}: if \(P\) and \(Q\) are two
  propositions, then we postulate that \(P = Q\) holds exactly when we have both
  \(P \to Q\) and \(Q \to P\).
\item \emph{Function extensionality}: if \(f,g : \prod_{x : X}A(x)\) are two
  (dependent) functions, then we postulate that \(f = g\) holds exactly when
  \(f \sim g\).
\end{enumerate}
Function extensionality has the important consequence that the propositions form
an exponential ideal, i.e.\ if \(X\) is a type and \(Y : X \to \U\) is such that
every \(Y(x)\) is a proposition, then so is
\(\Pi_{x : X}Y(x)\)~\cite[Example~3.6.2]{HoTTBook}. In light of this, universal
quantification is given by \(\Pi\)-types in our type~theory.

\subsection{Propositional truncation}
In Martin-L\"of Type Theory, an element of
\(\prod_{x : X}\sum_{y : Y}\phi(x,y)\), by definition, gives us a function
\(f : X \to Y\) such that \(\prod_{x : X}\phi(x,f(x))\). In some cases, we wish
to express the weaker ``for every \(x : X\), there exists some \(y : Y\) such
that \(\phi(x,y)\)'' without necessarily having an assignment of \(x\)'s to
\(y\)'s. A good example of this is when we define directed families later (see
\cref{def:directed-family}). This is achieved through the notion of
propositional truncation.

Given a type \(X : \U\), we postulate that we have a proposition
\(\squash*{X} : \U\) with a function \({\tosquash{-} : X \to \squash*{X}}\) such
that for every proposition \(P : \V\) in any universe \(\V\), every function
\(f : X \to P\) factors (necessarily uniquely, by function extensionality)
through \(\tosquash{-}\).
Diagrammatically,
\begin{equation*}
  \begin{tikzcd}
    X \ar[dr, "\tosquash*{-}"'] \ar[rr, "f"] & & P \\
    & \squash*{X} \ar[ur, dashed]
  \end{tikzcd}
\end{equation*}

Notice that the induction and recursion principles automatically hold up to an
identification: writing \(\bar f\) for the dashed map above, we have an
identification \(\bar{f}(\tosquash{x}) = f(x)\) for every \(x : X\) because
\(P\) is assumed to be a proposition.
This is sufficient for our purposes and we do not require these equalities
to hold judgementally.

Existential quantification \(\exists_{x : X}Y(x)\) is given by
\(\squash*{\Sigma_{x : X}Y(x)}\). One should note that if we have
\(\exists_{x : X}Y(x)\) and we are trying to prove some proposition \(P\), then
we may assume that we have \(x : X\) and \(y : Y(x)\) when constructing our
element of \(P\). Similarly, we can define disjunction as
\(P \lor Q \colonequiv \squash*{P + Q}\).

We assume throughout that every universe is closed under propositional
truncations, meaning that if \(X : \U\) then \(\squash{X} : \U\) as well.
We also stress that propositional truncation is the only higher inductive
type used in our work.

Finally we recall a useful result due to
\citeauthor{KrausEtAl2017}~\cite[Theorem~5.4]{KrausEtAl2017} which has several
applications in this paper.

\begin{lemma}\label{constant-map-to-set-factors-through-truncation}
  Every constant map to a set factors through the truncation of its domain.
  Here, a map is constant if any two of its values are equal.
\end{lemma}

\subsection{Size and impredicativity}\label{sec:impredicativity}
We introduce the notion of smallness and use it to define propositional resizing
axioms, which we take to be the definition of impredicativity in univalent
foundations.

\begin{definition}[Smallness]
    A type \(X\) in any universe is said to be \emph{\(\U\)-small} if it is
    equivalent to a type in the universe \(\U\). That is,
    \({X \issmall{\U}} \colonequiv \Sigma_{Y : \U} \pa*{Y \simeq X}\).
\end{definition}

Here, the symbol \(\simeq\) refers to Voevodsky's notion of equivalence
\cite{HoTTBook}. Notice that the type that expresses the \(\U\)-smallness of
\(X\) is a proposition if and only if the univalence axiom holds,
see~\cite[Sections~3.14 and 3.36.3]{Escardo2019}.

\begin{definition}[Type of propositions \(\Omega_{\U}\)]
  The type of propositions in a universe \(\U\) is
  \[
    \Omega_{\U} \colonequiv \sum_{P : \U} \isprop(P) : \U^+.
  \]
\end{definition}

Observe that \(\Omega_{\U}\) itself lives in the successor universe
\(\U^+\). We often think of the types in some fixed universe \(\U\) as
\emph{small} and accordingly we say that \(\Omega_{\U}\) is
\emph{large}.
Similarly, the powerset of a type \(X : \U\) is large.  Given our
predicative setup, we must pay attention to universes when considering
powersets:

\begin{definition}[\(\V \)-powerset \(\powerset_{\V }(X)\), \(\V \)-subsets]\label{def:powerset}
  Let \(\V \) be a universe and \(X : \U \) type. We~define the
  \emph{\(\V \)-powerset} \(\powerset_{\V }(X)\) as
  \(X \to \Omega_{\V} : \V^+\sqcup \U \). Its elements are called
  \emph{\(\V \)-subsets} of \(X\).
\end{definition}
\begin{definition}[\(\in,\subseteq\)]\label{def:membership}
  Let \(x\) be an element of a type \(X\) and let \(A\) be an element of the
  powerset \(\powerset_{\V }(X)\). We write \(x \in A\) for the type
  \(\fst\pa*{A(x)}\).  The first projection \(\fst\) is needed because \(A(x)\),
  being of type \(\Omega_\V\), is a pair. Given two \(\V \)-subsets
  \(A\)~and~\(B\) of \(X\), we write \(A \subseteq B\) for
  \(\prod_{x : X}\pa*{x \in A \to x \in B}\).
\end{definition}
Function extensionality and propositional extensionality imply that \(A=B\) if and only if
\(A \subseteq B\) and \(B \subseteq A\).

\begin{definition}[Total space of a subset, \(\totalspace\)]%
  \label{def:total-space}%
  The \emph{total space} of a \(\T\)-valued subset \(S\) of a type \(X\) is
  defined as \(\totalspace(S) \colonequiv \Sigma_{x : X} (x \in S)\).
\end{definition}

One could ask for a \emph{resizing axiom} asserting that \(\Omega_{\U}\) has
size \(\U\), which we call \emph{the propositional impredicativity of \(\U\)}. A
closely related axiom is \emph{propositional resizing}, which asserts that every
proposition \(P : \U^+\) is \(\U\)-small. Without the addition of such resizing
axioms, the type theory is said to be \emph{predicative}.  As an example of the
use of impredicativity in mathematics, we mention that the powerset has unions
of arbitrary subsets if and only if propositional resizing
holds~\cite[Section~3.36.6]{Escardo2019}.

We note that the resizing axioms are actually theorems when classical logic
is assumed. This is because if \(P \lor \lnot P\) holds for every proposition in
\(P : \U\), then the only propositions (up to equivalence) are \(\Zero_{\U}\)
and \(\One_{\U}\), which have equivalent copies in \(\U_0\), and
\(\Omega_{\U}\) is equivalent to a type \(\Two_{\U} : \U\) with exactly two
elements.

\section{Basic domain theory in univalent foundations}\label{sec:basic-domain-theory}

We review basic domain theory in constructive and predicative univalent
foundations, laying the foundations for developing the theory of continuous and
algebraic domains.
For brevity, we only include proofs when they deviate from their classical
counterparts, paying special attention to universe parameters, the distinction
between data and property, and the use of the propositional truncation. We note
that full details may be found in Chapter~3 of the first author's PhD
thesis~\cite{deJongThesis} or the accompanying
formalisation~\cite{TypeTopologyPaper}.

\subsection{Introduction to constructive and predicative domain theory}
We offer the following overture in preparation of our development, especially if
the reader is familiar with domain theory in a classical, set-theoretic setting.

The basic object of study in domain theory is that of a \emph{directed complete
  poset} (dcpo).
In (impredicative) set-theoretic foundations, a dcpo can be defined to be a
poset that has least upper bounds of all directed subsets.
A naive translation of this to our foundation would be to proceed as
follows. Define a poset in a universe \(\U\) to be a type \(P:\U\) with a
reflexive, transitive and antisymmetric relation
\(-\below- : P \times P \to \U\).
Since we wish to consider posets and not categories we require that the values
\(p \below q\) of the order relation are \emph{subsingletons}.
Then we could say that the poset \((P,\below)\) is \emph{directed complete} if
every directed family \(I \to P\) with indexing type \(I : \U\) has a least
upper bound (supremum). The problem with this definition is that there are no
interesting examples in our constructive and predicative setting.
For instance, assume that the poset~$\Two$ with two elements \(0\below 1\) is
directed complete, and consider a proposition~\(A:\U\) and the directed family
\(A + \One \to \Two\) that maps the left component to~\(0\) and the right
component to~\(1\). By case analysis on its hypothetical
supremum, we conclude that the negation of \(A\) is
decidable. This amounts to weak excluded middle~(which is equivalent to De
Morgan's Law) and is constructively unacceptable.

To try to get an example, we may move to the poset \(\Omega_{\U_0}\) of
propositions in the universe \(\U_0\), ordered by implication. This poset does
have all suprema of families \(I \to \Omega_{\U_0}\) indexed by types \(I\) in
the \emph{first universe} \(\U_0\), given by existential quantification. But if
we consider a directed family \(I \to \Omega_{\U_0}\) with \(I\) in the
\emph{same universe} as \(\Omega_{\U_0}\) lives, namely the \emph{second
  universe} \(\U_1\), existential quantification gives a proposition in the
\emph{second universe} \(\U_1\) and so doesn't give its supremum. In this
example, we get a poset such that
\begin{enumerate}[(i)]
\item the carrier lives in the universe \(\U_1\),
\item the order has truth values in the universe \(\U_0\), and
\item suprema of directed families indexed by types in \(\U_0\) exist.
\end{enumerate}

Regarding a poset as a category in the usual way, we have a large, but locally
small, category with small filtered colimits (directed suprema). This is typical
of all the concrete examples that we consider, such as the dcpos in the
Scott model of PCF~\cite{deJong2021a} and Scott's \(D_\infty\)
model of the untyped \(\lambda\)-calculus (\cref{sec:Scott-D-infty}).
We may say that the predicativity restriction increases
the universe usage by one.  However, for the sake of generality, we formulate
our definition of dcpo with the following universe conventions:
\begin{enumerate}[(i)]
\item the carrier lives in a universe \(\U\),
\item the order has truth values in a universe \(\T\), and
\item suprema of directed families indexed by types in a universe \(\V\) exist.
\end{enumerate}
So our notion of dcpo has three universe parameters \(\U,\V\) and \(\T\). We
will say that the dcpo is \emph{locally small} when \(\T\) is not necessarily
the same as \(\V\), but the order has \(\V\)-small truth values. Most of the
time we mention \(\V\) explicitly and leave \(\U\) and \(\T\) to be understood
from the context.

\subsection{Directed complete posets indexed by universe parameters}
We now define directed complete poset in constructive and predicative univalent
foundations. We carefully explain our use of the propositional truncation in our
definitions and, as mentioned above, the type universes involved.

\begin{definition}[Preorder and poset]
  A \emph{preorder} \((P,\sqsubseteq)\) is a type \(P : \U \) together with a
  proposition-valued binary relation \({\sqsubseteq} : {P \to P \to \Omega_\T}\)
  that is reflexive and transitive.
  A \emph{poset} is a preorder \((P,\below)\) that is antisymmetric: if
  \(p \below q\) and \(q \below p\), then \(p = q\) for every \(p,q : P\).
\end{definition}

\begin{lemma}\label{posets-are-sets}
  If \((P,\below)\) is a poset, then \(P\) is a set.
\end{lemma}
\begin{proof}
  For every \(p,q : P\), the composite
  \[
    \pa{p = q} \xrightarrow{\text{by reflexivity}}
    {\pa{p \below q} \times \pa{q \below p}} \xrightarrow{\text{by antisymmetry}}
    \pa{p = q}
  \]
  is constant since \({\pa{p \below q} \times \pa{q \below p}}\) is a
  proposition. By \cite[Lemma~3.11]{KrausEtAl2017} it therefore follows that
  \(P\) must be a set.
\end{proof}

From now on, we will simply write ``let \(P\) be a poset'' leaving the partial
order \(\below\) implicit. We will often use the symbol \({\below}\) for partial
orders on different carriers when it is clear from the context which one it
refers to.

\begin{definition}[(Semi)directed family]\label{def:directed-family}
  A family \(\alpha : I \to P\) of elements of a poset \(P\) is
  \emph{semidirected} if whenever we have \(i,j : I\), there exists some
  \(k : I\) such that \(\alpha_i \below \alpha_k\) and
  \(\alpha_j \below \alpha_k\).
  We frequently use the shorthand \({\alpha_i,\alpha_j} \below \alpha_k\) to
  denote the latter requirement.
  Such a family is \emph{directed} if it is semidirected and its domain \(I\) is
  inhabited.%
\end{definition}

\begin{remark}
  Note our use of the propositional truncation in defining when a family is
  \emph{directed}. To make this explicit, we write out the definition in
  type-theoretic syntax: a family \(\alpha : I \to P\) is directed if
  \begin{enumerate}[(i)]
  \item\label{dir-inh} we have an element of \(\squash{I}\), and
  \item\label{dir-semidir}
    \(\Pi_{i,j : I} \squash*{\Sigma_{k : I}\pa*{\alpha_i \below \alpha_k} \times
      \pa*{\alpha_j \below \alpha_k}}\).
  \end{enumerate}
  The use of the propositional truncation ensures that the types \eqref{dir-inh}
  and \eqref{dir-semidir} are propositions and hence that being (semi)directed
  is a property of a family.
  The type \eqref{dir-semidir} without truncation would express an
  assignment of a chosen \(k : I\) for every \(i,j : I\) instead.
\end{remark}

\emph{Least upper bounds} or \emph{suprema} (of families) are defined as usual.

\begin{definition}[\(\V \)-directed complete poset, \(\V\)-dcpo, %
  \(\bigsqcup \alpha\), \(\bigsqcup_{i : I}\alpha_i\)]
  For a universe~\(\V\), a \emph{\(\V \)-directed complete poset} (or
  \emph{\(\V \)-dcpo}, for short) is a poset \(D\) such that every directed
  family \(\alpha : I \to D\) with \(I : \V \) has a supremum in \(D\) that we
  denote by \(\bigsqcup \alpha\) or \(\bigsqcup_{i : I} \alpha_i\).
\end{definition}

\begin{remark}\label{directed-completeness-is-prop}
  Explicitly, we ask for an element of the type
  \[
    \Pi_{I : \V}\Pi_{\alpha : I \to D}\pa*{\operatorname{is-directed} \alpha \to
      \Sigma_{x : D}\pa*{x \mathrel{\operatorname{is-sup-of}} \alpha}},
  \]
  where \(\pa*{x \mathrel{\operatorname{is-sup-of}} \alpha}\) is the type expressing
  that \(x\) is the supremum of \(\alpha\).
  Even though we used \(\Sigma\) and not \(\exists\) in this expression, this
  type is still a proposition: By \cite[Example~3.6.2]{HoTTBook}, it suffices to
  prove that the type
  \(\Sigma_{x : D}(x \mathrel{\operatorname{is-sup-of}} \alpha)\) is a
  proposition. So suppose that we have \(x,y : D\) with
  \(p : x \mathrel{\operatorname{is-sup-of}} \alpha\) and
  \(q : y \mathrel{\operatorname{is-sup-of}} \alpha\). Being the supremum of a
  family is a property because the partial order is proposition-valued. Hence,
  by \cite[Lemma~3.5.1]{HoTTBook}, to prove that \((x,p) = (y,q)\), it suffices
  to prove that \(x = y\). But this follows from antisymmetry and the fact that
  \(x\) and \(y\) are both suprema of \(\alpha\).
\end{remark}

We will sometimes leave the universe \(\V \) implicit, and simply speak of a
dcpo. On other occasions, we need to carefully keep track of universe levels. To
this end, we make the following definition.
\begin{definition}[\(\DCPO{V}{U}{T}\)]\label{def:DCPO}
  Let \(\V\), \(\U\) and \(\T\) be universes. We write \(\DCPO{V}{U}{T}\) for
  the type of \(\V \)-dcpos with carrier in \(\U \) and order taking values in
  \(\T \).
  We often leave the parameters \(\U\) and \(\T\) implicit.
\end{definition}

\begin{remark}\label{universe-levels-of-lifting-and-exponentials}
  In particular, it is very important to keep track of the universe parameters
  of the lifting~(\cref{sec:lifting}) and of
  exponentials~(\cref{sec:exponentials}) in order to ensure that it is possible
  to construct Scott's \(D_\infty\) (\cref{sec:Scott-D-infty}) and the Scott
  model of PCF~\cite{deJong2021a} in our predicative setting.
\end{remark}

In many examples and applications, we deal with dcpos with a least element,
denoted by \(\bot_D\) or simply \(\bot\), and in which case we speak of
\emph{pointed} dcpos.

\begin{definition}[Local smallness]\label{def:local-smallness}
  A \(\V\)-dcpo \(D\) is \emph{locally small} if \(x \below y\) is \(\V\)-small
  for every \(x,y : D\).
\end{definition}

\begin{lemma}\label{local-smallness-alt}
  A \(\V\)-dcpo \(D\) is locally small if and only if we have a relation
  \({\below_{\V}} : D \to D \to \V\) such that \(x \below y\) holds
  precisely when \(x \below_{\V} y\) does.
\end{lemma}
\begin{proof}
  The \(\V\)-dcpo \(D\) is locally small exactly when we have an element of
  \[
    \Pi_{x,y : D}\Sigma_{T : \V}\pa*{T \simeq {(x \below y)}}.
  \]
  But this type is equivalent to
  \[
    \Sigma_{R : {D \to D \to \V}}\Pi_{x,y : D}\pa*{R(x,y) \simeq {(x \below y)}}
  \]
  by distributivity of \(\Pi\) over \(\Sigma\)~\cite[Theorem~2.5.17]{HoTTBook}.
\end{proof}

Nearly all examples of \(\V\)-dcpos in this paper will be locally small. We now
introduce two fundamental examples of dcpos: the type of subsingletons and
powersets.

\begin{example}[The type of subsingletons as a pointed dcpo]\label{Omega-as-pointed-dcpo}
  For any type universe~\(\V\), the type \(\Omega_{\V}\) of subsingletons in
  \(\V\) is a poset if we order the propositions by implication.
  Note that antisymmetry holds precisely because of propositional
  extensionality.
  Moreover, \(\Omega_{\V}\) has a least element, namely \(\Zero_{\V}\), the
  empty type in~\(\V\), and suprema for all (not necessarily directed) families
  indexed by a type in \(\V\).
  Finally, paying attention to the universe levels we observe that
  \(\Omega_{\V} : \DCPO{V}{V^+}{V}\), hence it is locally small.
\end{example}

\begin{example}[Powersets as pointed dcpos]\label{powersets-as-pointed-dcpos}
  Recalling our treatment of subset and powersets from
  \cref{sec:impredicativity}, we show that powersets give examples of
  pointed dcpos.
  Specifically, for every type \(X : \U\) and every type universe \(\V\), the
  subset inclusion \(\subseteq\) makes \(\powerset_{\V}(X)\) into a poset, where
  antisymmetry holds by function extensionality and propositional
  extensionality.
  Moreover, \(\powerset_{\V}(X)\) has a least element of course: the empty set
  \(\emptyset\).
  We also claim that \(\powerset_{\V}(X)\) has suprema for all (not necessarily
  directed) families \(\alpha : I \to \powerset_{\V}(X)\) with \(I : \V\).
  Given such a family \(\alpha\), its least upper bound is given by
  \(\bigcup \alpha \colonequiv \lambdadot{x}{\exists_{i : I}\,x\in\alpha_i}\),
  the set-theoretic union, which is well-defined as
  \(\pa*{\exists_{i : I}\,x\in\alpha_i} : \V\).
  Finally, paying attention to the universe levels we observe that
  \(\powerset_{\V}(X) : \DCPO{V}{V^+ \sqcup \U}{V \sqcup \U}\).
  In the case that \(X : \U \equiv \V\), we obtain the simpler, locally small
  \(\powerset_{\V}(X) : \DCPO{V}{V^+}{V}\).
\end{example}

Of course, \(\Omega_{\V}\) is easily seen to be equivalent to
\(\powerset_{\V}(\One_\V)\), so \cref{powersets-as-pointed-dcpos} subsumes
\cref{Omega-as-pointed-dcpo}, but it is instructive to understand
\cref{Omega-as-pointed-dcpo} first.

\subsection{Scott continuous maps}\label{sec:Scott-continuous-maps}

\begin{definition}[Scott continuity]
  A function \(f : D \to E\) between two \(\V\)\nobreakdash-dcpos is
  \emph{(Scott) continuous} if it preserves directed suprema, i.e.\ if
  \(I : \V \) and \(\alpha : I \to D\) is directed, then
  \(f\pa*{\bigsqcup \alpha}\) is the supremum in \(E\) of the family
  \(f \circ \alpha\).
\end{definition}

\begin{remark}
  When we speak of a Scott continuous function between \(D\) and~\(E\), then we
  will always assume that \(D\) and \(E\) are both \(\V\)-dcpos for the same
  arbitrary but fixed type universe \(\V\).
  Notice that Scott continuity is a property of a map and that any Scott
  continuous function is monotone.
\end{remark}

\begin{remark}
  In constructive mathematics it is not possible to exhibit a discontinuous
  function from \(\Nat^\Nat\) to \(\Nat\), because
  sheaf~\cite[Chapter~15]{TroelstraVanDalen1988} and realizability
  models~\cite[e.g.~Proposition~3.1.6]{vanOosten2008} imply that it is
  consistent to assume that all such functions are continuous.
  This does not mean, however, that we cannot exhibit a discontinuous function
  between dcpos. In fact, the negation map \({\lnot} : \Omega \to \Omega\) is
  not monotone and hence not continuous.
  If we were to preclude such examples, then we can no longer work with the full
  type \(\Omega\) of all propositions, but instead we must restrict to a subtype
  of propositions, for example by using dominances~\cite{Rosolini1986}.
  Indeed, this approach is investigated in the context of topos theory
  in~\cite{Phao1991,Longley1995} and for computability instead of continuity in
  univalent foundations in~\cite{EscardoKnapp2017}.
\end{remark}

\begin{definition}[Strictness]
  A Scott continuous function \(f : D \to E\) between pointed dcpos is \emph{strict}
  if \(f\pa*{\bot_{D}} = \bot_{E}\).
\end{definition}

\begin{lemma}\label{pointed-dcpos-sups}
  A poset \(D\) is a pointed \(\V\)-dcpo if and only if it has suprema for all
  semidirected families indexed by types in \(\V\) that we will denote using the
  \(\bigvee\) symbol.
  In~particular, a pointed \(\V\)-dcpo has suprema of all families indexed by
  propositions in \(\V\).

  Moreover, if \(f\) is a Scott continuous and strict map between pointed
  \(\V\)-dcpos, then \(f\) preserves suprema of semidirected families.
\end{lemma}

Classically, this can be proved by case distinction on whether the domain of a
semidirected family is inhabited, but here we avoid this as follows:

\begin{proof}
  If \(D\) is complete with respect to semidirected families indexed by types
  in~\(\V\), then it is clearly a \(\V\)-dcpo and it is pointed because the
  supremum of the family indexed by the empty type is the least element.
  Conversely, if \(D\) is a pointed \(\V\)-dcpo and \(\alpha : I \to D\) is a
  semidirected family with \(I : \V\), then
  \begin{align*}
    \hat\alpha : I + \One_{\V} &\to D \\
    \inl(i) &\mapsto \alpha_i \\
    \inr(\star) &\mapsto \bot
  \end{align*}
  is directed and hence has a sup in \(D\) which is also the least
  upper bound of~\(\alpha\).

  A pointed \(\V\)-dcpo must have suprema for all families indexed by
  propositions in~\(\V\), because any such family is semidirected.
  Finally, suppose that \(\alpha : I \to D\) is semidirected and that
  \(f : D \to E\) is Scott continuous and strict. Using the
  \(\widehat{(-)}\)-construction from above, as well as Scott continuity and
  strictness of~\(f\), we~get
  \[
    f\pa*{\textstyle\bigvee \alpha}
    \equiv f\pa*{\textstyle\bigsqcup\hat\alpha}
    = \textstyle\bigsqcup f \circ \hat\alpha
    = \textstyle\bigsqcup \widehat{f \circ \alpha}
    \equiv \textstyle\bigvee {f \circ \alpha},
  \]
  finishing the proof.
\end{proof}

\begin{definition}[Isomorphism]
  A Scott continuous map \(f : D \to E\) is an \emph{isomorphism} if we have a
  Scott continuous inverse \(g : E \to D\).
\end{definition}

\begin{definition}[Scott continuous retract]\label{def:continuous-retract}%
  A dcpo \(D\) is a \emph{(Scott continuous) retract} of \(E\) if we have Scott
  continuous maps \(s : D \to E\) and \(r : E \to D\) such that \(s\)~is a
  section of \(r\). We denote this situation by \(\retract{D}{E}\).
\end{definition}

\begin{lemma}\label{locally-small-retract}
  If \(D\) is a retract of \(E\) and \(E\) is locally small, then so is~\(D\).
\end{lemma}
\begin{proof}
  We claim that \(x \below_D y\) and \(s(x) \below_E s(y)\) are equivalent,
  which proves the lemma as \(E\) is assumed to be locally small.
  One direction of the equivalence is given by the fact that \(s\) is monotone.
  In the other direction, assume that \(s(x) \below s(y)\) and note that
  \(x = r(s(x)) \below r(s(y)) = y\), as \(r\) is monotone and \(s\) is a
  section of \(r\).
\end{proof}

\subsection{Lifting}\label{sec:lifting}

We now turn to constructing pointed \(\V\)-dcpos from sets.
First of all, every discretely ordered set is a \(\V\)-dcpo, where discretely
ordered means that we have \(x \below y\) exactly when \(x = y\).
In fact, ordering \(X\) discretely yields the free \(\V\)-dcpo on the set \(X\)
in the categorical sense.

With excluded middle, the situation for \emph{pointed} \(\V\)-dcpos is also
straightforward. Simply order the set \(X\) discretely and add a least
element.
However, in \cite[Lemma~17]{deJong2021a}, we showed, by considering
\(X \equiv \Nat\) and a reduction to the constructive taboo
LPO~\cite{Bishop1967}, that this approach is constructively unsatisfactory.
Moreover, in~\cite{deJongEscardo2023} we proved a general constructive no-go
theorem showing that there is a nontrivial dcpo with decidable equality if and
only if weak excluded middle holds.

Our solution to the above will be to work with the lifting monad, sometimes
known as the partial map classifier monad from topos
theory~\cite{Johnstone1977,Rosolini1986,Rosolini1987,Kock1991}, which has been
extended to constructive type theory by
\citeauthor{ReusStreicher1999}~\cite{ReusStreicher1999} and recently to
univalent foundations by
\citeauthor{EscardoKnapp2017}~\cite{EscardoKnapp2017,Knapp2018}.

\begin{definition}[Lifting, partial element, \(\lifting_{\V}(X)\); %
  {\cite[Section~2.2]{EscardoKnapp2017}}]
  We define the type of \emph{partial elements} of a type \(X : \U\) with
  respect to a universe \(\V\) as
  \[
    \lifting_{\V}(X) \colonequiv \Sigma_{P : \Omega_{\V}}(P \to X)
  \]
  and we also call it the \emph{lifting} of \(X\) with respect to \(\V\).
\end{definition}

Every (total) element of \(X\) gives rise to a partial element of \(X\) through
the following map:

\begin{definition}[\(\eta_X\)]
  The map \(\eta_X : X \to \lifting_{\V}(X)\) is defined by mapping \(x\) to the
  tuple \(\pa*{\One_{\V},\lambdadot{u}{x}}\), where we have omitted the witness
  that \(\One_{\V}\) is a subsingleton.
  We sometimes omit the subscript in \(\eta_X\).
\end{definition}

Besides these total elements, the lifting has another distinguished element that
will be the least in the order with which we shall equip the
lifting.

\begin{definition}[\(\bot\)]\label{def:lifting-bot}
  For every type \(X : \U\) and universe \(\V\), we denote the element
  \(\pa*{\Zero_{\V},\varphi} : \lifting_{\V}(X)\) by \(\bot\). (Here \(\varphi\)
  is the unique map from \(\Zero_{\V}\) to \(X\).)
\end{definition}

\begin{proposition}[{\cite[Lemma~18]{deJongEscardo2021a}}]\label{lifting-order}
  The \((\V^+\sqcup \U)\)-valued binary relation on \(\lifting_{\V}(X)\)
  given by
  \[
    (P,\varphi) \below (Q,\psi) \colonequiv {P \to (P,\varphi) = (Q,\psi)}
  \]
  is a partial order on \(\lifting_{\V}(X)\) for every set \(X : \U\).
  Moreover, it is equivalent to the relation
  \[
    (P,\varphi) \below' (Q,\psi) \colonequiv \Sigma_{f : P \to Q}%
    \pa*{\varphi \sim \psi \circ f}
  \]
  that is valued in \(\V \sqcup \U\).
\end{proposition}

In light of~\cref{universe-levels-of-lifting-and-exponentials}, we carefully
keep track of the universe parameters of the lifting in the following
proposition.

\begin{proposition}[cf.~{\cite[Theorem~1]{EscardoKnapp2017}}]%
  \label{lifting-is-pointed-dcpo}%
  For a set \(X : \U\), the lifting \(\lifting_{\V}(X)\) ordered as in
  \cref{lifting-order} is a pointed \(\V\)-dcpo.
  In full generality we have \(\lifting_{\V}(X) : \DCPO{V}{V^+ \sqcup U}{V^+ \sqcup U}\), but
  if \(X : \V\), then \(\lifting_{\V}(X)\) is locally small.
\end{proposition}

\subsection{Exponentials}\label{sec:exponentials}
Exponentials will be crucial in Scott's \(D_\infty\) construction
(\cref{sec:Scott-D-infty}).

\begin{definition}[Exponential of (pointed) dcpos, \(E^D\)]
  The \emph{exponential} of two \(\V\)\nobreakdash-dcpos \(D\) and \(E\) is
  given by the poset \(E^D\) defined as follows. Its carrier is the type of
  Scott continuous functions from \(D\) to \(E\).
  The order is given pointwise, i.e.\ \(f \below_{E^D} g\) holds if
  \(f(x) \below_{E} g(x)\) for every \(x : D\).
  Notice that if \(E\) is pointed, then so is \(E^D\) with least element given the constant function at the least element of \(E\).
  Finally, it is straightforward to show that \(E^D\) is \(\V\)-directed
  complete, so that \(E^D\) is another \(\V\)-dcpo.
\end{definition}

Note that the exponential \(E^D\) is a priori not locally small even if \(E\) is
because the partial order quantifies over all elements of \(D\). But if \(D\)
has a small basis then \(E^D\) will be locally small when \(E\) is
(\cref{exponential-is-locally-small}).

\begin{remark}\label{exponential-universe-parameters}%
  Recall from~\cref{universe-levels-of-lifting-and-exponentials} that it is
  necessary to carefully keep track of the universe parameters of the
  exponential.
  In general, the universe levels of \(E^{D}\) can be quite large and
  complicated. For~if \(D : \DCPO{V}{U}{T}\) and \(E : \DCPO{V}{U'}{T'}\), then
  the exponential \(E^D\) has a carrier in the universe
  \[
    \V^+ \sqcup \U \sqcup \T \sqcup \U' \sqcup \T'
  \]
  and an order relation that
  takes values in
  \(
    \U \sqcup \T'.
    \)

  Even~if
  \(\V = \U \equiv \T \equiv \U' \equiv \T'\), the carrier of \(E^{D}\) still
  lives in the larger universe~\(\V ^+\), because the type expressing Scott
  continuity for \(\V\)-dcpos quantifies over all types in~\(\V\).
  Actually, the scenario where \(\U = \U' = \V\) cannot happen in a predicative
  setting unless \(D\) and \(E\) are trivial, in a sense made precise in
  \cite{deJongEscardo2023}.

  Even so, in many applications such as those in \cite{deJong2021a}
  or~\cref{sec:Scott-D-infty}, if we take \(\V \equiv \U_0\) and all other
  parameters to be \(\U \equiv \T \equiv \U' \equiv \T' \equiv \U_1\), then the
  situation is much simpler and \(D\),~\(E\)~and the exponential \(E^D\) are all
  elements of \(\DCPOnum{0}{1}{1}\) with all of them being locally small
  (remember that this is defined up to equivalence).
  This turns out to be a very favourable situation for both the Scott model of
  PCF~\cite{deJong2021a} and Scott's \(D_\infty\) model of the untyped
  \(\lambda\)-calculus~(\cref{sec:Scott-D-infty}).
  In summary, if we take \(\V \equiv \U_0\) and all other parameters to be
  \(\U_1\), then the iterated exponentials all remain in \(\U_1\).
\end{remark}

After defining products of dcpos as usual (which we omit here for the sake of
brevity), we can state and prove a universe parametric version of the universal
property of exponentials.
In the proposition below we can have \(D : \DCPO{\V}{\U}{\T}\) and
\(E : \DCPO{\V}{\U'}{\T'}\) for arbitrary universes \(\U\), \(\T\), \(\U'\) and
\(\T'\). In particular, the universe parameters of \(D\) and \(E\), apart from
the universe of indexing types, need not be the same.

\begin{proposition}\label{exponential-universal-property}%
  The exponential defined above satisfies the appropriate universal property:
  the \emph{evaluation map} \(\ev: E^D \times D \to E, (g,x) \mapsto g(x)\) is
  Scott continuous and if \(f : {D'\times D} \to E\) is a Scott continuous
  function, then there is a unique Scott continuous map
  \(\bar{f} \colon D' \to E^D\) such that the diagram
  \[
    \begin{tikzcd}
      D' \times D \ar[dr,"f"] \ar[d,dashed,"{\bar{f}}\,\times\,{\id_D}"'] \\
      E^D \times D \ar[r,"\ev"'] & E
    \end{tikzcd}
  \]
  commutes.
\end{proposition}
\begin{proof}
  As in the classical case.
\end{proof}

\subsection{Bilimits}\label{sec:bilimits}
In \cref{sec:Scott-D-infty}, we give a predicative account of Scott's
\(D_\infty\) model of the untyped \(\lambda\)-calculus~\cite{Scott1972}.
Here, we describe the general machinery underlying Scott's construction.

A~priori one might expect that iterative constructions of dcpos, such as
Scott's~\(D_\infty\), may result in a need for ever-increasing universes and are
predicatively impossible. We show, through a careful tracking of type universe
parameters, that this is not the case.
Secondly, differences arise from proof relevance and these complications are
tackled with techniques in univalent foundations
and~\cref{constant-map-to-set-factors-through-truncation} in particular, as
discussed right before~\cref{kappa-is-constant}, for example.
Compared to Scott's original paper~\cite{Scott1972}, we also generalise from
sequential bilimits to directed bilimits.

\begin{definition}[Deflation]
  A continuous endofunction \(f : D \to D\) is a \emph{deflation} if
  \(f(x) \below x\) for all \(x : D\).
\end{definition}
\begin{definition}[Embedding-projection pair]\label{def:embedding-projection-pair}
  An \emph{embedding-projection pair} from a \(\V\)-dcpo \(D\) to a \(\V\)-dcpo
  \(E\) consists of two Scott continuous functions \(\varepsilon : D \to E\)
  (the~\emph{embedding}) and \(\pi : E \to D\) (the~\emph{projection}) such~that
  \(\varepsilon\) is a section of \(\pi\) and \(\varepsilon \circ \pi\) is a
  deflation.
\end{definition}

For the remainder of this section, fix the following setup, where we try to be
as general as possible regarding universe levels.
We fix a directed preorder \((I,\below)\) with \(I : \V\) and such that \(\below\) takes
values in some universe \(\W\). Now suppose that \((I,\below)\) indexes a family
of \(\V\)-dcpos with embedding-projection pairs between them, i.e.\ we have
\begin{itemize}
\item for every \(i : I\), a \(\V \)-dcpo \(D_i : \DCPO{V}{U}{T}\), and
\item for every \(i,j : I\) with \(i \sqsubseteq j\), an embedding-projection
  pair \(\pa*{\varepsilon_{i,j},\pi_{i,j}}\) from \(D_i\) to \(D_j\).
\end{itemize}
Moreover, we require that the following compatibility conditions hold:
\begin{align}
  &\text{for every \(i : I\), we have \(\varepsilon_{i,i} = \pi_{i,i} = \id\)}; %
  \label{epsilon-pi-id} \\
  &\text{for every \(i \sqsubseteq j \sqsubseteq k\) in \(I\), we have
  \(\varepsilon_{i,k} \sim \varepsilon_{j,k} \circ \varepsilon_{i,j}\) and
    \(\pi_{i,k} \sim \pi_{i,j} \circ \pi_{j,k}\).} %
  \label{epsilon-pi-comms}
\end{align}

The goal is now to construct another \(\V\)-dcpo \(D_\infty\) with
embedding-projections pairs
\(\pa*{\varepsilon_{i,\infty} : D_1 \hookrightarrow D_\infty, {\pi_{i,\infty} :
    D_\infty \to D_i}}\) for every \(i : I\), such that
\(\pa*{D_\infty,\pa*{\varepsilon_{i,\infty}}_{i : I}}\) is the colimit of the
diagram given by \(\pa*{\varepsilon_{i,j}}_{i \below j \text{ in } I}\) and
\(\pa*{D_\infty,\pa*{\pi_{i,\infty}}_{i : I}}\) is the limit of the
diagram given by \(\pa*{\pi_{i,j}}_{i \below j \text{ in } I}\).
In other words,
\(\pa*{D_\infty,\pa*{\varepsilon_{i,\infty}}_{i : I},\pa*{\pi_{i,\infty}}_{i :
    I}}\) is both the colimit and the limit in the category of \(\V\)-dcpos with
embedding-projections pairs between them. We say that it is the \emph{bilimit}.

\begin{definition}[\(D_\infty\)]\label{def:D-infty}%
  We define a poset \(D_\infty\) as follows. Its carrier is given by the type of
  elements \(\sigma\) of the product \(\Pi_{i : I}D_i\) satisfying
  \(\pi_{i,j}(\sigma_j) = \sigma_i\) whenever \(i \below j\).
  That is, the carrier is the type
  \[
    \sum_{\sigma : \Pi_{i : I} D_i} \prod_{{i,j} : I , i \below j}
    \pi_{i,j}\pa*{\sigma_j} = \sigma_i.
  \]
  Note that this defines a subtype of \(\Pi_{i : I}D_i\) as the condition
  \(\prod_{{i,j} : I , i \below j} \pi_{i,j}\pa*{\sigma_j} = \sigma_i\) is a
  property by \cite[Example~3.6.2]{HoTTBook} and the fact that each \(D_i\) is a
  set.
  These functions are ordered pointwise, i.e.\ if
  \(\sigma,\tau : \Pi_{i : I} D_i\), then \(\sigma \below_{D_\infty} \tau\)
  exactly when \(\sigma_i \below_{D_i} \tau_i\) for every \(i : I\).
\end{definition}

The proof of the following is as in the classical case, but we pay attention to
the universe levels.

\begin{lemma}
  The poset \(D_\infty\) is \(\V\)-directed complete with suprema calculated
  pointwise, and we have  \(D_\infty : \DCPO{V}{U \sqcup V \sqcup W}{U \sqcup T}\).
\end{lemma}

\begin{remark}\label{bilimit-universe-parameters}%
  We allow for general universe levels here, which is why \(D_\infty\) lives in
  the relatively complicated universe \(\U \sqcup \V \sqcup \W\). In concrete
  examples, the situation often simplifies. E.g., in \cref{sec:Scott-D-infty} we
  find ourselves in the favourable situation described
  in~\cref{exponential-universe-parameters} where \(\V \equiv \W \equiv \U_0\)
  and \(\U \equiv \T \equiv \U_1\), so that we get
  \(D_\infty : \DCPOnum{0}{1}{1}\), as the bilimit of a diagram of dcpos
  \(D_n : \DCPOnum{0}{1}{1}\) indexed by natural numbers.
\end{remark}

\begin{definition}[\(\pi_{i,\infty}\)]\label{pi-infty}
  For every \(i : I\), we define the Scott continuous function
  \(\pi_{i,\infty} : {D_\infty \to D_i}\) by \(\sigma \mapsto \sigma_i\).
\end{definition}

While we could closely follow~\cite{Scott1972} up until this point, we will now
need a new idea to proceed.
Our goal is to define maps \(\varepsilon_{i,\infty} : D_i \to D_\infty\) for
every \(i : I\) so that \(\varepsilon_{i,\infty}\) and \(\pi_{i,\infty}\) form
an embedding-projection pair.
We give an outline of the idea for defining this map
\(\varepsilon_{i,\infty}\). For~an arbitrary element \(x : D_i\), we need to
construct \(\sigma : D_\infty\) at component \(j : I\), say. If we had \(k : I\)
such that \(i,j \below k\), then we could define \(\sigma_j : D_j\) by
\(\pi_{j,k}\pa*{\varepsilon_{i,k}(x)}\).
Now semidirectedness of \(I\) tells us that there exists such a \(k : I\), so
the point is to somehow make use of this propositionally truncated fact. This is
where~\cref{constant-map-to-set-factors-through-truncation} comes in. We
recall that it says that a constant map to a set factors through the propositional
truncation of its domain.
We define a map
\(\kappa_{i,j}^x : \pa*{\Sigma_{k : I}\,\pa{i \below k} \times \pa{j \below k}}
\to D_j\) by sending \(k\) to \(\pi_{j,k}\pa*{\varepsilon_{i,k}(x)}\) and show
it to be constant, so that it factors through the truncation of its domain.
In the special case that \(I \equiv \Nat\), as in~\cite{Scott1972}, we could
simply take \(k\) to be the sum of the natural numbers \(i\) and \(j\), but this
does not work in the more general directed case, of course.

\begin{definition}[\(\kappa_{i,j}^x\)]\label{def:kappa}
  For every \(i,j : I\) and \(x : D_i\) we define the function
  \[
    \kappa_{i,j}^x : \pa*{\Sigma_{k : I}\,\pa{i \below k} \times \pa{j \below
        k}} \to D_j
  \]
  by mapping \(k : I\) with \(i,j \below k\) to
  \(\pi_{j,k}\pa*{\varepsilon_{i,k}(x)}\).
\end{definition}

\begin{lemma}\label{kappa-is-constant}
  The function \(\kappa_{i,j}^x\) is constant for ever \(i,j : I\) and
  \(x : D_i\).
  Hence, \(\kappa_{i,j}^x\) factors through the type
  \(\exists_{k : I}\,\pa{i \below k}\times\pa{j \below k}\)
  by~\cref{constant-map-to-set-factors-through-truncation}.
\end{lemma}
\begin{proof}
  If we have \(k_1,k_2 : I\) with \(i \below k_1,k_2\) and \(j \below k_1,k_2\),
  then by semidirectedness of \(I\), there exists some \(k : K\) with
  \(k_1,k_2 \below k\) and hence,
  \begin{align*}
    &\pa*{\pi_{j,k_1} \circ \varepsilon_{i,k_1}} (x) \\
    &= \pa*{\pi_{j,k_1} \circ \pi_{k_1,k} \circ \varepsilon_{k_1,k} \circ \varepsilon_{i,k_1}} (x)
      &&\text{(as \(\varepsilon_{k_1,k}\) is a section of \(\pi_{k_1,k}\))}
    \\
    &= \pa*{\pi_{j,k} \circ \varepsilon_{i,k}}(x)
      &&\text{(by \cref{epsilon-pi-comms})}
    \\
    &= \pa*{\pi_{j,k} \circ \pi_{k_2,k} \circ \varepsilon_{k_2,k} \circ \varepsilon_{i,k_2}} (x)
      &&\text{(as \(\varepsilon_{k_2,k}\) is a section of \(\pi_{k_2,k}\))}
    \\
    &= \pa*{\pi_{j,k_2} \circ \varepsilon_{i,k_2}}(x)
      &&\text{(by \cref{epsilon-pi-comms})},
  \end{align*}
  proving that \(\kappa_{i,j}^x\) is constant.
\end{proof}

\begin{definition}[\(\rho_{i,j}\)]
  For every \(i,j : I\), the type
  \(\exists_{k : I}\,\pa{i \below k} \times \pa{j \below k}\) has an element
  since \((I,\below)\) is directed. Thus, \cref{kappa-is-constant} tells us that
  we have a function \(\rho_{i,j} : D_i \to D_j\) such that if \(i,j \below k\),
  then the equation
  \begin{equation}\label{rho-eq}
    \rho_{i,j}(x) = \kappa_{i,j}^x(k) \equiv \pi_{j,k}\pa*{\varepsilon_{i,k}(x)}
  \end{equation}
  holds for every \(x : D_i\).
\end{definition}

\begin{definition}[\(\varepsilon_{i,\infty}\)]\label{epsilon-infty}
  The map \(\rho\) induces a map \(\varepsilon_{i,\infty} : D_i \to D_\infty\)
  by sending \(x : D_i\) to the function \(\lambdadot{j : I}{\rho_{i,j}(x)}\).
  To see that this is well-defined, assume that we have \(j_1 \below j_2\) in
  \(J\) and \(x : D_i\). We have to show that
  \(\pi_{j_1,j_2}\pa*{\pa*{\varepsilon_{i,\infty}(x)}_{j_2}} =
  \pa*{\varepsilon_{i,\infty}(x)}_{j_1}\).
  By semidirectedness of \(I\) and the fact that are looking to prove a
  proposition, we may assume to have \(k : I\) with \(i \below k\) and
  \(j_1 \below j_2 \below k\). Then,
  \begin{align*}
    \pi_{j_1,j_2}\pa*{\pa*{\varepsilon_{i,\infty}(x)}_{j_2}}
    &\equiv \pi_{j_1,j_2}\pa*{\rho_{i,j_2}(x)} \\
    &= \pi_{j_1,j_2}\pa*{\pi_{j_2,k}\pa*{\varepsilon_{i,k}(x)}}
    &&\text{(by \cref{rho-eq})} \\
    &= \pi_{j_1,k}\pa*{\varepsilon_{i,k}(x)}
    &&\text{(by \cref{epsilon-pi-comms})}
    \\
    &= \rho_{i,j_1}(x)
    &&\text{(by \cref{rho-eq})} \\
    &\equiv \pa*{\varepsilon_{i,\infty}(x)}_{j_1},
  \end{align*}
  as desired.
\end{definition}

This completes the definition of \(\varepsilon_{i,\infty}\). From this point on,
we can typically work with it using~\cref{rho-eq} and the fact that
\(\pa*{\varepsilon_{i,\infty}(x)}_j\) is defined as \(\rho_{i,j}(x)\).
Adapting~\cite{Scott1972} to the directed case, we can then prove the following
theorems.

\begin{theorem}\label{epsilon-pi-infty-ep-pair}
  For every \(i : I\), the pair
  \(\pa*{\varepsilon_{i,\infty},\pi_{i,\infty}}\) is an embedding-projection
  pair from \(D_i\) to \(D_\infty\).
\end{theorem}

\begin{theorem}\label{limit}%
  The \(\V\)-dcpo \(D_\infty\) with the maps \(\pa*{\pi_{i,\infty}}_{i : I}\) is
  the limit of the diagram
  \(\pa*{\pa*{D_i}_{i : I} , \pa*{\pi_{i,j}}_{i \below j}}\).
  That is, given a \(\V \)-dcpo \(E : \DCPO{V}{U'}{T'}\) and Scott continuous
  functions \(f_i : E \to D_i\) for every \(i : I\) such that the diagram
  \begin{equation*}\label{fs-are-cone}
    \begin{tikzcd}
      E \ar[dr,"f_j"'] \ar[rr,"f_i"] & & D_i \\
      & D_j \ar[ur,"\pi_{i,j}"']
    \end{tikzcd}
  \end{equation*}
  commutes for every \(i \below j\),
  we have a unique continuous function \({f_\infty : E \to D_\infty}\) making
  the diagram
  \begin{equation*}\label{f-infty}
    \begin{tikzcd}
      E \ar[dr,dashed,"f_\infty"'] \ar[rr,"f_i"] & & D_i \\
      & D_\infty \ar[ur,"\pi_{i,\infty}"']
    \end{tikzcd}
  \end{equation*}
  commute for every \(i : I\).

  Similarly, the \(\V\)-dcpo \(D_\infty\) with the maps
  \(\pa*{\varepsilon_{i,\infty}}_{i : I}\) is the colimit of the diagram
  \(\pa*{\pa*{D_i}_{i : I} , \pa*{\varepsilon_{i,j}}_{i \below j}}\).
\end{theorem}

It should be noted that in the above universal property, \(E\) can have its
carrier in any universe \(\U'\) and its order taking values in any universe
\(\T'\), even though we required all \(D_i\) to have their carriers and orders
in two fixed universes \(\U \)~and~\(\T\), respectively.

The proof of the colimit property relies on the following lemma which is also useful later on.
\begin{lemma}\label{sigma-sup-of-epsilon-pis}
  Every element \(\sigma : D_\infty\) is the directed supremum of\/
  \(\bigsqcup_{i : I} \varepsilon_{i,\infty}\pa*{\sigma_i}\).
\end{lemma}
\begin{proposition}\label{locally-small-bilimit}
  The bilimit of locally small dcpos is locally small, i.e.\ if every \(\V\)-dcpo
  \(D_i\) is locally small for all \(i : I\), then so is \(D_\infty\).
\end{proposition}
\begin{proof}
  If every \(D_i\) is locally small, then for every \(i : I\), we have a
  \emph{specified} \(\V\)\nobreakdash-valued partial order \(\below_{\V}^i\) on
  \(D_i\) such that for every \(i : I\) and every \(x,y : D_i\), we have an
  equivalence \(\pa{x \below_{D_i} y} \simeq \pa{x \below_{\V}^i y}\).
  Hence,
  \({\pa{\sigma \below_{D_\infty} \tau} \equiv \pa{\Pi_{i : I}\pa{\sigma_i
      \below_{D_i} \tau_i}}} \simeq \pa{\Pi_{i : I}\pa{\sigma_i \below_{\V}^i
      \tau_i}}\), but the latter is small, because \(I : \V\) and
  \(\below_{\V}^i\) is \(\V\)-valued.
\end{proof}

\subsection{Scott's \texorpdfstring{\(D_\infty\)}{D-infinity} model of the untyped
  \texorpdfstring{\(\lambda\)}{lambda}-calculus}\label{sec:Scott-D-infty}
We are finally in a position to construct Scott's \(D_\infty\)~\cite{Scott1972} predicatively. Formulated
precisely, we construct a pointed \(D_\infty : \DCPOnum{0}{1}{1}\) such that
\(D_\infty\) is isomorphic to its self-exponential~\(D_\infty^{D_\infty}\),
employing the machinery from~\cref{sec:bilimits}.

\begin{definition}[\(D_n\)]
  We inductively define pointed dcpos \(D_n : \DCPOnum{0}{1}{1}\) for every
  natural number \(n\) by setting
  \(D_0 \colonequiv \lifting_{\U_0}\pa*{\One_{\U_0}}\) and
  \(D_{n+1} \colonequiv D_n^{D_n}\).
\end{definition}

In light of~\cref{universe-levels-of-lifting-and-exponentials} we highlight the
fact that every \(D_n\) is a \(\U_0\)-dcpo with carrier in \(\U_1\) by the
discussion of universe parameters of exponentials
in~\cref{exponential-universe-parameters}.%

\begin{definition}[\(\varepsilon_n\), \(\pi_n\)]
  We inductively define for every natural number \(n\), two Scott continuous maps
  \(\varepsilon_n : D_n \to D_{n+1}\) and \({\pi_n : D_{n+1} \to D_n}\):
  \begin{enumerate}[(i)]
  \item
    \begin{itemize}
    \item \(\varepsilon_0 : D_0 \to D_1\) is given by mapping \(x : D_0\) to the
      continuous function that is constantly~\(x\),
    \item \(\pi_0 : D_1 \to D_0\) is given by evaluating a continuous function
      \({f : D_0 \to D_0}\) at~\(\bot\) which is itself continuous
      by~\cref{exponential-universal-property},
    \end{itemize}
  \item
    \begin{itemize}
    \item \(\varepsilon_{n+1} : D_{n+1} \to D_{n+2}\) takes a continuous
      function \(f : D_n \to D_n\) to the continuous composite
      \(D_{n+1} \xrightarrow{\pi_n} D_n \xrightarrow{f} D_n
      \xrightarrow{\varepsilon_n} D_{n+1}\), and
    \item \(\pi_{n+1} : D_{n+2} \to D_{n+1}\) takes a continuous function
      \(f : D_{n+1} \to D_{n+1}\) to the continuous composite
      \(D_n \xrightarrow{\varepsilon_n} D_{n+1} \xrightarrow{f} D_{n+1}
      \xrightarrow{\pi_n} D_n\). \qedhere
    \end{itemize}
  \end{enumerate}
\end{definition}

The maps \(\varepsilon_n\) and \(\pi_n\) form an embedding-projection pair for
each natural number \(n\), and, by taking compositions, we obtain
embedding-projection pairs \(\pa*{\varepsilon_{n,m},\pi_{n,m}}\) from \(D_n\) to
\(D_m\) whenever \(n \leq m\).

\begin{definition}[\(D_\infty\)]
  Applying \cref{def:D-infty,pi-infty,epsilon-infty} to the above diagram yields
  \(D_\infty : \DCPOnum{0}{1}{1}\) with embedding-projection pairs
  \(\pa*{\varepsilon_{n,\infty},\pi_{n,\infty}}\) from \(D_n\) to \(D_\infty\)
  for every natural number \(n\).
\end{definition}

Following~\cite{Scott1972}, we can show the following.
\begin{theorem}
  The pointed \(\U_0\)-dcpos \(D_\infty\) and \(D_\infty^{D_\infty}\) are isomorphic.
\end{theorem}

Moreover, by (for instance) embedding \(\eta(\star) : D_0\) into \(D_\infty\),
we see that \(D_\infty\) is not the trivial pointed dcpo.

\section{The way-below relation and compactness}\label{sec:way-below}

The way-below relation is the fundamental ingredient in the development of
continuous dcpos. Following Scott~\cite{Scott1970}, a computational intuition of
\(x \ll y\) says that every computation of \(y\) has to produce \(x\), or
something better than \(x\), at some stage.
\begin{definition}[Way-below relation, \(x \ll y\)]\label{def:way-below}%
  An element \(x\) of a \(\V\)-dcpo \(D\) is \emph{way below} an element \(y\)
  of \(D\) if whenever we have a directed family \(\alpha : I \to D\)
  indexed by \(I : \V\) such that \(y \below \bigsqcup \alpha\), then there
  exists \(i : I\) such that \(x \below \alpha_i\) already.
  We denote this situation by \(x \ll y\).
\end{definition}

\begin{lemma}\label{way-below-properties}
  The way-below relation enjoys the following properties.
  \begin{enumerate}[(i)]
  \item\label{way-below-prop-valued} it is proposition-valued;
  \item\label{below-if-way-below} if \(x \ll y\), then \(x \below y\);
  \item\label{below-way-below-way-below} if \(x \below y \ll v \below w\), then
    \(x \ll w\);
  \item\label{way-below-antisymmetric} it is antisymmetric;
  \item\label{way-below-transitive} it is transitive.
  \end{enumerate}
\end{lemma}
\begin{proof}
  \eqref{way-below-prop-valued} Using that a dependent product of propositions
  (over an arbitrary type) is again a proposition together with the fact that we
  propositionally truncated the existence of \(i : I\) in the definition.
  \eqref{below-if-way-below} Simply take \(\alpha : \One_{\V} \to D\) to be
  \(u \mapsto y\).
  \eqref{below-way-below-way-below} Suppose that \(\alpha : I \to D\) is directed
  with \(w \below \bigsqcup \alpha\). Then \(v \below \bigsqcup \alpha\), so by
  assumption that \(y \ll v\) there exists \(i : I\) with \(y \below \alpha_i\)
  already. But then \(x \below \alpha_i\).
  \eqref{way-below-antisymmetric} Follows from \eqref{below-if-way-below}.
  \eqref{way-below-transitive} Follows from \eqref{below-if-way-below} and
  \eqref{below-way-below-way-below}.
\end{proof}

In general, the way below relation is not reflexive. The elements for which it
is have a special status and are called compact. We illustrate this notion by a
series of examples.

\begin{definition}[Compactness]%
  An element of a dcpo is \emph{compact} if it is way below itself.
\end{definition}

\begin{example}\label{least-element-is-compact}%
  The least element of a pointed dcpo is always compact.
\end{example}

We recall the \(\V\)-dcpo of propositions \(\Omega_{\V}\) with designated
elements \(\Zero_{\V}\) (the empty type) and \(\One_{\V}\) (the unit type) from
\cref{Omega-as-pointed-dcpo}.

\begin{example}[Compact elements in \(\Omega_{\V}\)]%
  \label{compact-elements-in-Omega}%
  The compact elements of \(\Omega_{\V}\) are exactly \(\Zero_{\V}\) and
  \(\One_{\V}\). In other words, the compact elements of \(\Omega_{\V}\) are
  precisely the decidable propositions.
\end{example}
\begin{proof}
  By \cref{least-element-is-compact} we know that \(\Zero_{\V}\) must be
  compact. For \(\One_{\V}\), suppose that we have
  \(Q_{(-)} : I \to \Omega_{\V}\) directed such that
  \(\One_{\V} \below \exists_{i : I}Q_i\). Then there exists \(i : I\) such that
  \(Q_i\) holds, and hence, \(\One_{\V} \below Q_i\).
  Now suppose that \(P : \Omega_{\V}\) is compact. We show that \(P\) is
  decidable. The family \(\alpha : \pa*{P + \One_{\V}} \to \Omega_{\V}\) given
  by \(\inl(p) \mapsto \One_{\V}\) and \(\inr(\star) \mapsto \Zero_{\V}\) is
  directed and \(P \below \bigsqcup \alpha\). Hence, by compactness, there
  exists \(i : P + \One_{\V}\) such that \(P \below \alpha_i\) already. %
  Since being decidable is a property of a proposition, we actually get such an
  \(i\) and by case distinction on it we get decidability of \(P\).
\end{proof}

We recall the lifting \(\lifting_{\V}(X)\) of a set \(X\) as a \(\V\)-dcpo from \cref{sec:lifting}.
\begin{example}[Compact elements in the lifting]
  \label{compact-elements-in-lifting}%
  An element \((P,\varphi)\) of the lifting \(\lifting_{\V}(X)\) of a set
  \(X : \V\) is compact if and only if \(P\) is decidable.
  Hence, the compact elements of \(\lifting_{\V}(X)\) are exactly \(\bot\) and
  \(\eta(x)\) for \(x : X\).
\end{example}
\begin{proof}
  To see that compactness implies decidability of the domain of the partial
  element, we proceed as in the proof of \cref{compact-elements-in-Omega}, but
  for a partial element \((P,\varphi)\), we consider the family
  \(\alpha : \pa*{P + \One_{\V}} \to \lifting_{\V }(X)\) given by
  \(\inl(p) \mapsto \eta(\varphi(p))\) and \(\inr(\star) \mapsto \bot\).
  Conversely, if we have a partial element \((P,\varphi)\) with \(P\) decidable,
  then either \(P\) is false in which case \((P,\varphi) = \bot\) which is
  compact by \cref{least-element-is-compact}, or \(P\) holds. %
  So suppose that \(P\) holds and let \(\alpha : I \to \lifting_{\V}(X)\) be
  directed with \(P \below \bigsqcup \alpha\). %
  Since \(P\) holds, the element \(\bigsqcup \alpha\) must be defined, so there
  exists \(i : I\) such that \(\alpha_i\) is defined. %
  But for this \(i : I\) we also have \(\bigsqcup \alpha = \alpha_i\) by
  construction of the supremum, and hence, \(P \below \alpha_i\), proving
  compactness of \((P,\varphi)\).
\end{proof}

For characterising the compact elements of the powerset, we introduce a lemma,
as well as the notion of Kuratowski finiteness and the induction principle for
Kuratowski finite subsets.

\begin{lemma}\label{binary-join-is-compact}%
  The compact elements of a dcpo are closed under (existing) binary joins.
\end{lemma}
\begin{proof}
  Suppose that \(x\) and \(y\) are compact elements of a \(\V\)-dcpo \(D\) with
  \(z\) as their least upper bound and suppose that we have \(\alpha : I \to D\)
  directed with \(z \below \bigsqcup \alpha\). Then
  \(x \below \bigsqcup \alpha\) and \(y \below \bigsqcup \alpha\), so by
  compactness there exist \(i,j : I\) such that \(x \below \alpha_i\) and
  \(y \below \alpha_j\). By semidirectedness of \(\alpha\), there exists
  \(k : I\) with \({\alpha_i,\alpha_j} \below \alpha_k\), so that
  \({x,y} \below \alpha_k\). But \(z\) is the join of \(x\) and \(y\), so
  \(z \below \alpha_k\), as desired.
\end{proof}

Kuratowski finiteness is investigated
in~\cite{Kuratowski1920,Johnstone2002,CoquandSpiwack2010,FruminEtAl2018} among
other places.
\pagebreak[3]\begin{definition}[Kuratowski finiteness]\hfill%
  \begin{enumerate}[(i)]
  \item A type \(X\) is \emph{Kuratowski finite} if there exists some natural
    number \(n : \Nat\) and a surjection \(e : \Fin(n) \surj X\), where
    \(\Fin(n)\) is the standard finite type with exactly \(n\)
    elements~\cite[Section~7.3]{Rijke2022}.%
  \item A subset is \emph{Kuratowski finite} if its total space
    (recall~\cref{def:total-space}) is a Kuratowski finite type.%
    \qedhere
  \end{enumerate}
\end{definition}

Thus, a type \(X\) is \emph{Kuratowski finite} if its elements can be finitely
enumerated, possibly with repetitions, although the repetitions can be removed
when \(X\) has decidable equality.

\begin{lemma}\label{Kuratowski-finite-closure-properties}
  The Kuratowski finite subsets of a set are closed under finite unions and
  contain all singletons.
\end{lemma}
\begin{proof}
  The empty set and any singleton are clearly Kuratowski finite. Moreover, if
  \(A\) and \(B\) are Kuratowski finite subsets, then we may assume to have
  natural numbers \(n\) and \(m\) and surjections
  \(\sigma : \Fin(n) \surj \totalspace(A)\) and
  \(\tau : \Fin(m) \surj \totalspace(B)\). We can then patch these together to
  obtain a surjection \(\Fin(n + m) \surj \totalspace(A \cup B)\), as desired.
\end{proof}

The following induction principle appears as \cite[Definition~5.4.1]{Johnstone2002}
and is closely related to the higher inductive presentation
in~\cite{FruminEtAl2018}.
\begin{lemma}[Induction for Kuratowski finite subsets]%
  \label{Kuratowski-finite-subsets-induction}%
  A property of subsets of a type~\(X\) holds for all Kuratowski finite subsets
  of \(X\) as soon as
  \begin{enumerate}[(i)]
  \item\label{empty-set-case} it holds for the empty set,
  \item\label{singleton-case} it holds for any singleton subset, and
  \item\label{binary-union-case} it holds for \(A \cup B\), whenever it holds
    for \(A\) and \(B\).
  \end{enumerate}
\end{lemma}
\begin{proof}
  Let \(Q\) be such a property and let \(A\) be an arbitrary Kuratowski finite
  subset of \(X\). Since \(Q\) is proposition-valued, we may assume to have a
  natural number \(n\) and a surjection
  \(\sigma : \Fin(n) \surj \totalspace(A)\). Then the subset \(A\) must be equal
  to the finite join of singletons
  \(\{\sigma_0\} \cup \{\sigma_1\} \cup \dots \cup \{\sigma_{n-1}\}\), which can
  be shown to satisfy \(Q\) by induction on~\(n\), and hence, so must \(A\).
\end{proof}

\begin{definition}[\(\beta\)]\label{def:list-to-powerset}%
  For a set \(X : \U\), we write \(\beta : \List(X) \to \powerset_{\U}(X)\) for
  the canonical map which takes a list to its set of elements.
  (For the inductive definition of the type \(\List(X)\), recall
  e.g.~\cite[Section~5.1]{HoTTBook}.)
\end{definition}

\begin{lemma}\label{Kuratowski-finite-iff-list}
  A subset \(A : \powerset_{\U}(X)\) of a set \(X : \U\) is Kuratowski finite if
  and only if it is in the image of \(\beta\).
\end{lemma}
\begin{proof}
  The left to right direction follows from
  \cref{Kuratowski-finite-closure-properties}, while the converse follows easily
  from the induction principle for Kuratowski finite subsets where we use list
  concatenation in case~\eqref{binary-union-case}.
\end{proof}

\begin{example}[Compact elements in \(\powerset_{\U }(X)\)]%
  \label{compact-elements-in-powerset}%
  The compact elements of \(\powerset_{\U}(X)\) for a set \(X : \U\) are exactly
  the Kuratowski finite subsets of \(X\).
\end{example}
\begin{proof}
  Suppose first that \(A : \powerset_{\U}(X)\) is a compact element. The family
  \[
    \pa*{\Sigma_{l : \List(X)}\,\beta(l) \subseteq A}
    \xrightarrow{\beta \circ \fst}
    \powerset_{\U}(X)
  \]
  is directed, as it contains \(\emptyset\) and we can concatenate lists to
  establish semidirectedness. Moreover,
  \(\pa*{\Sigma_{l : \List(X)}\,\beta(l) \subseteq A}\) lives in \(\U\) and we
  clearly have \(A \subseteq \bigsqcup {\beta \circ \fst}\). So by compactness,
  there exists \(l : \List(X)\) with \(\beta(l) \subseteq A\) such that
  \(A \subseteq \beta(l)\) already. But this says exactly that \(A\) is
  Kuratowski finite by \cref{Kuratowski-finite-iff-list}.

  For the converse we use the induction principle for Kuratowski finite subsets:
  the empty set is compact by \cref{least-element-is-compact}, singletons are
  easily shown to be compact, and binary unions are compact by
  \cref{binary-join-is-compact}.
\end{proof}

We end this section by presenting a few lemmas connecting the way-below relation
and compactness to Scott continuous sections (\cref{def:continuous-retract}).

\begin{lemma}\label{continuous-retraction-way-below-criterion}%
  If we have a Scott continuous retraction \(\retract{D}{E}\), then \(y \ll s(x)\)
  implies \(r(y) \ll x\) for every \(x : D\) and \(y : E\).
\end{lemma}
\begin{proof}
  Suppose that \(y \ll s(x)\) and that \(x \below \bigsqcup \alpha\) for a
  directed family \(\alpha : I \to D\). Then
  \(s(x) \below s\pa*{\bigsqcup \alpha} = \bigsqcup \pa*{s \circ \alpha}\) by Scott
  continuity of \(s\), so there exists \(i : I\) such that
  \(y \below s(\alpha_i)\) already. Now monotonicity of \(r\) implies
  \(r(y) \below r(s(\alpha_i)) = \alpha_i\) which completes the proof that
  \(r(y) \ll x\).
\end{proof}

We also recall embedding-projection pairs from \cref{def:embedding-projection-pair}.

\begin{lemma}\label{embedding-preserves-and-reflects-way-below}%
  The embedding in an embedding-projection pair
  \(\retractalt{D}{E}{\varepsilon}{\pi}\) preserves and reflects the way-below
  relation, i.e.\ \(x \ll y \iff \varepsilon(x) \ll \varepsilon(y)\).
  In particular, an element \(x\) is compact if and only if \(\varepsilon(x)\)
  is.
\end{lemma}
\begin{proof}
  Suppose that \(x \ll y\) in \(D\) and let \(\alpha : I \to E\) be directed
  with \(\varepsilon(y) \below \bigsqcup \alpha\). Then
  \(y = \pi(\varepsilon(y)) \below \bigsqcup \pi \circ \alpha\) by Scott
  continuity of \(\pi\). Hence, there exists \(i : I\) such that
  \(x \below \pi(\alpha_i)\). But then
  \(\varepsilon(x) \below \varepsilon(\pi(\alpha_i)) \below \alpha_i\) by
  monotonicity of \(\varepsilon\) and the fact that \(\varepsilon \circ \pi\) is a
  deflation. This proves that \(x \ll y\).
  Conversely, if \(\varepsilon(x) \ll \varepsilon(y)\), then
  \(x = \pi(\varepsilon(x)) \ll y\) by
  \cref{continuous-retraction-way-below-criterion}.
\end{proof}

\section{The ind-completion}\label{sec:ind-completion}

The ind-completion will be a useful tool for phrasing and proving results about
directed complete \emph{posets} and is itself a directed complete
\emph{preorder}, cf.~\cref{ind-completion-is-directed-complete}.
It was introduced by \citeauthor{SGA41} in \cite[Section~8]{SGA41} in the
context of category theory, but its role in order theory is discussed in
\cite[Section~1]{JohnstoneJoyal1982}.
We will also use it in the context of order theory, but our treatment will
involve a careful consideration of the universes involved, very similar to the
original treatment in~\cite{SGA41}.

\begin{definition}[\(\V\)-ind-completion \(\Ind{V}(X)\), exceed, \({\cof}\)]%
  \label{def:exceeds}
  The \emph{\(\V\)-ind-completion} \(\Ind{V}(X)\) of a preorder \(X\) is the
  type of directed families in \(X\) indexed by types in the universe~\(\V\).
  Such a family \(\beta : J \to X\) \emph{exceeds} another family
  \(\alpha : I \to X\) if for every \(i : I\), there exists \(j : J\) such that
  \(\alpha_i \below \beta_j\), and we denote this relation by \(\alpha \cof \beta\).
\end{definition}

\begin{lemma}\label{ind-completion-is-directed-complete}\leavevmode
  \begin{enumerate}[(i)]
  \item The relation \(\cof\) defines a preorder on the ind-completion.
  \item The \(\V\)-ind-completion \(\Ind{V}(X)\) of a preorder \(X\) is
    \(\V\)-directed complete.
  \end{enumerate}
\end{lemma}
\begin{proof}
  The first item is proved straightforwardly.
  For the second, suppose that we
  have a directed family \(\alpha : I \to \Ind{V}(X)\) with \(I : \V\). Then
  each \(\alpha_i\) is a directed family in \(X\) indexed by a type
  \(J_i : \V\). We define the family
  \(\hat\alpha : \pa*{\Sigma_{i : I}J_i} \to X\) by
  \((i,j) \mapsto \alpha_i(j)\). It is clear that \(\hat\alpha\) exceeds each
  \(\alpha_i\), and that \(\beta\) exceeds \(\hat\alpha\) if \(\beta\) exceeds
  every \(\alpha_i\). So it remains to show that \(\hat\alpha\) is in fact an
  element of \(\Ind{V}(X)\), i.e.\ that it is directed. Because \(\alpha\) and
  each \(\alpha_i\) are directed, \(I\) and each \(J_i\) are inhabited. Hence,
  so is the domain of \(\hat\alpha\).
  It remains to show that \(\hat\alpha\) is semidirected. Suppose we have
  \((i_1,j_1) , (i_2,j_2)\) in the domain of \(\hat\alpha\). By directedness of
  \(\alpha\), there exists \(i : I\) such that \(\alpha_i\) exceeds both
  \(\alpha_{i_1}\) and \(\alpha_{i_2}\). Hence, there exist \(j_1',j_2' : J_i\)
  with \(\alpha_{i_1}(j_1) \below \alpha_{i}(j_1')\) and
  \(\alpha_{i_2}(j_2) \below \alpha_{i}(j_2')\).
  Because the family \(\alpha_i\) is directed in \(X\), there exists \(j : J_i\)
  such that \(\alpha_{i}(j_1'),\alpha_{i}(j_2') \below \alpha_i(j)\).  Hence, we
  conclude that
  \(\hat\alpha(i_1,j_1) \equiv \alpha_{i_1}(j_1) \below \alpha_i(j_1') \below
  \alpha_i(j) \equiv \hat\alpha(i,j)\), and similarly for \((i_2,j_2)\), which
  proves semidirectedness of \(\hat\alpha\).
\end{proof}

\begin{lemma}\label{sup-map-is-monotone}%
  Taking directed suprema defines a monotone map from a \(\V\)-dcpo to its
  \(\V\)-ind-completion, denoted by \({\bigsqcup} : \Ind{V}(D) \to D\).
\end{lemma}
\begin{proof}
  We have to show that \(\bigsqcup \alpha \below \bigsqcup \beta\) for directed
  families \(\alpha\) and \(\beta\) such that \(\beta\) exceeds \(\alpha\). Note
  that the inequality holds as soon as \(\alpha_i \below \bigsqcup \beta\) for
  every \(i\) in the domain of \(\alpha\). For this, it suffices that for every
  such \(i\), there exists a \(j\) in the domain of \(\beta\) such that
  \(\alpha_i \below \beta_j\). But the latter says exactly that \(\beta\)
  exceeds \(\alpha\).
\end{proof}

\citeauthor{JohnstoneJoyal1982}~\cite{JohnstoneJoyal1982} generalise the notion
of continuity from posets to categories and do so by phrasing it in terms of
\(\bigsqcup : \Ind{V}(D) \to D\) having a left adjoint.
We follow their approach and now work towards this.
It turns out to be convenient to use the following two notions, which are in
fact equivalent by \cref{approximate-adjunct-coincidence}:

\begin{definition}[Approximate, left adjunct]%
  For an element \(x\) of a dcpo \(D\) and a directed family
  \(\alpha : I \to D\), we say that
  \begin{enumerate}[(i)]
  \item \(\alpha\) \emph{approximates} \(x\) if the supremum of \(\alpha\) is
    \(x\) and each \(\alpha_i\) is way below \(x\), and
  \item \(\alpha\) is \emph{left adjunct} to \(x\) if
    \(\alpha \cof \beta \iff x \below \bigsqcup \beta\) for every directed
      family \(\beta\).
      \qedhere
  \end{enumerate}
\end{definition}

\begin{remark}\label{left-adjoint-in-terms-of-left-adjunct-to}
  For a \(\V\)-dcpo \(D\), a function \(L : D \to \Ind{V}(D)\) is a left adjoint
  to \({{\bigsqcup} : \Ind{V}(D) \to D}\) precisely when \(L(x)\) is left
  adjunct to \(x\) for every \(x : D\).
  Of course, we need to know that \(L\) is monotone and this is shown in the
  next lemma.
\end{remark}

\begin{lemma}\label{left-adjoint-monotone}
  A function \(L : D \to \Ind{V}(D)\) is monotone and order-reflecting if
  \(L(x)\) is left adjunct to \(x\) for every \(x : D\).
\end{lemma}
\begin{proof}
  Suppose we are given elements \(x,y : D\). By assumption, we know that
  \(L(x) \cof L(y) \iff x \below \bigsqcup L(y)\), but \(L(y)\) approximates
  \(y\), so \(\bigsqcup L(y) = y\) and hence \(L(x) \cof L(y) \iff x \below y\),
  so \(L\) preserves and reflects the order.
\end{proof}

\begin{lemma}\label{approximate-adjunct-coincidence}
  A directed family \(\alpha\) approximates an element \(x\) if and only if it
  is left adjunct to it.
\end{lemma}
\begin{proof}
  Suppose first that \(\alpha\) approximates \(x\). If \(\alpha \cof \beta\),
  then \(x = \bigsqcup \alpha \below \bigsqcup \beta\), by
  \cref{sup-map-is-monotone}. Conversely, if \(x \below \bigsqcup \beta\), then
  \(\beta\) exceeds \(\alpha\): for if \(i\) is in the domain of \(\alpha\),
  then \(\alpha_i \ll x\), so there exists \(j\) such that
  \(\alpha_i \below \beta_j\) already.

  In the other direction, suppose that \(\alpha\) is left adjunct to \(x\). We
  show that each \(\alpha_i\) is way below \(x\). If \(\beta\) is a directed
  family with \(x \below \bigsqcup \beta\), then \(\beta\) exceeds \(\alpha\) as
  \(\alpha\) is assumed to be left adjunct to \(x\). Hence, for every \(i\),
  there exists \(j\) with \(\alpha_i \below \beta_j\), proving that
  \(\alpha_i \ll x\).
  Since \(\alpha\) exceeds itself, we get \(x \below \bigsqcup \alpha\) by
  assumption. For the other inequality, we note that
  \(x \below \bigsqcup \hat{x}\), where \(\hat{x} : \One \to D\) is the directed
  family that maps to \(x\). Hence, as \(\alpha\) is left adjunct to \(x\), we
  must have that \(\hat{x}\) exceeds \(\alpha\), which means that each
  \(\alpha_i\) is below \(x\). Thus, \(\bigsqcup \alpha \below x\)
  and \(x = \bigsqcup \alpha\) hold, as desired.
\end{proof}

\begin{proposition}\label{left-adjoint-characterisation}%
  For a \(\V\)-dcpo \(D\), a function \(L : D \to \Ind{V}(D)\) is a left
  adjoint to \(\bigsqcup : \Ind{V}(D) \to D\) if and only if \(L(x)\)
  approximates \(x\) for every \(x : D\).
\end{proposition}
\begin{proof}
  Immediate from \cref{approximate-adjunct-coincidence} and
  \cref{left-adjoint-in-terms-of-left-adjunct-to}.
\end{proof}

\section{Continuous and algebraic dcpos}\label{sec:continuous-and-algebraic-dcpos}
Using the ind-completion from the previous section, we turn to defining
continuous and algebraic dcpos, paying special attention to size and
constructivity issues regarding the axiom of choice.
This second issue is discussed in \cref{sec:pseudocontinuity} from two
perspectives: type-theoretically, via a discussion on the placement of the
propositional truncation, and categorically, via left adjoints.

\subsection{Continuous dcpos}\label{sec:continuous-dcpos}

We define what it means for a \(\V\)-dcpo to be continuous and prove the
fundamental interpolation properties for the way-below relation. Examples are
postponed (see~\cref{sec:algebraic-examples,sec:dyadics}) until we have
developed the theory further.

\begin{definition}[Continuity data, \(I_x\), \(\alpha_x\)]%
  \emph{Continuity data} for a \(\V\)-dcpo \(D\) assigns to every \(x : D\) a
  type \(I_x : \V\) and a directed \emph{approximating family}
  \(\alpha_x : I \to D\) such that \(\alpha_x\) has supremum \(x\) and each
  \(\alpha_x(i)\) is way below~\(x\).
\end{definition}

The notion of continuity data can be understood categorically as follows.

\begin{proposition}\label{structural-continuity-in-terms-of-ladj}%
  The type of continuity data for a \(\V\)-dcpo \(D\) is equivalent to the type
  of left adjoints to \({\bigsqcup} : \Ind{V}(D) \to D\).
\end{proposition}
\begin{proof}
  Given continuity data for \(D\), we define a function \(D \to \Ind{V}(D)\) by
  sending \(x : D\) to the directed family \(\alpha_x\). This is indeed a left
  adjoint to \({\bigsqcup} : \Ind{V}(D) \to D\) because of
  \cref{left-adjoint-characterisation}.
  Conversely, given a left adjoint \(L : {D \to \Ind{V}(D)}\), the assignment
  \(x \mapsto L(x)\) is continuity data for \(D\), again by
  \cref{left-adjoint-characterisation}.
  That these two maps make up a type equivalence can be checked directly,
  using that the type expressing that a map is a left adjoint to
  \({\bigsqcup} : \Ind{V}(D) \to D\) is a proposition.
\end{proof}

\begin{remark}\label{continuity-prop-vs-data}
  It should be noted that having continuity data is \emph{not} property of
  a dcpo, i.e., the type of continuity data for a dcpo is not a subsingleton.
  Indeed an element \(x : D\) can have several different approximating
  families, e.g.\ if \(\alpha : I \to D\) approximates \(x\), then so does
  \([\alpha,\alpha] : (I + I) \to D\).
  In other words, the left adjoint to \(\bigsqcup : \Ind{V}(D) \to D\) is not
  unique, although it is unique up to isomorphism (of the preorder
  \(\Ind{V}(D)\)).
  In category theory this is typically sufficient, and often the best one can
  do. \citeauthor{JohnstoneJoyal1982} follow this philosophy
  in~\cite{JohnstoneJoyal1982}, but we want the type of continuous
  \(\V\)\nobreakdash-dcpos to be a subtype of the \(\V\)-dcpos.
  One reason that property is preferred is that we only get a univalent
  category~\cite{AhrensKapulkinShulman2015} of continuous dcpos if we consider
  maps that preserve imposed structure.
  In the case of continuity data, this would imply preservation
  of the way-below relation, but this excludes many Scott continuous maps, e.g.\
  if the value of a constant map is not compact, then the map does not preserve the
  way-below relation.

  It is natural to ask whether the univalence axiom can be used to identify
  these isomorphic objects. However, this is not the case because the
  ind-completion \(\Ind{V}(D)\) is \emph{not} a univalent category in the sense
  of~\cite{AhrensKapulkinShulman2015}, as it is a preorder and not a poset.
  One way to obtain a subtype is to propositionally truncate the notion of
  continuity data and this is indeed the approach that we will take. However,
  another choice that would yield a property is to identify isomorphic elements
  of \(\Ind{V}(D)\). This approach is discussed at length in
  \cref{sec:pseudocontinuity} and in particular it is explained to be inadequate
  in a constructive setting.
\end{remark}

\begin{definition}[Continuity of a dcpo]
  A \(\V\)-dcpo is \emph{continuous} if it has unspecified continuity
  data, i.e.\ if the type of continuity data is inhabited.
\end{definition}
Thus, a dcpo is continuous if we have an \emph{unspecified} function assigning
an approximating family to every element of the dcpo.

\begin{proposition}\label{equality-of-continuous-dcpos}
  Assuming univalence, the identity type of two continuous \(\V\)-dcpos \(D\)
  and \(E\) is equivalent to the type of dcpo isomorphisms from \(D\) to \(E\).
\end{proposition}
\begin{proof}
  By an application of the structure identity
  principle~\cite[Section~33.14]{Escardo2019}, the identity type of \(D\)
  and \(E\) is equivalent to the type of order preserving and reflecting
  equivalences from \(D\) to \(E\). But it is straightforward to show that an
  order preserving and reflecting equivalence is precisely a dcpo isomorphism.
\end{proof}

If we are working with small dcpos, which requires propositional resizing, then
it is possible to extract continuity data from knowing that a dcpo is continuous
because we can simply consider all elements way below a given element, see
\cite{vanCollem2023}.

\begin{proposition}
  Continuity of a \(\V\)-dcpo \(D\) is equivalent to having an unspecified left
  adjoint to \({\bigsqcup} : \Ind{V}(D) \to D\).
\end{proposition}
\begin{proof}
  By \cref{structural-continuity-in-terms-of-ladj} and functoriality of the
  propositional truncation.
\end{proof}

\begin{lemma}\label{structurally-continuous-below-characterisation}
  For elements \(x\) and \(y\) of a dcpo with continuity data, the
  following are equivalent:
  \begin{enumerate}[(i)]
  \item\label{item-below} \(x \below y\);
  \item\label{item-approx-below} \(\alpha_x(i) \below y\) for every \(i : I_x\);
  \item\label{item-approx-way-below} \(\alpha_x(i) \ll y\) for every
    \(i : I_x\).
  \end{enumerate}
\end{lemma}
\begin{proof}
  Note that \eqref{item-approx-way-below} implies \eqref{item-approx-below} and
  \eqref{item-approx-below} implies \eqref{item-below}, because if
  \(\alpha_x(i) \below y\) for every \(i : I_x\), then
  \(x = \bigsqcup \alpha_x \below y\), as desired. So it remains to prove that
  \eqref{item-below} implies \eqref{item-approx-way-below}, but this holds, because
  \(\alpha_x(i) \ll x\) for every \(i : I_x\).
\end{proof}

\begin{lemma}\label{structurally-continuous-way-below-characterisation}
  For elements \(x\) and \(y\) of a dcpo with continuity data, \(x\) is
  way below~\(y\) if and only if there exists \(i : I_y\) such that
  \(x \below \alpha_y(i)\).
\end{lemma}
\begin{proof}
  The left-to-right implication holds, because \(\alpha_y\) is a directed family
  with supremum \(y\), while the converse holds because \(\alpha_y(i) \ll y\)
  for every \(i : I_y\).
\end{proof}

We now prove three interpolation lemmas for dcpos with continuity data. Because
the conclusions of the lemmas are propositions, the results will follow for
continuous dcpos immediately.

\begin{lemma}[Nullary interpolation for the way-below relation]
  \label{nullary-interpolation}
  For every \(x : D\) of a continuous dcpo \(D\), there exists \(y : D\) such
  that \(y \ll x\).
\end{lemma}
\begin{proof}
  The approximating family \(\alpha_x\) is directed, so there exists \(i : I_x\)
  and hence we can take \(y \colonequiv \alpha_x(i)\) since
  \(\alpha_x(i) \ll x\).
\end{proof}

Although there are constructive proofs of the following in the literature,
e.g.~\cite[Proposition~2.12]{AbramskyJung1994}, they are impredicative. Instead,
we develop a predicative proof inspired by
\cite[Proposition~2.12]{JohnstoneJoyal1982}.
\begin{lemma}[Unary interpolation for the way-below relation]\label{unary-interpolation}
  If \(x \ll y\) in a continuous dcpo \(D\), then there exists an
  \emph{interpolant} \(d : D\) such that \(x \ll d \ll y\).
\end{lemma}
\begin{proof}
  Since we are proving a proposition, we may assume to be given continuity
  data for \(D\).
  Thus, we can approximate every approximant \(\alpha_y(i)\) of \(y\) by an
  approximating family \(\beta_{i} : J_i \to D\). This defines a map
  \(\hat{\beta}\) from \(I_y\) to \(\Ind{V}(D)\), the ind-completion of the
  \(\V\)-dcpo \(D\), by sending \(i : I_y\) to the directed family \(\beta_i\).
  We claim that \(\hat\beta\) is directed in \(\Ind{V}(D)\). Since \(\alpha_y\)
  is directed, \(I_y\) is inhabited, so it remains to prove that \(\hat\beta\)
  is semidirected. So suppose we have \(i_1,i_2 : I_y\). Because \(\alpha_y\) is
  semidirected, there exists \(i : I_y\) such that
  \(\alpha_y(i_1),\alpha_y(i_2) \below \alpha_y(i)\). We claim that \(\beta_i\)
  exceeds \(\beta_{i_1}\) and \(\beta_{i_2}\), which would prove
  semidirectedness of \(\hat\beta\). We give the argument for \(i_1\) only as
  the case for \(i_2\) is completely analogous.
  We have to show that for every \(j : J_{i_1}\), there
  exists \(j' : J_i\) such that \(\beta_{i_1}(j) \below \beta_{i}(j')\).
  But this holds because \(\beta_{i_1}(j) \ll \bigsqcup \beta_i\) for every such
  \(j\), as we have
  \(\beta_{i_1}(j) \ll \alpha_y(i_1) \below \alpha_y(i) \below \bigsqcup
  \beta_i\).

  Thus, \(\hat\beta\) is directed in \(\Ind{V}(D)\) and hence we can calculate
  its supremum in \(\Ind{V}(D)\) to obtain the \emph{directed} family
  \(\gamma : \pa{\Sigma_{i : I}J_i} \to D\) given by
  \((i,j) \mapsto \beta_i(j)\).

  We now show that \(y\) is below the supremum of \(\gamma\). By
  \cref{structurally-continuous-below-characterisation}, it suffices to prove
  that \(\alpha_y(i) \below \bigsqcup \gamma\) for every \(i : I_y\), and, in
  turn, to prove this for an \(i : I_y\) it suffices to prove that
  \(\beta_i(j) \below \bigsqcup \gamma\) for every \(j : J_i\). But this is
  immediate from the definition of \(\gamma\). Thus,
  \(y \below \bigsqcup \gamma\).
  Because \(x \ll y\), there exists \((i,j) : \Sigma_{i : I}J_i\) such that
  \(x \below \gamma(i,j) \equiv \beta_i(j)\).

  Finally, for our interpolant, we take \(d \colonequiv \alpha_y(i)\). Then,
  indeed, \(x \ll d \ll y\), because
  \(x \below \beta_i(j) \ll \alpha_y(i) \equiv d\) and
  \(d \equiv \alpha_y(i) \ll y\), completing the proof.
\end{proof}

The proof of the following is a straightforward application of unary
interpolation as in the classical case.
\begin{lemma}[Binary interpolation for the way-below relation]%
  \label{binary-interpolation}
  If \(x \ll z\) and \(y \ll z\) in a continuous dcpo \(D\), then
  there exists an interpolant \(d : D\) such that \(x,y \ll d\) and
  \(d \ll z\).
\end{lemma}

Continuous dcpos are closed under retracts. Keeping track of universes, it holds
in the following generality, where we recall (\cref{def:DCPO}) that we write
\(\DCPO{V}{U}{T}\) for the type of \(\V\)-dcpos with carriers in~\(\U\) and
order relations taking values in \(\T\).

\begin{theorem}\label{continuity-closed-under-continuous-retracts}%
  If we have dcpos \(D : \DCPO{V}{U}{T}\) and \(E : \DCPO{V}{U'}{T'}\) such that
  \(D\)~is a retract of \(E\), then \(D\) is continuous if \(E\) is.
  Moreover, we can give continuity data for \(D\) if we have such data for
  \(E\).
\end{theorem}
\begin{proof}
  We prove the result in case we are given continuity data for \(E\), as the
  other will follow from that and the fact that the propositional truncation is
  functorial. So suppose that we have a Scott continuous section \(s : D \to E\)
  and retraction \(r : E \to D\). We claim that for every \(x : D\), the family
  \(r \circ \alpha_{s(x)}\) approximates~\(x\).
  Firstly, it is directed, because \(\alpha_{s(x)}\) is and \(r\) is Scott
  continuous.
  Secondly,
  \begin{align*}
    \textstyle\bigsqcup r \circ \alpha_{s(x)}
    &= r\pa*{\textstyle\bigsqcup\alpha_{s(x)}}
      &&\text{(by Scott continuity of \(r\))}
    \\
    &= r(s(x))
      &&\text{(as \(\alpha_{s(x)}\) is the approximating family of \(s(x)\))}
    \\
    &= x &&\text{(because \(s\) is a section of \(r\))},
  \end{align*}
  so the supremum of \(r \circ \alpha_{s(x)}\) is \(x\).
  Finally, we must prove that \(r\pa*{\alpha_{s(x)}(i)} \ll x\) for every
  \(i : I_x\). By \cref{continuous-retraction-way-below-criterion}, this is
  implied by \(\alpha_{s(x)}(i) \ll s(x)\), which holds as \(\alpha_{s(x)}\) is
  the approximating family of \(s(x)\).
\end{proof}

Recall from \cref{def:local-smallness}
that a \(\V\)-dcpo is locally small if
\(x \below y\) is equivalent to a type in \(\V\) for all elements \(x\) and
\(y\).

\begin{proposition}\label{cont-loc-small-iff-way-below-small}%
  A continuous dcpo is locally small if and only if its way-below relation has
  small truth values.
\end{proposition}
\begin{proof}
  By \cref{structurally-continuous-below-characterisation,%
    structurally-continuous-way-below-characterisation}, we have
  \[
    x \below y \iff \forall_{i : I_x}\pa*{\alpha_x(i) \ll y} %
    \quad\text{and}\quad %
    x \ll y \iff \exists_{i : I_y}\pa*{x \below \alpha_y(i)},
  \]
  for every two elements \(x\) and \(y\) of a dcpo with continuity data.
  But the types \(I_x\)~and~\(I_y\) are small, finishing the proof.
  The result also holds for continuous dcpos, because what we are proving is a
  proposition.
\end{proof}

\cref{cont-loc-small-iff-way-below-small} is significant because the definition
of the way-below relation for a \(\V\)-dcpo \(D\) quantifies over all families
into \(D\) indexed by types in \(\V\).

\subsection{Pseudocontinuity}\label{sec:pseudocontinuity}
In light of~\cref{structural-continuity-in-terms-of-ladj}, we see that a
\(\V\)\nobreakdash-dcpo \(D\) can have continuity data in more than one
way: the map \({\bigsqcup} : \Ind{V}(D) \to D\) can have two left adjoints
\(L_1,L_2\) such that for some \(x : D\), the directed families \(L_1(x)\) and
\(L_2(x)\) exceed each other, yet are unequal.
In order for the left adjoint to be truly unique (and not just up to
isomorphism), the preorder \(\Ind{\V}(D)\) would have to identify families that
exceed each other.
Of course, we could enforce this identification by passing to the poset
reflection \(\Ind{V}(D)/{\approx}\) of \(\Ind{V}(D)\) and this section studies
exactly that.

Another perspective on the situation is the following: The type-theoretic
definition of having continuity data for a \(\V\)-dcpo \(D\) is of the
form \(\Pi_{x : D}\Sigma_{I : \V}\Sigma_{\alpha : I \to D}\dots\),
while continuity is defined as its propositional truncation
\(\squash*{\Pi_{x : D}\Sigma_{I : \V}\Sigma_{\alpha : I \to D}\dots}\).
Yet another way to obtain a property is by putting the propositional truncation on
the \emph{inside} instead:
\(\Pi_{x : D}\squash*{\Sigma_{I : \V}\Sigma_{\alpha : I \to D}\dots}\).
We study what this amounts to and how it relates to continuity and the poset
reflection. Our results are summarised in~\cref{continuity-table}.

\begin{definition}[Pseudocontinuity]
  A \(\V\)-dcpo \(D\) is \emph{pseudocontinuous} if for every \(x : D\) there
  exists an unspecified directed family that approximates \(x\).
\end{definition}

Notice that continuity data \(\Rightarrow\) continuity \(\Rightarrow\)
pseudocontinuity, but reversing the first implication is an instance of global
choice~\cite[Section~3.35.6]{Escardo2019}, while reversing the second amounts to
an instance of the axiom of choice that we do not expect to be provable in our
constructive setting. We further discuss this point in
\cref{pseudocontinuous-choice-issues}.

For a \(\V\)-dcpo \(D\), the map \({\bigsqcup} : \Ind{V}(D) \to D\) is monotone, so
it induces a unique monotone map
\({\bigsqcup_{\approx}} : \Ind{V}(D)/{\approx} \to D\) such that the diagram
\begin{equation}\label{supremum-map-quotient-comm}
  \begin{tikzcd}
    \Ind{V}(D)/{\approx}\ar[rr,"{\bigsqcup_{\approx}}"]
    \ar[dr,"\toquotient{-}"']
    & & D \\
    & \Ind{V}(D) \ar[ur,"{\bigsqcup}"']
  \end{tikzcd}
\end{equation}
commutes.

\begin{proposition}
  A \(\V\)-dcpo \(D\) is pseudocontinuous if and only if the map of posets
  \({\bigsqcup}_{\approx} : \Ind{V}(D)/{\approx} \to D\) has a necessarily unique left
  adjoint.
\end{proposition}

Observe that the type of left adjoints to
\({\bigsqcup}_{\approx} : \Ind{V}(D)/{\approx} \to D\) is a
proposition, precisely because \(\Ind{V}(D)/{\approx}\) is a poset,
cf.~\cite[Lemma~5.2]{AhrensKapulkinShulman2015}, and hence the above
uniqueness condition amounts to the contractibility of the type of
left adjoints.

\begin{proof}
  Suppose that \({\bigsqcup}_{\approx} : \Ind{V}(D)/{\approx} \to D\) has a left
  adjoint \(L\) and let \(x : D\) be arbitrary. We have to prove that there
  exists a directed family \(\alpha : I \to D\) that approximates \(x\). By
  surjectivity of the universal map \(\toquotient{-}\), there exists a directed
  family \(\alpha : I \to D\) such that \(L(x) = \toquotient{\alpha}\).
  Moreover, \(\alpha\) approximates \(x\) by virtue of
  \cref{approximate-adjunct-coincidence}, since for every
  \(\beta : \Ind{V}(D)\), we have
  \begin{align*}
    \alpha \cof \beta
    &\iff L(x) \leq \toquotient{\beta}
      &&\text{(since \(L(x) = \toquotient{\alpha}\))} \\
    &\iff x \below \textstyle\bigsqcup_{\approx}{\toquotient{\beta}}
      &&\text{(since \(L\) is a left adjoint to \(\textstyle{\bigsqcup}_{\approx}\))}\\
    &\iff x \below \textstyle\bigsqcup \beta
      &&\text{(by \cref{supremum-map-quotient-comm})}.
  \end{align*}

  The converse is more involved and we apply
  \cref{constant-map-to-set-factors-through-truncation} which we recall says
  that every constant map with values in a set factors through the propositional
  truncation of its domain.
  Assume that \(D\) is pseudocontinuous. We start by constructing the left
  adjoint, so let \(x : D\) be arbitrary. Writing \(\mathcal A_x\) for the type
  of directed families that approximate \(x\), we have an obvious map
  \(\varphi_x : \mathcal A_x \to \Ind{V}(D)\) that forgets that the directed
  family approximates \(x\).

  We claim that all elements in the image of \(\varphi_x\) exceed each
  other. For if \(\alpha\) and \(\beta\) are directed families both
  approximating \(x\), then for every \(i\) in the domain of \(\alpha\) we know
  that \(\alpha_i \ll x = \bigsqcup \beta\), so that there exists \(j\) with
  \(\alpha_i \below \beta_j\).
  Hence, passing to the poset reflection, the composite
  \(\toquotient{-} \circ \varphi_x\) is constant.
  Thus, by~\cref{constant-map-to-set-factors-through-truncation} we have a
  (necessarily unique) map \(\psi_x\) making the diagram
  \[
    \begin{tikzcd}
      \mathcal A_x \ar[rr,"\toquotient{-} \circ \varphi_x"]
      \ar[dr,"\tosquash{-}"']
      & & \Ind{V}(D)/{\approx} \\
      & \squash*{\mathcal A_x}
      \ar[ur,"\psi_x"',dashed]
    \end{tikzcd}
  \]
  commute.
  Since \(D\) is assumed to be pseudocontinuous, we have that
  \(\squash*{\mathcal A_x}\) holds for every \(x : D\), so together with \(\psi_x\)
  this defines a map \(L : D \to \Ind{V}(D)/{\approx}\) by
  \(L(x) \colonequiv \psi_x(p)\), where \(p\) witnesses pseudocontinuity at
  \(x\).

  Lastly, we prove that \(L\) is indeed a left adjoint to
  \({\bigsqcup}_{\approx}\). So let \(x : D\) be arbitrary. Since we're proving
  a property, we can use pseudocontinuity at \(x\) to specify a directed family
  \(\alpha\) that approximates \(x\). We now have to prove
  \(\toquotient{\alpha} \leq \beta' \iff x \below \bigsqcup_{\approx} \beta'\)
  for every \(\beta' : \Ind{V}(D)/{\approx}\). This is a proposition, so using
  quotient induction, it suffices to prove
  \(\toquotient{\alpha} \leq \toquotient{\beta} \iff x \below
    \bigsqcup_{\approx} \toquotient{\beta}\) for every \(\beta : \Ind{V}(D)\).
  Indeed, for such \(\beta\) we have
  \begin{align*}
    \toquotient{\alpha} \leq \toquotient{\beta}
    &\iff \alpha \cof \beta \\
    &\iff x \below \textstyle\bigsqcup \beta
    &&\text{(by \cref{approximate-adjunct-coincidence} and because \(\alpha\) approximates \(x\))} \\
    &\iff x \below \textstyle\bigsqcup_{\approx} \toquotient{\beta}
    &&\text{(by \cref{supremum-map-quotient-comm})},
  \end{align*}
  finishing the proof.
\end{proof}

Thus, the explicit type-theoretic formulation and the formulation using left
adjoints in each row of~\cref{continuity-table} (which summarises our findings)
are equivalent.

\begin{table}[h]
  \centering
  \begin{tabular}{p{1.6cm}cp{3.85cm}c}\toprule 
    & Type-theoretic formulation
    & Formulation with adjoints & Prop. \\\midrule
    Cont.\ data
    & \(\Pi_{x : D} \Sigma_{I : \V} \Sigma_{\alpha : I \to D}\,\delta(\alpha,x)\)
    & Specified left adjoint to \({\bigsqcup : \Ind{V}(D) \to D}\) 
    & \ding{53} \\
    Continuity
    & \(\squash*{\Pi_{x : D} \Sigma_{I : \V} \Sigma_{\alpha : I \to D}\,\delta(\alpha,x)}\)
    & Unspecified left adjoint to \({\bigsqcup : \Ind{V}(D) \to D}\) 
    & \checkmark \\
    Pseudocont.
    & \(\Pi_{x : D}\squash*{\Sigma_{I : \V} \Sigma_{\alpha : I \to D}\,\delta(\alpha,x)}\)
    & Specified left adjoint to \({\bigsqcup_{\approx} : \Ind{V}(D)/{\approx} \to D}\)
    & \checkmark \\
    \bottomrule
  \end{tabular}
  \caption{Continuity (data) and pseudocontinuity of a dcpo~\(D\), where
    \(\delta(\alpha,x)\) abbreviates that \(\alpha\) is directed and
    approximates~\(x\).}
  \label{continuity-table}
\end{table}

\begin{remark}\label{pseudocontinuous-choice-issues}%
  The issue with pseudocontinuity is that taking the quotient introduces a
  dependence on instances of the axiom of choice when it comes to proving
  properties of pseudocontinuous dcpos. An illustrative example is the proof of
  unary interpolation~(\cref{unary-interpolation}), where we used the continuity
  data to first approximate an element \(y\) by \(\alpha_y\) and then, in turn,
  approximate every approximant \(\alpha_y(i)\). With pseudocontinuity this
  argument would require \emph{choosing} an approximating family for every
  \(i\).
  Another example is that while the preorder \(\Ind{V}(D)\) is \(\V\)-directed
  complete, a direct lifting of the proof of this fact to the poset reflection
  \(\Ind{V}(D)/{\approx}\) requires the axiom of choice.
  Hence, the Rezk completion~\cite{AhrensKapulkinShulman2015}, of which the
  poset reflection is a special case, does not necessarily preserve (filtered)
  colimits.
  The same issues concerning the axiom of choice occur
  in~\cite[pp.~260--261]{JohnstoneJoyal1982} and the notion of continuity data
  follows their solution precisely. We then truncate this to get a property of
  dcpos (recall~\cref{continuity-prop-vs-data}), resulting in our definition of
  continuity.
\end{remark}

\subsection{Algebraic dcpos}\label{sec:algebraic-dcpos}

Many of our examples of dcpos are not just continuous, but satisfy the stronger
condition of being algebraic, meaning their elements can be approximated by
compact elements only.

\begin{definition}[Algebraicity data, \(I_x\), \(\kappa_x\)]%
  \emph{Algebraicity data} for a \(\V\)-dcpo \(D\) assigns to every
  \(x : D\) a type \(I_x : \V\) and a directed family
  \(\kappa_x : I \to D\) of \emph{compact} elements such that \(\kappa_x\)
  has supremum \(x\).
\end{definition}

\begin{definition}[Algebraicity]
  A \(\V\)-dcpo is \emph{algebraic} if it has some unspecified algebraicity
  data, i.e.\ if the type of algebraicity data is inhabited.
\end{definition}

\begin{lemma}
  Every algebraic dcpo is continuous.
\end{lemma}
\begin{proof}
  We prove that algebraicity data for a dcpo yields continuity data. The
  claim for algebraic and continuous then follows by functoriality of the
  propositional truncation.
  It suffices to prove that \(\kappa_x(i) \ll x\) for every \(i : I_x\).  By
  assumption, \(\kappa_x(i)\) is compact and has supremum \(x\). Hence,
  \(\kappa_x(i) \ll \kappa_x(i) \below \bigsqcup\kappa_x = x\), so
  \(\kappa_x(i) \ll x\).
\end{proof}

\section{Small bases}\label{sec:small-bases}

Recall that the traditional, set-theoretic definition of a dcpo \(D\) being
continuous says that for every element \(x \in D\), the subset
\(\set{y \in D \mid y \ll x}\) is directed with supremum~\(x\).
As explained in the introduction, the problem with this definition in a
predicative context is that the subset \(\set{y \in D \mid y \ll x}\) is not
small in general.
But, as is well-known in domain theory, it is sufficient (and in fact
equivalent) to instead ask that \(D\) has a subset \(B\), known as a
\emph{basis}, such that the subset \(\set{b \in B \mid b \ll x} \subseteq B\) is
directed with supremum \(x\), see~\cite[Section~2.2.2]{AbramskyJung1994} and
\cite[Definition~III-4.1]{GierzEtAl2003}.
The idea developed in this section is that in many examples we can find a
\emph{small} basis giving us a predicative handle on the situation.

If a dcpo has a small basis, then it is continuous. In fact, all our running
examples of continuous dcpos are actually examples of dcpos with small
bases. Moreover, dcpos with small bases are better behaved. For example, they
are all locally small and so are their exponentials, which also have small bases
(\cref{sec:exponentials-with-small-bases}). In
\cref{sec:ideal-completions-of-small-bases} we also show that having a small
basis is equivalent to being presented by ideals.

Following the flow of \cref{sec:continuous-and-algebraic-dcpos}, we first consider
small bases for continuous dcpos, before turning to small compact bases for
algebraic dcpos (\cref{sec:small-compact-bases}).
After presenting examples of dcpos with small compact bases in
\cref{sec:algebraic-examples}, we describe the canonical small compact basis for
an algebraic dcpo and the role that the univalence axiom and a set replacement
principle play in \cref{sec:basis-of-compact-elements}.

\begin{definition}[Small basis]
  For a \(\V\)-dcpo \(D\), a map \(\beta : B \to D\) with \(B : \V\) is a
  \emph{small basis} for \(D\) if the following conditions hold:
  \begin{enumerate}[(i)]
  \item\label{basis-approximating} for every \(x : D\), the family
    \(\pa*{\Sigma_{b : B}\pa*{\beta(b) \ll x}} \xrightarrow{\beta \circ \fst}
    D\) is directed and has supremum \(x\);
  \item\label{basis-small-way-below} for every \(x : D\) and \(b : B\), the
    proposition \(\beta(b) \ll x\) is \(\V\)-small.
  \end{enumerate}
  We will write \(\ddset_\beta x\) for the type
  \(\Sigma_{b : B}\pa*{\beta(b) \ll x}\) and conflate this type with the
  canonical map \(\ddset_\beta x \xrightarrow {\beta \circ \fst} D\).
\end{definition}

Item \eqref{basis-small-way-below} ensures not only that the type
\(\Sigma_{b : B}\pa*{\beta(b) \ll x}\) is \(\V\)-small, but also that a dcpo
with a small basis is locally small (\cref{locally-small-if-small-basis}).

\begin{remark}\label{tacitly-use-small-basis}
  If \(\beta : B \to D\) is a small basis for a \(\V\)-dcpo \(D\), then the type
  \(\ddset_\beta x\) is small. Hence, we have a type \(I : \V\) and an equivalence
  \(\varphi : I \simeq \ddset_\beta x\) and we see that the family
  \(I \xrightarrow{\varphi} \ddset_\beta x \xrightarrow{\beta \circ \fst} D\) is
  directed and has the same supremum as \(\ddset_\beta x \to D\).
  We will use this tacitly and write as if the type \(\ddset_\beta x\) is actually a
  type in \(\V\).
\end{remark}

\begin{lemma}\label{structural-continuity-if-small-basis}%
  If a dcpo comes equipped with a small basis, then it can be equipped with
  continuity data. Hence, if a dcpo has an unspecified small basis, then it is
  continuous.
\end{lemma}
\begin{proof}
  For every element \(x\) of a dcpo \(D\), the family \(\ddset_\beta x \to D\)
  approximates \(x\), so the assignment \(x \mapsto \ddset_\beta x\) is
  continuity data for \(D\).
\end{proof}

\begin{lemma}\label{below-in-terms-of-way-below-basis}
  In a dcpo \(D\) with a small basis \(\beta : B \to D\), we have
  \( x \below y\) if and only if for every \(b : B\), the condition
  \(\beta(b) \ll x\) implies \(\beta(b) \ll y\).
\end{lemma}
\begin{proof}
  If \(x \below y\) and \(\beta(b) \ll x\), then \(\beta(b) \ll y\), so the
  left-to-right implication is clear.
  For the converse, suppose that the condition of the lemma holds. Because
  \(x = \bigsqcup \ddset_\beta x\), the inequality \(x \below y\) holds as soon as
  \(\beta(b) \below y\) for every \(b : B\) with \(\beta(b) \ll x\), but this is
  implied by the condition.
\end{proof}

\begin{proposition}\label{locally-small-if-small-basis}%
  A dcpo with a small basis is locally small. Moreover, the way-below relation
  on all of the dcpo has small values.
\end{proposition}
\begin{proof}
  The first claim follows from \cref{below-in-terms-of-way-below-basis} and the
  second follows from the first and \cref{cont-loc-small-iff-way-below-small}.
\end{proof}

A notable feature of dcpos with a small basis is that interpolants for the
way-below relation, cf.\
\cref{nullary-interpolation,unary-interpolation,binary-interpolation}, can be
found in the basis.
Using \cref{structural-continuity-if-small-basis} which constructs
continuity data from a small basis, the proofs are as in the classical case.

\begin{lemma}[Interpolation in the basis for the way-below relation]\label{interpolation-basis}
  Suppose \(D\) is a dcpo with a small basis \(\beta : B \to D\).
  \begin{enumerate}[(i)]
  \item For every \(x : D\), there exists \(b : B\) with \(\beta(b) \ll x\).
  \item If \(x \ll y\), then there exists an interpolant \(b : B\) such that
    \(x \ll \beta(b) \ll y\).
  \item If \(x \ll z\) and \(y \ll z\), then there exists an interpolant
    \(b : B\) such that \(x,y \ll \beta(b) \ll z\).
  \end{enumerate}
\end{lemma}

Before proving the analogue of
\cref{continuity-closed-under-continuous-retracts} (closure under retracts) for
small bases, we need a type-theoretic analogue of
\cite[Proposition~2.2.4]{AbramskyJung1994} and
\cite[Proposition~III-4.2]{GierzEtAl2003}, which essentially says that it is
sufficient for a ``subset'' of \(\ddset_\beta x\) (given by \(\sigma\) in the
lemma) to be directed and have supremum~\(x\).

\begin{lemma}\label{subbasis-lemma}
  Suppose that we have an element \(x\) of a \(\V\)-dcpo \(D\) together with two
  maps \(\beta : B \to D\) and
  \(\sigma : I \to \Sigma_{b : B}\pa*{\beta(b) \ll x}\) with \(I : \V\).
  If \(\ddset_\beta x \circ \sigma\) is directed and has supremum \(x\), then
  \(\ddset_\beta x\) is directed with supremum \(x\) too.
\end{lemma}
\begin{proof}
  Suppose that \({\ddset_\beta x} \circ {\sigma}\) is directed and has supremum
  \(x\). Obviously, \(x\) is an upper bound for \(\ddset_\beta x\), so we are to
  prove that it is the least. If \(y\) is an upper bound for \(\ddset_\beta x\),
  then it is also an upper bound for \(\ddset_\beta x \circ \sigma\) which has
  supremum \(x\), so that \(x \below y\) follows.  So the point is directedness
  of \(\ddset_\beta x\). Its domain is inhabited, because \(\sigma\) is
  directed.
  Now suppose that we have \(b_1,b_2 : B\) with \(\beta(b_1),\beta(b_2) \ll
  x\). Since \(x = \bigsqcup \pa*{\ddset_\beta x \circ \sigma}\), there exist
  \(i_1,i_2 : I\) such that \(\beta(b_1) \below \beta(\fst(\sigma(i_1)))\) and
  \(\beta(b_2) \below \beta(\fst(\sigma(i_2)))\).
  Since \(\ddset_\beta x \circ \sigma\) is directed, there exists \(i : I\) with
  \(\beta(\fst(\sigma(i_1))),\beta(\fst(\sigma(i_2))) \below
  \beta(\fst(\sigma(i)))\). Hence, writing \(b \colonequiv \fst(\sigma(i))\), we
  have \(\beta(b) \ll x\) and \(\beta(b_1),\beta(b_2) \below \beta(b)\).
  Thus, \(\ddset_\beta x\) is directed, as desired.
\end{proof}

\begin{theorem}\label{small-basis-closed-under-continuous-retracts}
  If we have a retract \(\retract{D}{E}\) and a small basis \(\beta : B \to E\)
  for \(E\), then \(r\circ \beta\) is a small basis for \(D\).
\end{theorem}
\begin{proof}
  First of all, note that \(E\) is locally small by
  \cref{locally-small-if-small-basis}. But being locally small is closed under
  retracts by~\cref{locally-small-retract}, so \(D\) is locally small
  too. Moreover, we have continuity data for \(D\) by virtue of
  \cref{continuity-closed-under-continuous-retracts} and
  \cref{structural-continuity-if-small-basis}. Hence, the way-below relation is
  small-valued by \cref{cont-loc-small-iff-way-below-small}.
  In particular, the proposition \(r(\beta(b)) \ll x\) is small for every \(b : B\) and
  \(x : D\).

  We are going to use \cref{subbasis-lemma} to show that
  \(\ddset_{r \circ \beta} x\) is directed and has supremum \(x\) for every
  \(x : D\). By \cref{continuous-retraction-way-below-criterion}, the identity
  on \(B\) induces a well-defined map
  \(\sigma : \pa*{\Sigma_{b : B}\pa*{\beta(b) \ll s(x)}} \to \pa*{\Sigma_{b :
      B}\pa*{r(\beta(b)) \ll y}}\).
  Now \cref{subbasis-lemma} tells us that it suffices to prove that
  \(r \circ \ddset_\beta s(x)\) is directed with supremum \(x\).
  But \(\ddset_\beta s(x)\) is directed with supremum \(x\), so by Scott
  continuity of \(r\), the family \(r \circ \ddset_\beta s(x)\) is directed with
  supremum \(r(s(x)) = x\), as desired.
\end{proof}

Finally, a useful property of dcpos with small bases is that they yield locally small
exponentials, as we can restrict the quantification in the pointwise order to
elements of the small basis.

\begin{proposition}\label{exponential-is-locally-small}%
  If \(D\) is a dcpo with an unspecified small basis and \(E\) is a locally
  small dcpo, then the exponential \(E^D\) is locally small too.
\end{proposition}
\begin{proof}
  Being locally small is a proposition, so in proving the result we may assume
  that \(D\) comes equipped with a small basis \(\beta : B \to D\). For
  arbitrary Scott continuous functions \(f,g : D \to E\), we claim that
  \(f \below g\) precisely when \(\forall_{b : B}\pa*{f(\beta(b)) \below
    g(\beta(b))}\), which is a small type using that \(E\) is locally small.
  The left-to-right implication is obvious, so suppose that
  \(f(\beta(b)) \below g(\beta(b))\) for every \(b : B\) and let \(x : D\) be
  arbitrary. We are to show that \(f(x) \below g(x)\). Since
  \(x = \bigsqcup \ddset_\beta x\), it suffices to prove
  \(f\pa*{\bigsqcup \ddset_\beta x} \below g\pa*{\bigsqcup \ddset_\beta x}\) and in
  turn, that \(f(\beta(b)) \below g\pa*{\bigsqcup \ddset_\beta x}\) for every
  \(b : B\). But is easily seen to hold, because
  \(f(\beta(b)) \below g(\beta(b))\) for every \(b : B\) by assumption.
\end{proof}

\subsection{Small compact bases}\label{sec:small-compact-bases}

Similarly to the progression from continuous dcpos (\cref{sec:continuous-dcpos})
to algebraic ones (\cref{sec:algebraic-dcpos}), we now turn to small
\emph{compact} bases.

\begin{definition}[Small compact basis]%
  For a \(\V\)-dcpo \(D\), a map \(\beta : B \to D\) with \(B : \V\) is a
  \emph{small compact basis} for \(D\) if the following conditions hold:
  \begin{enumerate}[(i)]
  \item for every \(b : B\), the element \(\beta(b)\) is compact in \(D\);
  \item for every \(x : D\), the family
    \(\pa*{\Sigma_{b : B}\pa*{\beta(b) \below x}} \xrightarrow{\beta \circ \fst} D\)
    is directed and has supremum \(x\);
  \item for every \(x : D\) and \(b : B\), the proposition \(\beta(b) \below x\) is
    \(\V\)-small.
  \end{enumerate}
  We will write \(\dset_\beta x\) for the type
  \(\Sigma_{b : B}\pa*{\beta(b) \below x}\) and conflate this type with the
  canonical map \(\dset_\beta x \xrightarrow {\beta \circ \fst} D\).
\end{definition}

\begin{remark}
  If \(\beta : B \to D\) is a small compact basis for a \(\V\)-dcpo \(D\), then the
  type \(\dset_\beta x\) is small. Similarly to \cref{tacitly-use-small-basis}, we
  will use this tacitly and write as if the type \(\dset_\beta x\) is actually a
  type in \(\V\).
\end{remark}

\begin{lemma}\label{structural-algebraicity-if-small-compact-basis}%
  If a dcpo comes equipped with a small compact basis, then it can be equipped
  with algebraicity data. Hence, if a dcpo has an unspecified small compact
  basis, then it is algebraic.
\end{lemma}
\begin{proof}
  For every element \(x\) of a dcpo \(D\) with a small compact basis
  \(\beta : B \to D\), the family \(\dset_\beta x \to D\) consists of compact
  elements and approximates \(x\), so the assignment \(x \mapsto \dset_\beta x\)
  is algebraicity data for \(D\).
\end{proof}

Actually, with suitable assumptions, we can get canonical algebraicity data
from an unspecified small compact basis, as discussed in detail in
\cref{sec:basis-of-compact-elements}. This observation relies on
\cref{small-compact-basis-contains-all-compact-elements} below and we prefer to
present examples of dcpos with small compact bases first
(\cref{sec:algebraic-examples}).

\begin{lemma}\label{small-and-compact-basis}
  A map \(\beta : B \to D\) is a small compact basis for a dcpo \(D\) if and
  only if \(\beta\) is a small basis for \(D\) and \(\beta(b)\) is compact for
  every \(b : B\).
\end{lemma}
\begin{proof}
  If \(\beta(b)\) is compact for every \(b : B\), then \(\beta(b) \below x\) if
  and only if \(\beta(b) \ll x\) for every \(b : B\) and \(x : D\), so that
  \(\ddset_\beta x \simeq \dset_\beta x\) for every \(x : D\).
  In particular, \(\ddset_\beta x\) approximates \(x\) if and only if
  \(\dset_\beta x\) does, which completes the proof.
\end{proof}

\begin{proposition}\label{small-compact-basis-contains-all-compact-elements}
  A small compact basis contains every compact element. That is, if
  \(\beta : B \to D\) is a small compact basis for a dcpo \(D\) and \(x : D\) is
  compact, then there exists \(b : B\) such that \(\beta(b) = x\).
\end{proposition}
\begin{proof}
  Suppose we have a compact element \(x : D\). By compactness of \(x\) and the
  fact that \(x = \dset_\beta x\), there exists \(b : B\) with \(\beta(b) \ll x\)
  such that \(x \below \beta(b)\). But then \(\beta(b) = x\) by antisymmetry.
\end{proof}

\subsection{Examples of dcpos with small compact bases}%
\label{sec:algebraic-examples}

Now that we have the theory of small bases we turn to examples illustrating
small bases in practice. Our examples will involve small \emph{compact} bases
and an example of a dcpo with a small basis that is not compact will have to
wait until \cref{sec:dyadics} when we have developed the ideal completion.

\begin{example}\label{Omega-small-compact-basis}%
  The map \(\beta : \Two \to \Omega_{\U}\) defined by \(0 \mapsto \Zero_{\U}\)
  and \(1 \mapsto \One_{\U}\) is a small compact basis for \(\Omega_{\U}\). In
  particular, \(\Omega_{\U}\) is algebraic.
\end{example}

The basis \(\beta : \Two \to \Omega_{\U}\) defined above has the special
property that it is \emph{dense} in the sense
of~\cite[\mkTTurl{TypeTopology.Density}]{TypeTopology}: its image has empty
complement, i.e.\ the type
\(\Sigma_{P : \Omega_{\U}}\lnot\pa*{\Sigma_{b : \Two}\,\beta(b) = P}\) is empty.

\begin{proof}[Proof of~\cref{Omega-small-compact-basis}]
  By \cref{compact-elements-in-Omega}, every element in the image of \(\beta\)
  is compact. Moreover, \(\Omega_{\U}\) is locally small, so we only need to
  prove that for every \(P : \Omega_{\U}\) the family \(\dset_{\beta} P\) is
  directed with supremum \(P\).
  The domain of the family is inhabited, because \(\beta(0)\) is the least
  element. Semidirectedness also follows easily, since \(\Two\) has only two
  elements for which we have \(\beta(0) \below \beta(1)\).
  Finally, the supremum of \(\dset_{\beta} P\) is obviously below
  \(P\). Conversely, if \(P\) holds, then
  \(\bigsqcup {\dset_{\beta} P} = \One = P\).
  The final claim follows from
  \cref{structural-algebraicity-if-small-compact-basis}.
\end{proof}

\begin{example}\label{lifting-has-small-compact-basis}%
  For a set \(X : \U\), the map \(\beta : \pa*{\One + X} \to \lifting_{\U}(X)\)
  given by \(\inl(\star) \mapsto \bot\) and \(\inr(x) \mapsto \eta(x)\) is a
  small compact basis for \(\lifting_{\U}(X)\). In particular,
  \(\lifting_{\U}(X)\) is algebraic.
\end{example}

Similar to~\cref{Omega-small-compact-basis}, the basis
\(\beta : (\One + X) \to \lifting_{\U}(X)\) defined above is also dense.

\begin{proof}[Proof of~\cref{lifting-has-small-compact-basis}]
  By \cref{compact-elements-in-lifting}, every element in the image of \(\beta\)
  is compact. Moreover, the lifting is locally small, so
  we only need to prove that for every
  partial element \(l\), the family \(\dset_{\beta} l\) is directed with
  supremum~\(l\).
  The~domain of the family is inhabited, because \(\beta(\inl(\star))\) is the
  least element. Semidirectedness also follows easily: First of all,
  \(\beta(\inl(\star))\) is the least element. Secondly, if we have \(x,x' : X\)
  such that \(\beta(\inr(x)),\beta(\inr(x')) \below l\), then because
  \(\beta(\inr(x)) \equiv \eta(x)\) is defined, we must have
  \(\beta(\inr(x)) = l = \beta(\inr(x'))\) by definition of the order.
  Finally, the supremum of \(\dset_{\beta} l\) is obviously a partial element
  below \(l\). Conversely, if \(l\) is defined, then \(l = \eta(x)\) for some
  \(x : X\), and hence, \(l = \eta(x) \below \bigsqcup \dset_{\beta} l\).
  The final claim follows from
  \cref{structural-algebraicity-if-small-compact-basis}.
\end{proof}

\begin{example}\label{powerset-small-compact-basis}%
  For a set \(X : \U\), the map \(\beta : \List(X) \to \powerset_{\U}(X)\) from
  \cref{def:list-to-powerset} (whose image is the type of Kuratowski finite
  subsets of \(X\)) is a small compact basis for \(\powerset_{\U}(X)\). In
  particular, \(\powerset_{\U}(X)\) is algebraic.
\end{example}

\begin{proof}[Proof of~\cref{powerset-small-compact-basis}]
  By \cref{Kuratowski-finite-iff-list,compact-elements-in-powerset}, all
  elements in the image of \(\beta\) are compact. Moreover,
  \(\powerset_{\U}(X)\) is locally small, so we only need to prove that for
  every \(A : \powerset(X)\) the family \(\dset_{\beta} A\) is directed with
  supremum \(A\), but this was also proven in
  \cref{compact-elements-in-powerset}.
  The final claim follows from
  \cref{structural-algebraicity-if-small-compact-basis}.
\end{proof}

At this point the reader may ask whether we have any examples of dcpos which can
be equipped with algebraicity data but that do not have a small compact basis.
The following example shows that this can happen in our predicative setting:

\begin{example}\label{lifting-structurally-algebraic-but-no-small-basis}%
  The lifting \(\lifting_{\V}(P)\) of a proposition \(P : \U\) can be
  given algebraicity data, but has a small compact basis if and only
  if \(P\) is \(\V\)-small. Thus, requiring that \(\lifting_{\V}(P)\)
  has a small basis for every proposition \(P : \U\) is equivalent to
  the propositional resizing principle that every proposition in
  \(\U\) is equivalent to one in \(\V\).
\end{example}

\begin{proof}[Proof of~\cref{lifting-structurally-algebraic-but-no-small-basis}]
  Note that \(\lifting_{\V}(P)\) is simply the type of propositions in \(\V\)
  that imply \(P\). It has algebraicity data, because given such a
  proposition \(Q\), the family
  \begin{align*}
    Q + \One_{\V} &\to \lifting_{\V}(P) \\
    \inl(q) &\mapsto \One_{\V} \\
    \inr(\star) &\mapsto \Zero_{\V}
  \end{align*}
  is directed, has supremum \(Q\) and consists of compact elements.
  But if \(\lifting_{\V}(P)\) had a small compact basis
  \(\beta : B \to \lifting_{\V}(P)\), then we would have
  \(P \simeq \exists_{b : B}\pa*{\beta(b) \simeq \One_{\V}}\) and the latter is
  \(\V\)-small.
  Conversely, if \(P\) is equivalent to \(P_0 : \V\), then \(\lifting_{\V}(P)\) is isomorphic
  to \(\lifting_{\V}(P_0)\), which has a small compact basis by
  \cref{lifting-has-small-compact-basis}.
\end{proof}

\begin{example}\label{ordinals-structurally-continuous-but-no-small-basis}%
  In classical mathematics, it is known~\cite[Proposition~2.6]{JiaEtAl2015} that
  every well-ordered set \(C\) with a top element \(\top\) is an algebraic
  lattice, and every compact element of it is equal to the least element or of
  the form \(c + 1\) for some \(c \in C \setminus \{\top\}\).
  The ordinals in univalent foundations, as introduced
  in~\cite[Section~10.3]{HoTTBook} and further developed by the second author
  in~\cite{TypeTopologyOrdinals}, give a constructive example of a large (even
  impredicatively) algebraic sup-lattice without a small basis.
  In~\cite[Theorem~5.8]{deJongEscardo2023} we showed that the large poset of
  small ordinals has small suprema, so indeed the ordinals form a
  sup-lattice. (This result needs small set quotients.)
  Moreover, they cannot have a small basis, as otherwise we could take the
  supremum of all ordinals in the basis which would yield a greatest ordinal
  which does not exist (as a consequence of~\cite[Lemma~10.3.21]{HoTTBook}).
  It remains to show that the sup-lattice of ordinals is algebraic.
  This follows from the following two facts.
  \begin{enumerate}[(i)]
  \item\label{ordinal-compact} Every successor ordinal, i.e.\ one of the form
    \(\alpha + \One\), is compact.
  \item\label{ordinal-sup-eq} Every ordinal \(\alpha\) is the supremum of the
    family \(x : \alpha \mapsto \alpha \initseg x + \One\), where
    \(\alpha \initseg x\) denotes the ordinal of elements of \(\alpha\) that are
    strictly less than \(x\).
  \end{enumerate}
  While the family in \eqref{ordinal-sup-eq} is not necessarily directed, this
  does not pose a problem, since we can take its directification
  (see~\cref{directification} later) by considering finite joins of elements in
  the family which are necessarily compact again by
  \cref{binary-join-is-compact}.

  For proving these facts, we recall
  from~\cite[\mkTTurl{Ordinals.OrdinalOfOrdinals}]{TypeTopologyOrdinals} that
  the order \(\preceq\) on ordinals can be characterised as follows:
  \[
    \alpha \preceq \beta \quad\iff\quad \forall_{x : \alpha}\exists_{y :
      \beta}\, \alpha \initseg x = \beta \initseg y.
  \]
  We now prove \eqref{ordinal-compact}: Suppose that
  \(\alpha + \One \preceq \bigsqcup_{i : I} \beta_i\).
  Since \(\alpha = (\alpha + \One) \initseg \inr(\star)\), there exists
  \(s : \bigsqcup_{i : I}\beta_i\) with
  \(\alpha = \bigsqcup_{i : I} \beta_i \initseg s\).
  By \cite[Lemma~15]{LICS2023} there exist \(j : I\) and \(b : \beta_j\) with
  \({\bigsqcup_{i : I}\beta_i} \initseg s = \beta_j \initseg b\).
  Hence, \(\alpha = \beta_j \initseg b\), but then it follows that
  \(\alpha + \One \preceq \beta_j\).

  To see that \eqref{ordinal-sup-eq} is true, we first notice that
  \({\alpha \initseg x + \One} \preceq \alpha\) holds (using the
  characterisation of the partial order) so that \(\alpha\) is an upper bound
  for the family.
  Now suppose that \(\beta\) is another upper bound. We need to show that
  \(\alpha \preceq \beta\). So let \(x : \alpha\) be arbitrary. Since \(\beta\)
  is an upper bound of the family, there is \(b : \beta\) with
  \(\alpha \initseg x = (\alpha \initseg x + \One) \initseg \inr \star = \beta
  \initseg b\), so we are done.
\end{example}

\subsection{The canonical basis of compact elements}\label{sec:basis-of-compact-elements}
So far, our development of algebraic dcpos (with small compact bases) has
resulted from specialising the treatment of continuous dcpos with small bases.
In this section we take a closer look at the algebraic case.

Classically, the subset \(K\) of compact elements of an algebraic dcpo \(D\)
forms a basis for \(D\).
In our predicative context, we only consider small bases, and a priori there is
no reason for \(K\) to be a small type.
However, if \(D\) comes equipped with a small compact basis, then set
replacement implies that \(K\) is in fact small.

We recall the \emph{set replacement} principle
from~\cite[Definition~3.27]{deJongEscardo2023}: it asserts that the image of a
map \(f : X \to Y\) is \({\U \sqcup \V}\)-small if \(X\) is \(\U\)-small and
\(Y\) is a locally \(\V\)-small set, where local smallness refers to smallness
of the identity types.
We also recall \cite[Theorem~3.29]{deJongEscardo2023} that set replacement is
logically equivalent to the existence of small set quotients.

If we additionally assume univalence, then the relevant smallness condition is a
property~\cite[Section~2.3]{deJongEscardo2023}, which means that having an
unspecified small compact basis is sufficient.

In particular, with set replacement and univalence, we can show:
\begin{align*}
  \|&D \text{ has a specified small compact basis}\| \to \\
  &D \text{ has a specified small compact basis}
\end{align*}
Using the terminology of \cite[Definition~3.9]{KrausEtAl2017}, we may thus say
that the type expressing that \(D\) has a specified small compact basis has
split support.
This observation is due to Ayberk Tosun (private communication) who also
formalised the result for spectrality in the context of locale theory in
predicative univalent foundations~\cite[Theorem~4.17]{ArrietaEscardoTosun2024}.

\begin{lemma}\label{dir-sup-of-compact-elts-below}
  Every element of an algebraic dcpo is the directed supremum of all
  compact elements below it.
\end{lemma}
\begin{proof}
  Writing \(D\) for the algebraic dcpo, and letting \(x : D\) be arbitrary, we
  have to show that the inclusion family
  \[
    \iota_x : \pa*{\Sigma_{c : D}\pa*{c \text{ is compact}} \times \pa*{c \below x}}
    \to D
  \]
  is directed with supremum \(x\).
  Since this is a proposition, we may assume to be given algebraicity data
  for \(D\).
  Thus, we have a directed family \(\kappa_x : I_x \to D\) of compact elements
  with supremum \(x\).
  By directedness, \(I_x\) is inhabited, so we see that the domain of
  \(\iota_x\) is inhabited too.
  For semidirectedness, assume we have compact elements \(c_1,c_2 : D\) below
  \(x\). Since \(x\) is the directed supremum of \(\kappa_x\), there exist
  elements \(i_1,i_2 : I_x\) with \(c_1 \below \kappa_x(i_1)\) and
  \(c_2 \below \kappa_x(i_2)\).
  By directedness of \(\kappa_x\), there then exists \(i : I_x\) such that
  \(c_1\) and \(c_2\) are both below \(\kappa_x(i)\).
  But \(\kappa_x(i)\) is a compact element below \(x\) so we are done.
  Finally, we show that \(\iota_x\) has supremum \(x\). Clearly, \(x\) is an
  upper bound for \(\iota_x\).
  Now suppose that \(y\) is any other upper bound. It then suffices to show that
  \(\kappa_x(i) \below y\) for all \(i : I_x\). But each \(\kappa_x(i)\) is a
  compact element below \(x\), so this holds since \(y\) is an upper bound for
  \(\iota_x\).
\end{proof}

\begin{lemma}\label{K-is-small'}
  Assuming set replacement, if a \(\V\)-dcpo is equipped with a small compact
  basis, then the subtype of compact elements is \(\V\)-small.
  If we additionally assume univalence, then having an unspecified small compact
  basis suffices.
\end{lemma}
\begin{proof}
  Let \(\beta : B \to D\) be the small compact basis of the \(\V\)-dcpo \(D\).
  Notice that \(\beta\) factors through the subtype \(K\) of compact elements of
  \(D\).
  Moreover, by \cref{small-compact-basis-contains-all-compact-elements} the map
  \(\beta : B \to K\) is surjective.
  Hence, \(K\) is equivalent to the image of \(\beta : B \to D\).
  Now an application of set replacement finishes the proof, since \(B\) is small
  and \(K\) is locally small because it is a subtype of~\(D\) which is locally
  small by antisymmetry,
  \cref{locally-small-if-small-basis,small-and-compact-basis}.

  Assuming univalence, being small is a
  property~\cite[Proposition~2.8]{deJongEscardo2023}, so that the result follows
  from the above and the universal property of the truncation.
\end{proof}

\begin{proposition}
  Assuming univalence and set replacement, the types expressing that a dcpo has
  a specified, resp.\ unspecified, small compact basis are logically equivalent.
\end{proposition}
\begin{proof}
  In one direction, we simply apply the propositional truncation.
  In the other direction, we apply
  \cref{dir-sup-of-compact-elts-below,K-is-small'} to see that
  \[
    {K_s} \simeq K \hookrightarrow D
  \]
  is a small compact basis for the \(\V\)-dcpo \(D\), where \(K\) denotes the
  subtype of compact elements and \(K_s\) is the \(\V\)-small copy of \(K\).
\end{proof}

We note that the above cannot be promoted to an equivalence of types, because
the type of specified small compact bases is not a proposition. This may seem
puzzling because there is a unique basis---\emph{as a
  subset}---which consists of compact elements.
If we had asked for the map \(\beta : B \to D\) in the definition of a small
compact basis to be an embedding, then (ignoring size issues for a moment) the
resulting type is a proposition: it has a unique element in case the dcpo is
algebraic, given by the \emph{subset} of compact elements. (This is true because
any basis must contain all compact elements.)

We illustrate why we do not impose this requirement by revisiting
\cref{powerset-small-compact-basis}. This example showed that the canonical map
from lists into the powerset of a set \(X\) is a small compact basis for the
algebraic dcpo \(\powerset(X)\).
This map is not an embedding, as any permutation of a list will give rise to the
same subset.
If we insisted on having an embedding, we would instead have to use the
inclusion of the Kuratowski finite subsets \(\mathcal K(X)\) into
\(\powerset(X)\).  However, \(\mathcal K(X)\) is not a small type without
additional assumptions, such as HITs or more specifically, set replacement (as
\(\mathcal K(X)\) is precisely the image of the inclusion of lists into the powerset).

Returning to the main line of thought, we conclude that, in the presence of set
replacement and univalence, if there is some unspecified small compact basis,
then the subset of compact elements is small.

\section{The round ideal completion}\label{sec:round-ideal-completion}

We have seen that in continuous dcpos, the basis essentially ``generates'' the
whole dcpo, because the basis suffices to approximate any of its elements.
It is natural to ask whether one can start from a more abstract notion of basis
and ``complete'' it to a continuous dcpo.
Following \citeauthor{AbramskyJung1994}~\cite[Section~2.2.6]{AbramskyJung1994},
but keeping track of size, this is exactly what we do here using the notion of
an \emph{abstract basis} and the \emph{round ideal completion}.

\begin{definition}[Abstract basis]%
  An \emph{abstract \(\V\)-basis} is a type \(B : \V\) with a binary relation
  \({\prec} : B \to B \to \V\) that is proposition-valued, transitive and
  satisfies \emph{nullary} and \emph{binary interpolation}:
  \begin{enumerate}[(i)]
  \item\label{abstract-nullary-interpolation}
    for every \(a : B\), there exists \(b : B\) with \(b \prec a\), and
  \item\label{abstract-binary-interpolation}
    for every \(a_1,a_2 \prec b\), there exists \(a : B\) with
    \({a_1,a_2} \prec a \prec b\).\qedhere
  \end{enumerate}
\end{definition}

\begin{definition}[Ideal, (round) ideal completion, \(\Idl{V}(B,\prec)\)]%
  \hfill%
  \begin{enumerate}[(i)]
  \item A subset \(I : \powerset_{\V}(B)\) of an abstract \(\V\)-basis
    \((B,{\prec})\) is a \emph{\ideal{V}} if it is a directed lower set
    with respect to \({\prec}\).
    That it is a lower set means: if \(b \in I\) and \(a \prec b\), then
    \(a \in I\) too.
  \item We write \(\Idl{V}(B,{\prec})\) for the type of \ideal{V}s of an
    abstract \(\V\)-basis \((B,{\prec})\) and call \(\Idl{V}({B,\prec})\) the
    \emph{(round) ideal completion} of \((B,\prec)\). \qedhere
  \end{enumerate}
\end{definition}

For the remainder of this section, we will fix an abstract \(\V\)-basis
\((B,{\prec})\) and consider its \ideal{V}s.
The name round ideal completion is justified
by~\cref{round-ideals} below.

\begin{definition}[Union of ideals, \(\bigcup \mathcal I\)]%
  \label{def:union-of-ideals}
  Given a family \(\mathcal I : S \to \Idl{V}(B,{\prec})\) of ideals, indexed
  by \(S : \V\), we write
  \[
    \textstyle\bigcup \mathcal I \colonequiv \set{b \in B \mid \exists_{s : S}\pa*{b \in
        \mathcal I_s}}.
  \]
\end{definition}

The following lemma is proved just like in the
classical case~\cite[Section~2.2.6]{AbramskyJung1994}.
\pagebreak[3]\begin{lemma}\label{round-ideals}
  \leavevmode
  \begin{enumerate}[(i)]
  \item If \(\mathcal I : S \to \Idl{V}(B,{\prec})\) is directed, then
    \(\bigcup \mathcal I\) is an ideal.
  \item The round ideal completion is a \(\V\)-dcpo when ordered by subset
    inclusion.
    Paying attention to the universe levels, the ideals form a large but locally
    small \(\V\)-dcpo because \(\Idl{V}(B,{\prec}) : \DCPO{V}{V^+}{V}\).
  \item The ideals of an abstract basis are \emph{round}: for every element
    \(a\) of an ideal~\(I\), there exists \(b \in I\) such that \(a \prec b\).
  \end{enumerate}
\end{lemma}

Roundness makes up for the fact that we have not required an abstract basis to
be reflexive. If it is, then (\cref{sec:round-ideal-completion-reflexive}) the
ideal completion is algebraic.

\begin{definition}[Principal ideal, \(\dset b\)]%
  The \emph{principal ideal} of an element \(b : B\) is defined as the subset
  \(\dset b \colonequiv \set{a \in B \mid a \prec b}\).
  Observe that the principal ideal is indeed an ideal: it is a lower set by
  transitivity of the relation~\({\prec}\), and inhabited and semidirected precisely by
  nullary and binary interpolation, respectively.
\end{definition}

\begin{lemma}\label{principal-ideal-is-monotone}
  The assignment \(b \mapsto \dset b\) is monotone, i.e.\ if \(a \prec b\), then
  \(\dset a \subseteq \dset b\).
\end{lemma}
\begin{proof}
  By transitivity of \({\prec}\).
\end{proof}

\begin{lemma}\label{directed-sup-of-principal-ideals}
  Every ideal is the directed supremum of its principal ideals. That is, for an
  ideal \(I\), the family
  \(\pa*{\Sigma_{b : B}\pa*{b \in I}} \xrightarrow{b \mapsto \dset b}
  \Idl{V}(B,{\prec})\) is directed and has supremum \(I\).
\end{lemma}
\begin{proof}
  Since ideals are lower sets, we have \(\dset b \subseteq I\) for every
  \(b \in I\). Hence, the union \(\bigcup_{b \in I}\dset b\) is a subset of
  \(I\). Conversely, if \(a \in I\), then by roundness of \(I\) there exists
  \(a' \in I\) with \(a \prec a'\), so that \(a \in \bigcup_{b \in I}\dset b\).
  So it remains to show that the family is directed. Notice that it is
  inhabited, because \(I\) is an ideal.
  Now suppose that \(b_1,b_2 \in I\). Since \(I\) is directed, there exists
  \(b \in I\) such that \({b_1,b_2} \prec b\). But this implies
  \({\dset b_1,\dset b_2} \subseteq \dset b\)
  by~\cref{principal-ideal-is-monotone}, so the family is semidirected, as
  desired.
\end{proof}

\begin{lemma}\label{Idl-way-below-characterisation}%
  The following are equivalent for every two ideals \(I\) and \(J\):
  \begin{enumerate}[(i)]
  \item\label{Idl-way-below} \(I \ll J\);
  \item\label{Idl-way-below-1} there exists \(b \in J\) such that
    \(I \subseteq \dset b\);
  \item\label{Idl-way-below-2} there exist \(a \prec b\) such that
    \(I \subseteq \dset a \subseteq \dset b \subseteq J\).
  \end{enumerate}
  In particular, if \(b\) is an element of an ideal \(I\), then
  \(\dset b \ll I\).
\end{lemma}
\begin{proof}
  We show that \eqref{Idl-way-below} \(\Rightarrow\) \eqref{Idl-way-below-1}
  \(\Rightarrow\) \eqref{Idl-way-below-2} \(\Rightarrow\) \eqref{Idl-way-below}.
  So suppose that \(I \ll J\). Then \(J\) is the directed supremum of its
  principal ideals by \cref{directed-sup-of-principal-ideals}. Hence, there
  exists \(b \in J\) such that \(I \subseteq \dset b\) already, which is exactly
  \eqref{Idl-way-below-1}.
  Now suppose that we have \(a \in J\) with \(I \subseteq \dset a\). By
  roundness of \(J\), there exists \(b \in J\) with \(a \prec b\). But then
  \(I \subseteq \dset a \subseteq \dset b \subseteq J\) by
  \cref{principal-ideal-is-monotone} and the fact that \(J\) is a lower set,
  establishing \eqref{Idl-way-below-2}.
  Now suppose that condition \eqref{Idl-way-below-2} holds and that \(J\) is a
  subset of some directed join of ideals \(\mathcal J\) indexed by a type
  \(S : \V\).
  Since \(a \in \dset b \subseteq J\), there exists \(s : S\) such that
  \(a \in \mathcal J_s\). In particular, \(\dset a \subseteq \mathcal J_s\)
  because ideals are lower sets. Hence, if \(a' \in I \subseteq \dset a\), then
  \(a' \in \mathcal J_s\), so \(I \subseteq \mathcal J_s\), which proves that
  \(I \ll J\).

  Finally, if \(b\) is an element of an ideal \(I\), then \(\dset b \ll I\),
  because \eqref{Idl-way-below-1} implies \eqref{Idl-way-below} and
  \(\dset b \subseteq \dset b\) obviously holds.
\end{proof}

\begin{theorem}\label{Idl-has-small-basis}%
  The principal ideals \(\dset{(-)} : B \to \Idl{V}(B,{\prec})\) yield a small
  basis for \(\Idl{V}(B,{\prec})\). In particular, \(\Idl{V}(B,{\prec})\) is
  continuous.
\end{theorem}
\begin{proof}
  First of all, note that the way-below relation on \(\Idl{V}(B,{\prec})\) is
  small-valued because of \cref{Idl-way-below-characterisation}. So it remains
  to show that for every ideal \(I\), the family
  \(\pa*{\Sigma_{b : B}\pa*{\dset b \ll I}} \xrightarrow{b \mapsto \dset b}
  \Idl{V}(B,{\prec})\) is directed with supremum \(I\).
  That the domain of this family is inhabited follows from
  \cref{Idl-way-below-characterisation} and the fact that \(I\) is inhabited.
  For semidirectedness, suppose we have \(b_1,b_2 : B\) with
  \({\dset b_1,\dset b_2} \ll I\). By \cref{Idl-way-below-characterisation}
  there exist \({c_1,c_2} \in I\) such that \(\dset b_1 \subseteq \dset c_1\)
  and \(\dset b_2 \subseteq \dset c_2\).
  Since \(I\) is directed, there exists \(b \in I\) with \({c_1,c_2} \prec b\).
  But now \(\dset b_1 \subseteq \dset c_1 \subseteq \dset b \ll I\) by
  \cref{principal-ideal-is-monotone,Idl-way-below-characterisation} and
  similarly, \(\dset b_2 \subseteq \dset b \ll I\). Hence, the family is
  semidirected, as we wished to show.
  Finally, we show that \(I\) is the supremum of the family. If \(b \in I\),
  then, since \(I\) is round, there exists \(c \in I\) with \(b \prec
  c\). Moreover, \(\dset c \ll I\) by
  \cref{Idl-way-below-characterisation}. Hence, \(b\) is included in the join of
  the family.
  Conversely, if we have \(b : B\) with \(\dset b \ll I\), then
  \(\dset b \subseteq I\), so \(I\) is also an upper bound for the family.
\end{proof}

\subsection{The round ideal completion of a reflexive abstract basis}%
\label{sec:round-ideal-completion-reflexive}
If the relation of an abstract basis is reflexive, then we obtain an algebraic
dcpo, as we show now.

\begin{lemma}\label{abstract-basis-if-reflexive}%
  If \({\prec} : B \to B \to \V\) is proposition-valued, transitive and
  reflexive, then \((B,{\prec})\) is an abstract basis.
\end{lemma}
\begin{proof}
  The interpolation properties for \({\prec}\) are easily proved when it is
  reflexive.
\end{proof}

\begin{lemma}\label{principal-ideal-below-characterisation}
  If an element \(b : B\) is reflexive, i.e.\ \(b \prec b\) holds, then
  \(b \in I\) if and only if \(\dset b \subseteq I\) for every ideal \(I\).
\end{lemma}
\begin{proof}
  The left-to-right implication holds because \(I\) is a lower set and the
  converse holds because \(b \in \dset b\) as \(b\) is assumed to be reflexive.
\end{proof}

\begin{lemma}\label{principal-ideals-are-compact}
  If \(b : B\) is reflexive, then its principal ideal \(\dset b\) is compact.
\end{lemma}
\begin{proof}
  Suppose that we have \(b : B\) such that \(b \prec b\) holds and that
  \(\dset b \subseteq \bigcup \mathcal I\) for some directed family
  \(\mathcal I\) of ideals. By \cref{principal-ideal-below-characterisation}, we
  have \(b \in \bigcup\mathcal I\), which means that there exists \(s\) in the
  domain of \(\mathcal I\) such that \(b \in \mathcal I_s\). Using
  \cref{principal-ideal-below-characterisation} once more, we see that
  \(\dset b \subseteq \mathcal I_s\), proving that \(\dset b\) is compact.
\end{proof}

\begin{theorem}\label{Idl-has-small-compact-basis}%
  If \({\prec}\) is reflexive, then a small \emph{compact} basis for
  \(\Idl{V}(B,{\prec})\) is given by the principal ideals
  \(\dset{(-)} : B \to \Idl{V}(B,{\prec})\). In particular,
  \(\Idl{V}(B,{\prec})\) is algebraic.
\end{theorem}
\begin{proof}
  This follows from
  \cref{Idl-has-small-basis,principal-ideals-are-compact,small-and-compact-basis}.
\end{proof}

\begin{theorem}\label{Idl-mediating-map}
  If \(f : B \to D\) is a monotone map to a \(\V\)-dcpo \(D\), then the map
  \(\bar{f} : \Idl{V}(B,{\prec}) \to D\) defined by taking an ideal \(I\) to the
  supremum of the directed family
  \({f \circ \fst} : \pa*{\Sigma_{b : B}\pa*{b \in I}} \to D\) is Scott
  continuous.
  Moreover, if \({\prec}\) is reflexive, then \(\bar{f}\) is the unique Scott
  continuous map making the diagram
  \[
    \begin{tikzcd}
      B \ar[rr,"f"] \ar[dr,"\dset{(-)}"'] & & D \\
      & \Idl{V}(B,{\prec}) \ar[ur,"\bar{f}"']
    \end{tikzcd}
  \]
  commute.
\end{theorem}
\begin{proof}
  Note that \({f \circ \fst} : \pa*{\Sigma_{b : B}\pa*{b \in I}} \to D\) is
  indeed a directed family, because \(I\) is a directed subset of \(B\) and
  \(f\) is monotone.
  For Scott continuity of \(\bar{f}\), assume that we have a directed family
  \(\mathcal I\) of ideals indexed by \(S : \V\).
  We first show that \(\bar{f}\pa*{\bigcup \mathcal I}\) is an upper bound of
  \(\bar{f} \circ \mathcal I\). So let \(s : S\) be arbitrary and note that
  \(\bar{f}\pa*{\bigcup\mathcal I} \supseteq \bar{f}\pa*{\mathcal I_s}\) as soon
  as \(\bar{f}\pa*{\bigcup \mathcal I} \aboveorder f(b)\) for every
  \(b \in \mathcal I_s\). But for such \(b\) we have
  \(b \in \bigcup \mathcal I\), so this holds.
  Now suppose that \(y\) is an upper bound of \(\bar{f} \circ \mathcal I\). To
  show that \(\bar{f}\pa*{\bigcup\mathcal I} \below y\), it is enough to prove
  that \(f(b) \below y\) for every \(b \in \mathcal I\). But for such \(b\),
  there exists \(s : S\) such that \(b \in \mathcal I_s\) and hence,
  \(f(b) \below \bar{f}\pa*{\mathcal I_s} \below y\).

  Finally, if \({\prec}\) is reflexive, then we prove that
  \(\bar{f}\pa*{\dset b} = f(b)\) for every \(b : B\) by antisymmetry.  Since
  \({\prec}\) is assumed to be reflexive, we have \(b \in \dset b\) and
  therefore, \(f(b) \below \bar{f}\pa*{\dset b}\).
  Conversely, for every \(c \prec b\) we have \(f(c) \below f(b)\) by
  monotonicity of \(f\) and hence, \(\bar{f}(\dset b) \below f(b)\), as desired.
  Uniqueness is proved easily, because if \(g : {\Idl{V}(B,{\prec}) \to D}\) is
  Scott continuous with \(g(\dset b) = f(b)\), then for an arbitrary ideal \(I\)
  we have
  \(g(I) = g(\bigcup_{b \in I}\dset b) = \bigcup_{b \in I}g(\dset b) =
  \bigcup_{b \in I}f(b) \equiv \bar{f}(I)\) by
  \cref{directed-sup-of-principal-ideals} and Scott continuity of \(g\).
\end{proof}

\subsection{Example: the ideal completion of the dyadics rationals}\label{sec:dyadics}

We describe an example of a continuous dcpo, built using the ideal completion,
that is not algebraic. In fact, this dcpo has no compact elements at all.

We inductively define a type and an order representing dyadic rationals
\(m / 2^n\) in the interval \((-1,1)\) for integers \(m,n\).
The~intuition for the upcoming definitions is the following: Start with the
point~\(0\) in the middle of the interval.
Then consider the two functions, respectively standing for \emph{left}
and \emph{right},
\begin{align*}
  l,r &: (-1,1) \to (-1,1) \\
  l(x) &\colonequiv (x-1)/2 \\
  r(x) &\colonequiv (x+1)/2
\end{align*}
that generate the dyadic rationals. Observe that \(l(x) < 0 < r(x)\) for every
\(x : (-1,1)\). Accordingly, we inductively define the following types.

\begin{definition}[{Dyadics}, \(\dyadics\), \(\prec\)]%
  The type of \emph{dyadics} \(\dyadics : \U_0\) is the inductive type with
  these three constructors
  \[
    \dyadicmiddle : \dyadics \quad \dyadicleft : {\dyadics \to \dyadics}
    \quad \dyadicright : {\dyadics \to \dyadics}.
  \]
  We also inductively define \({\prec} : {\dyadics \to \dyadics \to \U_0}\) as
  \begin{alignat*}{6}
    \dyadicmiddle & \prec \dyadicmiddle && \colonequiv \Zero &\quad
    {\dyadicleft (x)} & \prec {\dyadicmiddle} && \colonequiv \One
    &\quad
    {\dyadicright (x)} & \prec {\dyadicmiddle} && \colonequiv \Zero \\
    \dyadicmiddle & \prec {\dyadicleft (y)} && \colonequiv \Zero &
    {\dyadicleft (x)} & \prec {\dyadicleft (y)} && \colonequiv {x \prec y} &
    {\dyadicright (x)} & \prec {\dyadicleft (y)} && \colonequiv \Zero \\
    \dyadicmiddle & \prec {\dyadicright (y)} && \colonequiv \One &
    {\dyadicleft (x)} & \prec {\dyadicright (y)} && \colonequiv \One &
    {\dyadicright (x)} & \prec {\dyadicright (y)} && \colonequiv {x \prec y}.
    \!\!\!\!\qedhere
  \end{alignat*}
\end{definition}

\begin{lemma}
  The type of dyadics is a set with decidable equality.
\end{lemma}
\begin{proof}
  Sethood follows from having decidable equality by Hedberg's
  Theorem.
  To see that \(\dyadics\) has decidable equality, one can use a standard
  inductive proof.
\end{proof}

\begin{definition}[Trichotomy, density, having no endpoints]%
  We say that a binary relation \({<}\) on a type \(X\) is
  \begin{itemize}
  \item \emph{trichotomous} if exactly one of \(x < y\), \(x = y\) or
    \(y < x\) holds.
  \item \emph{dense} if for every \(x,y : X\), there exists some \(z : X\)
    such that \(x < z < y\).
  \item \emph{without endpoints} if for every \(x : X\), there exist some
    \(y,z : X\) with \(y < x < z\).\qedhere
  \end{itemize}
\end{definition}

\begin{lemma}\label{dyadics-order-properties}
  The relation \({\prec}\) on the dyadics is proposition-valued, transitive,
  irreflexive, trichotomous, dense and without endpoints.
\end{lemma}
\begin{proof}
  That \({\prec}\) is proposition-valued, transitive, irreflexive and
  trichotomous is all proven by a straightforward induction on the definition on
  \(\dyadics\).
  That it has no endpoints is witnessed by the fact that for every
  \(x : \dyadics\), we have
  \begin{equation}\label{left-and-right-in-order}\tag{\(\dagger\)}
    {\dyadicleft x} \prec x \prec {\dyadicright x}
  \end{equation}
  which is proven by induction on \(\dyadics\) as well.
  We spell out the inductive proof that it is dense, making use of
  \eqref{left-and-right-in-order}. Suppose that \(x \prec y\). Looking at the
  definition of the order, we see that we need to consider five cases.
  \begin{itemize}
  \item If \(x = \dyadicmiddle\) and \(y = \dyadicright y'\), then we have
    \(x \prec \dyadicright(\dyadicleft(y')) \prec y\).
  \item If \(x = \dyadicleft(x')\) and \(y = \dyadicmiddle\), then we have
    \(x \prec \dyadicleft(\dyadicright(x')) \prec y\).
  \item If \(x = \dyadicleft(x')\) and \(y = \dyadicright y'\), then we have
    \(x \prec \dyadicmiddle \prec y\).
  \item If \(x = \dyadicright (x')\) and \(y = \dyadicright y'\), then we have
    \(x' \prec y'\) and therefore, by induction hypothesis, there exists
    \(z' : \dyadics\) such that \(x' \prec z' \prec y'\).
    Hence, \(x \prec \dyadicright(z') \prec y\).
  \item If \(x = \dyadicleft (x')\) and \(y = \dyadicleft (y')\), then
    \(x' \prec y'\) and so, by induction hypothesis, there exists
    \(z' : \dyadics\) such that \(x' \prec z' \prec y'\).
    Thus, \({x \prec \dyadicleft(z') \prec y}\). \qedhere
  \end{itemize}
\end{proof}

\begin{proposition}\label{dyadics-form-abstract-basis}%
  The pair \((\dyadics,{\prec})\) is an abstract \(\U_0\)-basis.
\end{proposition}
\begin{proof}
  By \cref{dyadics-order-properties} the relation \({\prec}\) is
  proposition-valued and transitive.
  Moreover, that it has no endpoints implies unary interpolation.
  For binary interpolation, suppose that we have \(x \prec z\) and
  \(y \prec z\). Then by trichotomy there are three cases.
  \begin{itemize}
  \item If \(x = y\), then using density and our assumption that \(x \prec z\),
    there exists \(d : \dyadics\) with \(y = x \prec d \prec z\), as desired.
  \item If \(x \prec y\), then using density and our assumption that
    \(y \prec z\), there exists \(d : \dyadics\) with \(y \prec d \prec z\), but
    then also \(x \prec d\) since \(x \prec y\), so we are done.
  \item If \(x \prec y\), then the proof is similar to that of the second case.
    \qedhere
  \end{itemize}
\end{proof}

\begin{proposition}
  The ideal completion \(\Idlnum{0}(\dyadics,{\prec}) : \DCPOnum{0}{1}{0}\) is
  continuous with small basis
  \(\dset (-) : \dyadics \to \Idlnum{0}(\dyadics,{\prec})\). Moreover, it cannot be
  algebraic, because none of its elements are compact.
\end{proposition}
\begin{proof}
  The first claim follows from \cref{Idl-has-small-basis}. Now suppose for a
  contradiction that we have a compact ideal \(I\). By
  \cref{Idl-way-below-characterisation}, there exists \(x \in I\) with
  \(I \subseteq \dset x\). But this implies \(x \prec x\), which is impossible
  as \({\prec}\) is irreflexive.
\end{proof}

\subsection{Ideal completions of small bases}%
\label{sec:ideal-completions-of-small-bases}%

Given a \(\V\)-dcpo \(D\) with a small basis \(\beta : B \to D\), we show that
there are two natural ways of turning \(B\) into an abstract basis. Either
define \(b \prec c\) by \(\beta(b) \ll \beta(c)\), or by
\(\beta (b) \below \beta(c)\).
Taking their \ideal{V} completions we show that the former yields a
continuous dcpo isomorphic to \(D\), while the latter yields an algebraic dcpo
(with a small compact basis) in which \(D\) can be embedded.
The latter fact will find application in
\cref{sec:exponentials-with-small-bases}, while the former gives us a
presentation theorem: every dcpo with a small basis is isomorphic to a dcpo of
ideals. In particular, if \(D : \DCPO{V}{U}{T}\) has a small basis, then it is
isomorphic to a dcpo with simpler universe parameters, namely
\(\Idl{V}\pa*{B,{\ll_{\beta}}} : \DCPO{V}{V^+}{V}\).
Of course a similar result holds for dcpos with a small compact basis. In
studying these variations, it is helpful to first develop some machinery that
all of them have in common.

Fix a \(\V\)-dcpo \(D\) with a small basis \(\beta : B \to D\). In what follows
we conflate the family
\(\ddset_{\beta} x : \pa*{\Sigma_{b : B}\pa*{\beta(b) \ll x}} \xrightarrow{\beta
  \circ \fst} D\) with its associated subset
\({\set{b \in B \mid \beta(b) \ll x}}\), formally given by the map
\(B \to \Omega_{\V}\) defined as
\(b \mapsto \exists_{b : B}\pa*{\beta(b) \ll x}\).

\begin{lemma}\label{ddsets-is-continuous}
  The assignment \(x : D \mapsto \ddset_{\beta} x : \powerset(B)\) is Scott
  continuous.
\end{lemma}
\begin{proof}
  Note that \(\ddset_{\beta}(-)\) is monotone: if \(x \below y\) and \(b : B\)
  is such that \(\beta(b) \ll x\), then also \(\beta(b) \ll y\).  So it suffices
  to prove that
  \(\ddset_{\beta} \pa*{\bigsqcup \alpha} \subseteq \bigcup_{i :
    I}\ddset_{\beta}{\alpha_i}\). So suppose that \(b : B\) is such that
  \(\beta(b) \ll \bigsqcup \alpha\). By \cref{interpolation-basis}, there
  exists \(c : B\) with \(\beta(b) \ll \beta(c) \ll \bigsqcup \alpha\).
  Hence, there exists \(i : I\) such that
  \(\beta(b) \ll \beta(c) \below \alpha_i\) already, and therefore,
  \(b \in \bigcup_{j : J}\ddset_{\beta} {\alpha_j}\), as desired.
\end{proof}

By virtue of the fact that \(\beta\) is a small basis for \(D\), we know that
taking the directed supremum of \(\ddset_{\beta} x\) equals \(x\) for every
\(x : D\). In other words, \(\ddset_{\beta}{(-)}\) is a section of
\(\bigsqcup{(-)}\). The following lemma gives conditions for the other composite
to be an inflation or a deflation.

\begin{lemma}\label{inflation-deflation-criteria}
  Let \(I : \powerset_{\V}(B)\) be a subset of \(B\) such that its associated
  family
  \({\bar{I} : \pa*{\Sigma_{b : B}\pa*{b \in I}} \xrightarrow{\beta \circ \fst}
  D}\) is directed.
  \begin{enumerate}[(i)]
  \item\label{inflation-criterion} If the conjunction of
    \(\beta(b) \below \beta(c)\) and \(c \in I\) implies \(b \in I\), then
    \(\ddset_{\beta} \bigsqcup\bar{I} \subseteq I\).
  \item\label{deflation-criterion} If for every \(b \in I\) there exists
    \(c \in I\) such that \(\beta(b) \ll \beta(c)\), then
    \(I \subseteq \ddset_{\beta} \bigsqcup\bar{I}\).
  \end{enumerate}
  In particular, if both conditions hold, then
  \(I = \ddset_{\beta}\bigsqcup \bar{I}\).
\end{lemma}

Note that the first condition says that \(I\) is a lower set with respect to the
order of \(D\), while the second says that \(I\) is round with respect to the
way-below relation.

\begin{proof}
  \eqref{inflation-criterion} Suppose that \(I\) is a lower set and let \(b : B\)
  be such that \(\beta(b) \ll \bigsqcup \bar{I}\). Then there exists \(c \in I\)
  with \(\beta(b) \below \beta(c)\), which implies \(b \in I\) as desired,
  because \(I\) is assumed to be a lower set.
  \eqref{deflation-criterion} Assume that \(I\) is round and let \(b \in I\) be
  arbitrary. By roundness of \(I\), there exists \(c \in I\) such that
  \(\beta(b) \ll \beta(c)\). But then
  \(\beta(b) \ll \beta(c) \below \bigsqcup\bar{I}\), so that
  \(b \in \ddset_{\beta}\bigsqcup \bar{I}\), as we wished to show.
\end{proof}

\begin{lemma}\label{semidirected-lower-set-criteria}
  Suppose that we have \({\prec} : B \to B \to \V\) and let \(x : D\) be
  arbitrary.
  \begin{enumerate}[(i)]
  \item\label{lowerset-criterion} If \(b \prec c\) implies
    \(\beta(b) \below \beta(c)\) for every \(b,c : B\), then
    \(\ddset_{\beta} x\) is a lower set w.r.t.~\({\prec}\).
  \item\label{semidirected-criterion} If \(\beta(b) \ll \beta(c)\) implies
    \(b \prec c\) for every \(b,c : B\), then \(\ddset_{\beta} x\) is
    semidirected w.r.t.~\(\prec\).
  \end{enumerate}
\end{lemma}
\begin{proof}
  \eqref{lowerset-criterion} This is immediate, because \(\ddset{\beta} x\) is a
  lower set with respect to the order relation on \(D\). %
  \eqref{semidirected-criterion} Suppose that the condition holds and that we have
  \({b_1,b_2} : B\) such that \(\beta(b_1),\beta(b_2) \ll x\). Using binary
  interpolation in the basis, there exist \({c_1,c_2} : B\) with
  \(\beta(b_1) \ll \beta(c_1) \ll x\) and \(\beta(b_2) \ll \beta(c_2) \ll
  x\). Hence, \(c_1,c_2 \in \ddset_{\beta x}\) and moreover, by assumption we
  have \(b_1 \prec c_1\) and \(b_2 \prec c_2\), as desired.
\end{proof}

\subsubsection{Ideal completion with respect to the way-below relation}%

\begin{lemma}\label{way-below-abstract-basis}
  If \(\beta : B \to D\) is a small basis for a \(\V\)-dcpo \(D\), then
  \(\pa*{B,\ll_{\beta}}\) is an abstract \(\V\)-basis where \(b \ll_{\beta} c\)
  is defined as \(\beta(b) \ll \beta(c)\).
\end{lemma}

\begin{remark}
  The definition of an abstract \(\V\)-basis requires the relation on it to be
  \(\V\)-valued. Hence, for the lemma to make sense we appeal to the fact that
  \(\beta\) is a \emph{small} basis which tells us that we can substitute
  \(\beta(b) \ll \beta(c)\) by an equivalent type in \(\V\).
\end{remark}

\begin{proof}[Proof of \cref{way-below-abstract-basis}]
  The way-below relation is proposition-valued and transitive. Moreover,
  \({\ll_{\beta}}\) satisfies nullary and binary interpolation precisely because
  we have nullary and binary interpolation in the basis for the way-below
  relation by \cref{interpolation-basis}.
\end{proof}

The following theorem is a presentation result for dcpos with a small basis:
every such dcpo can be presented as the round ideal completion of its small
basis.
\begin{theorem}\label{Idl-iso-continuous}%
  The map \(\ddset_{\beta}{(-)} : D \to \Idl{V}\pa*{B,{\ll_{\beta}}}\) is an
  isomorphism of \(\V\)-dcpos.
\end{theorem}
\begin{proof}
  First of all, we should check that the map is well-defined, i.e.\ that
  \(\ddset_{\beta} x\) is an \(\pa*{B,{\ll_{\beta}}}\)-ideal. It is an inhabited
  subset by nullary interpolation in the basis and a semidirected lower set
  because the criteria of \cref{semidirected-lower-set-criteria} are satisfied
  when taking \({\prec}\) to be \({\ll_{\beta}}\).
  Secondly, the map \(\ddset_{\beta}{(-)}\) is Scott continuous by
  \cref{ddsets-is-continuous}.

  Now notice that the map \(\beta : \pa*{B,{\ll_\beta}} \to D\) is monotone and
  that the Scott continuous map it induces by \cref{Idl-mediating-map} is
  exactly the map \(\bigsqcup : \Idl{V}\pa*{B,{\ll_\beta}} \to D\) that takes an
  ideal \(I\) to the supremum of its associated directed family
  \(\beta \circ \fst : \pa*{\Sigma_{b : B}\pa*{b \in I}} \to D\).

  Since \(\beta\) is a basis for \(D\), we know that
  \(\bigsqcup{\ddset_{\beta} x} = x\) for every \(x : D\). So it only
  remains to show that \(\ddset_{\beta} \circ \bigsqcup\) is the identity on
  \(\Idl{V}\pa*{B,{\ll_{\beta}}}\), for which we will use
  \cref{inflation-deflation-criteria}.
  So suppose that \(I : \Idl{V}\pa*{B,{\ll_{\beta}}}\) is arbitrary. Then we
  only need to prove that
  \begin{enumerate}[(i)]
  \item\label{to-prove-lower} the conjunction of \(\beta(b) \below \beta(c)\)
    and \(c \in I\) implies \(b \in I\) for every \(b,c : B\);
  \item\label{to-prove-round} for every \(b \in I\), there exists \(c \in I\)
    such that \(\beta(b) \ll \beta(c)\).
  \end{enumerate}
  Note that \eqref{to-prove-round} is just saying that \(I\) is a round ideal
  w.r.t.\ \({\ll_{\beta}}\), so this holds.
  For \eqref{to-prove-lower}, suppose that \(\beta(b) \below \beta(c)\) and
  \(c \in I\). By roundness of \(I\), there exists \(c' \in I\) such that
  \(c \ll_{\beta} c'\). But then \(\beta(b) \below \beta(c) \ll \beta(c')\), so
  that \(b \ll_{\beta} c'\) which implies that \(b \in I\), because ideals are
  lower sets.
\end{proof}%

\subsubsection{Ideal completion with respect to the order relation}

\begin{lemma}\label{below-abstract-basis}
  If \(\beta : B \to D\) is a small basis for a \(\V\)-dcpo \(D\), then
  \(\pa*{B,\below_{\beta}}\) is an abstract \(\V\)-basis where
  \(b \below_{\beta} c\) is defined as \(\beta(b) \below \beta(c)\).
\end{lemma}
\begin{proof}
  The relation \({\below_{\beta}}\) is reflexive, so this follows from
  \cref{abstract-basis-if-reflexive}.
\end{proof}

\begin{remark}
  The definition of an abstract \(\V\)-basis requires the relation on it to be
  \(\V\)-valued. Hence, for the lemma to make sense we appeal to
  \cref{locally-small-if-small-basis} to know that \(D\) is locally small which
  tells us that we can substitute \(\beta(b) \below \beta(c)\) by an equivalent
  type in \(\V\).
\end{remark}

\begin{theorem}\label{Idl-iso-algebraic}
  The map \(\ddset_{\beta}{(-)} : D \to \Idl{V}\pa*{B,{\below_{\beta}}}\) is the
  embedding in an embedding-projection pair.
  In particular, \(D\) is a retract of the algebraic dcpo
  \(\Idl{V}\pa*{B,\below_{\beta}}\) that has a small compact basis.
  Moreover, if \(\beta\) is a small \emph{compact} basis, then the map is an
  isomorphism.
\end{theorem}
\begin{proof}
  First of all, we should check that the map is well-defined, i.e.\ that
  \(\ddset_{\beta} x\) is an \(\pa*{B,{\below_{\beta}}}\)-ideal. It is an
  inhabited subset by nullary interpolation in the basis and a semidirected
  lower set because the criteria of \cref{semidirected-lower-set-criteria} are
  satisfied when taking \({\prec}\) to be \({\below_{\beta}}\).
  Secondly, the map \(\ddset_{\beta}{(-)}\) is Scott continuous by
  \cref{ddsets-is-continuous}.

  Now notice that the map \(\beta : \pa*{B,{\below_\beta}} \to D\) is monotone
  and that the continuous map it induces by \cref{Idl-mediating-map} is
  exactly the map \({\bigsqcup : \Idl{V}\pa*{B,{\below_\beta}} \to D}\) that takes
  an ideal \(I\) to the least upper bound of its associated directed family
  \({\beta \circ \fst : {\pa*{\Sigma_{b : B}\pa*{b \in I}} \to D}}\).

  Since \(\beta\) is a basis for \(D\), we know that
  \(\bigsqcup{\ddset_{\beta} x} = x\) for every \(x : D\). So it only remains to
  show that \(\ddset_{\beta} \circ \bigsqcup\) is a deflation, for which we will
  use \cref{inflation-deflation-criteria}.
  So suppose that \(I : \Idl{V}\pa*{B,{\below_{\beta}}}\) is arbitrary. Then we
  only need to prove that the conjunction of \(\beta(b) \below \beta(c)\) and
  \(c \in I\) implies \(b \in I\), but this holds, because \(I\) is a lower set
  with respect to \({\below_{\beta}}\).

  Finally, assume that \(\beta\) is a small compact basis. We show that
  \(\ddset_{\beta} \circ \bigsqcup\) is also inflationary in this case. So let
  \(I\) be an arbitrary ideal. By \cref{inflation-deflation-criteria} it is
  enough to show that for every \(b \in I\), there exists \(c \in I\) such that
  \(\beta(b) \ll \beta(c)\). But by assumption, \(\beta(b)\) is compact, so we
  can simply take \(c\) to be \(b\).
\end{proof}

Combining \cref{small-basis-closed-under-continuous-retracts}, \cref{Idl-has-small-basis,Idl-iso-continuous}, and
\cref{Idl-iso-algebraic,Idl-has-small-compact-basis}, we obtain the following
result:
\begin{corollary}\hfill%
  \begin{enumerate}[(i)]
  \item A \(\V\)-dcpo has a small basis if and only if it is isomorphic
    to \(\Idl{V}(B,{\prec})\) for an abstract basis \((B,{\prec})\).
  \item A \(\V\)-dcpo has a small compact basis if and only if it is
    isomorphic to \(\Idl{V}(B,{\prec})\) for an abstract basis \((B,{\prec})\)
    where \({\prec}\) is reflexive.
  \item A \(\V\)-dcpo has a small basis if and only if it is a retract of a
    \(\V\)-dcpo with a small compact basis.

    Hence every continuous \(\V\)-dcpo with a small basis is a retract of some algebraic \(\V\)-dcpo.
  \end{enumerate}
  In particular, every \(\V\)-dcpo with a small basis is isomorphic to one whose
  order takes values in \(\V\) and whose carrier lives in \(\V^+\).
\end{corollary}

\section{Bilimits and exponentials}\label{sec:continuous-bilimits-and-exponentials}
\subsection{Structurally continuous and algebraic bilimits}%
\label{sec:structurally-continuous-and-algebraic-bilimits}%

We show that bilimits are closed under equipment with continuity/algebraicity
data.
For the reminder of this section, fix a directed diagram of \(\V\)-dcpos
\((D_i)_{i : I}\) with embedding-projection pairs
\(\pa*{\varepsilon_{i,j},\pi_{i,j}}_{i \below j \text{ in } I}\) between them
as in~\cref{sec:bilimits}.
We stress that, throughout this section, the word ``embedding'' is
only used in the domain-theoretic sense, i.e.\ it is reserved for one
half of an embedding-projection pair, rather than in the homotopy type
theory sense of having subsingleton fibers. But notice that
domain-theoretic embeddings are homotopy embeddings, because they are
sections, and sections of sets are always homotopy embeddings~\cite{Shulman2016}.

Now suppose that for every \(i : I\), we have \(\alpha_i : J_i \to D_i\) with
each \(J_i : \V\). Then we define \(J_\infty \colonequiv \Sigma_{i : I}J_i\) and
\(\alpha_\infty : J_\infty \to D_\infty\) by
\((i,j) \mapsto \varepsilon_{i,\infty}(\alpha_i(j))\), where
\(\varepsilon_{i,\infty}\) is as in~\cref{epsilon-infty}.

\begin{lemma}\label{infty-family-directed-sup}
  If every \(\alpha_i\) is directed and we have \(\sigma : D_\infty\) such that
  \(\alpha_i\) approximates \(\sigma_i\), then \(\alpha_\infty\) is directed and
  approximates \(\sigma\).
\end{lemma}
\begin{proof}
  Observe that \(\alpha_\infty\) is equal to the supremum, if it exists, of the
  directed families \(\pa*{\varepsilon_{i,\infty} \circ \alpha_i}_{i : I}\) in
  the ind-completion of \(D_\infty\), cf.\ the proof of
  \cref{ind-completion-is-directed-complete}.
  Hence, for directedness of \(\alpha_\infty\), it suffices to prove that the
  family \(i \mapsto \varepsilon_{i,\infty} \circ \alpha_i\) is directed with
  respect to the exceeds-relation.
  The index type \(I\) is inhabited, because we are working with a directed
  diagram of dcpos. For semidirectedness, we will first prove that if
  \(i \below i'\), then \(\varepsilon_{i',\infty} \circ \alpha_{i'}\) exceeds
  \(\varepsilon_{i,\infty} \circ \alpha_i\).

  So suppose that \(i \below i'\) and \(j : J_i\). As \(\alpha_{i}\)
  approximates \(\sigma_{i}\), we have \(\alpha_{i}(j) \ll \sigma_i\).
  Because \(\varepsilon_{i,i'}\) is an embedding, it preserves the way-below
  relation (\cref{embedding-preserves-and-reflects-way-below}), so that we get
  \(\varepsilon_{i,i'}(\alpha_i(j)) \ll \varepsilon_{i,i'}(\sigma_i) \below
  \sigma_{i'} = \bigsqcup \alpha_{i'}\). Hence, there exists \(j' : J_{i'}\)
  with \(\varepsilon_{i,i'}(\alpha_i(j)) \below \alpha_{i'}(j')\) which yields
  \(\varepsilon_{i,\infty}(\alpha_i(j)) =
  \varepsilon_{i',\infty}\pa*{\varepsilon_{i,i'}(\alpha_i(j))} \below
  \varepsilon_{i',\infty}(\alpha_{i'}(j'))\), completing the proof that
  \(\varepsilon_{i',\infty} \circ \alpha_{i'}\) exceeds
  \(\varepsilon_{i,\infty} \circ \alpha_i\).

  Now to prove that the family
  \(i \mapsto \varepsilon_{i,\infty} \circ \alpha_i\) is semidirected with
  respect to the exceeds-relation, suppose we have \(i_1,i_2 : I\). Since \(I\)
  is a directed preorder, there exists \(i : I\) such that \(i_1,i_2 \below
  i\). But then \(\varepsilon_{i,\infty} \circ \alpha_i\) exceeds both
  \(\varepsilon_{i_1,\infty} \circ \alpha_{i_1}\) and
  \(\varepsilon_{i_2,\infty} \circ \alpha_{i_2}\) by the above.

  Thus, \(\alpha_\infty\) is directed. To see that its supremum is \(\sigma\),
  observe that
  \begin{align*}
    \sigma &= \textstyle\bigsqcup_{i : I} \varepsilon_{i,\infty}(\sigma_i)
    &&\text{(by \cref{sigma-sup-of-epsilon-pis})} \\
    &= \textstyle\bigsqcup_{i : I} \varepsilon_{i,\infty}\pa*{\textstyle\bigsqcup \alpha_i}
    &&\text{(since \(\alpha_i\) approximates \(\sigma_i\))} \\
    &= \textstyle\bigsqcup_{i : I}\bigsqcup \varepsilon_{i,\infty} \circ \alpha_i
    &&\text{(by Scott continuity of \(\varepsilon_{i,\infty}\))} \\
    &= \textstyle\bigsqcup_{(i,j) : J_\infty} \alpha_\infty(i,j),
  \end{align*}
  as desired.

  Finally, we wish to show that \(\alpha_\infty(i,j) \ll \sigma\) for every
  \((i,j) : J_\infty\). But \(\varepsilon_{i,\infty}\) is an embedding and
  therefore preserves the way-below relation while \(\alpha_i(j)\)
  approximates~\(\sigma_i\), so we get
  \(\alpha_\infty(i,j) \equiv \varepsilon_{i,\infty}(\alpha_i(j)) \ll
  \varepsilon_{i,\infty}(\sigma_i) \below \sigma\) where the final inequality
  holds because \(\varepsilon_{i,\infty} \circ \pi_{i,\infty}\) is a deflation.
\end{proof}

\begin{lemma}\label{infty-family-compact}
  If \(\alpha_i(j)\) is compact for every \(i : I\) and \(j : J_i\), then all
  the values of \(\alpha_\infty\) are compact too.
\end{lemma}
\begin{proof}
  Let \((i,j) : J_\infty\) be arbitrary. Since \(\varepsilon_{i,\infty}\) is an
  embedding it preserves compact elements, so
  \(\alpha_\infty(i,j) \equiv \varepsilon_{i,\infty}(\alpha_i(j))\) is compact.
\end{proof}

\begin{theorem}\label{structurally-continuous-bilimit}%
  If each \(D_i\) comes equipped with continuity (resp.\ algebraicity) data,
  then we can give continuity (resp.\ algebraicity) data for \(D_\infty\).
\end{theorem}
\begin{proof}
  Let \(\sigma : D_\infty\) be arbitrary. By assumption on each \(D_i\), we have
  a directed family \(\alpha_i : J_i \to D_i\) approximating
  \(\sigma_i\). Hence, by \cref{infty-family-directed-sup}, the family
  \(\alpha_\infty\) is directed and approximates \(\sigma\), giving continuity
  data for \(D_\infty\).
  For the algebraic case, we apply the above and \cref{infty-family-compact}.
\end{proof}

Note that we do not expect to be able to prove that \(D_\infty\) is continuous
if each \(D_i\) is, because it would require an instance of the axiom of choice
to get continuity data on each \(D_i\), and without those we have
nothing to operate on.

\begin{theorem}\label{bilimit-has-small-basis}%
  If each \(D_i\) has a small basis \(\beta_i : B_i \to D_i\), then the map
  \begin{eqnarray*}
    \beta_\infty : \pa*{B_\infty \colonequiv \Sigma_{i : I}B_i} & \to & D_\infty \\
    (i,b) & \mapsto & \varepsilon_{i,\infty}(\beta_i(b))
  \end{eqnarray*}
  is a small basis for \(D_\infty\).
  Furthermore, if each \(\beta_i\) is a small compact basis, then
  \(\beta_\infty\) is a small compact basis too.
\end{theorem}
\begin{proof}
  First of all, we must show that the proposition
  \(\beta_\infty(i,b) \ll \sigma\) is small for every \(i : I\), \(b : B_i\) and
  \(\sigma : D_\infty\). This is the case as the way-below
  relation on \(D_\infty\) has small values. Indeed, by
  \cref{cont-loc-small-iff-way-below-small} and
  \cref{structurally-continuous-bilimit}, it suffices to prove that \(D_\infty\)
  is locally small. But this holds by \cref{locally-small-bilimit}
  as each
  \(D_i\) is locally small by \cref{locally-small-if-small-basis}.

  It remains to prove that, for an arbitrary element \(\sigma : D_\infty\), the
  family \(\ddset_{\beta_\infty} \sigma\) given by
  \(\pa*{\Sigma_{(i,b) : B_\infty}\beta_\infty(i,b) \ll \sigma}
  \xrightarrow{\beta_\infty \circ \fst} D_\infty\) is directed with supremum
  \(\sigma\).
  Note that for every \(i : I\) and \(b : B_i\), we have that
  \(\beta_i(b) \ll \sigma_i\) implies
  \[
    \beta_\infty(i,b) \equiv \varepsilon_{i,\infty}(\beta_i(b)) \ll
    \varepsilon_{i,\infty}(\sigma_i) \below \sigma,
  \]
  since \cref{embedding-preserves-and-reflects-way-below} tells us that the
  embedding \(\varepsilon_{i,\infty}\) preserves the way-below relation.
  Hence, the identity map induces a well-defined map
  \[
    \iota : \pa*{\Sigma_{i : I}\Sigma_{b :B_i}\beta_i(b) \ll \sigma_i} \to
    \pa*{\Sigma_{(i,b) : B_\infty}\beta_\infty(i,b) \ll \sigma}.
  \]
  \cref{subbasis-lemma} now tells us that we only need to show that
  \(\ddset_{\beta_\infty}\sigma \circ \iota\) is directed and has supremum
  \(\sigma\).
  But if we write
  \(\alpha_i : \pa*{\Sigma_{b : B_i}\beta_i(b) \ll \sigma_i} \to D_i\) for the
  map \(b \mapsto \beta_i(b)\), then we see that
  \(\ddset_{\beta_\infty} \sigma \circ \iota\) is given by \(\alpha_\infty\)
  defined from $\alpha_i$ as in the start of this section. But then \(\alpha_\infty\) is indeed
  seen to be directed with supremum \(\sigma\) by virtue of
  \cref{infty-family-directed-sup} and the fact that \(\alpha_i\) approximates
  \(\sigma_i\).

  Finally, if every \(\beta_i\) is a small compact basis, then \(\beta_\infty\)
  is also a small compact basis because by \cref{small-and-compact-basis} all we
  need to know is that
  \(\beta_\infty(i,b) \equiv \varepsilon_{i,\infty}(\beta_i(b))\) is compact for
  every \(i : I\) and \(b : B_i\). But this follows from the fact that
  embeddings preserve compactness and that each \(\beta_i(b)\) is compact.
\end{proof}%

\subsection{Exponentials with small (compact) bases}%
\label{sec:exponentials-with-small-bases}%

Just as in the classical, impredicative setting, the exponential of two
continuous dcpos need not be continuous~\cite{Jung1988}.
However, with some work, we are able to show that \(E^D\) has a small basis
provided that both \(D\) and \(E\) do and that \(E\) has all (not necessarily
directed) \(\V\)-suprema, that is, $E$ is a continuous lattice.
We first establish this for small compact bases using \emph{step functions} and
then derive the result for compact bases using \cref{Idl-iso-algebraic}.

\subsubsection{Single step functions}

Suppose that we have a dcpo \(D\) and a pointed dcpo
\(E\). Classically~\cite[Exercise~II-2.31]{GierzEtAl2003}, the single step
function given by \(d : D\) and \(e : E\) is defined as
\begin{align*}
  \ssf{d}{e} : D &\to E \\
  x &\mapsto
  \begin{cases}
    e &\text{if \(d \below x\);} \\
    \bot &\text{otherwise}.
  \end{cases}
\end{align*}

Constructively, we can't expect to make this case distinction in general, so we define
single step functions using subsingleton suprema instead.

\begin{definition}[Single step function, \(\ssf{d}{e}\)]%
  The \emph{single step function} given by two elements \(d : D\) and \(e : E\),
  where \(D\) is a locally small \(\V\)-dcpo and \(E\) is a pointed \(\V\)-dcpo,
  is the function \(\ssf{d}{e} : D \to E\) that maps a given \(x : D\) to the
  supremum of the family indexed by the subsingleton \(d \below x\) that is
  constantly~\(e\).

  Note that we need the domain \(D\) to be locally small, because we need the type \(d \below x\) to
  be a subsingleton in \(\V\) to use the \(\V\)-directed-completeness of \(E\).
  For the definition of \(\ssf{d}{e}\) to make sense, we need to know that the
  supremum of the constant family exists. This is the case by
  \cref{pointed-dcpos-sups}, which says that the supremum of a
  subsingleton-indexed family \(\alpha : {P \to E}\) is given by the supremum of
  the directed family \(\One + P \to E\) defined by \(\inl(\star) \mapsto \bot\)
  and \(\inr(p) \mapsto \alpha(p)\).
\end{definition}

\begin{lemma}\label{single-step-function-continuous-if-compact}
  If \(d : D\) is compact, then \(\ssf{d}{e}\) is Scott continuous for all
  \(e : E\).
\end{lemma}
\begin{proof}
  Suppose that \(d : D\) is compact and that \(\alpha : I \to D\) is a directed
  family. We first show that \(\ssf{d}{e}\pa*{\bigsqcup \alpha}\) is an upper
  bound of \(\ssf{d}{e} \circ \alpha\). So let \(i : I\) be arbitrary. Then we
  have to prove
  \(\bigsqcup_{d \below \alpha_i} e \below \bigsqcup {\ssf{d}{e} \circ
    \alpha}\). Since the supremum gives a lower bound of the upper bounds, it
  suffices to prove that \(e \below \bigsqcup {\ssf{d}{e} \circ \alpha}\)
  whenever \(d \below \alpha_i\). But in this case we have
  \(e = \ssf{d}{e}(\alpha_i) \below \bigsqcup {\ssf{d}{e} \circ \alpha}\),
  so we are done.

  To see that \(\ssf{d}{e}\pa*{\bigsqcup \alpha}\) is a lower bound of the upper
  bounds, suppose that we are given \(y : E\) such that \(y\) is an upper bound
  of \(\ssf{d}{e} \circ \alpha\). We are to prove that
  \(\pa*{\bigsqcup_{d \below \bigsqcup \alpha} e} \below y\).
  Note that it suffices for \(d \below \bigsqcup \alpha\) to imply \(e \below y\).
  So assume that \(d \below \bigsqcup \alpha\). By compactness of \(d\) there
  exists \(i : I\) such that \(d \below \alpha_i\) already. But then
  \(e = \ssf{d}{e}(\alpha_i) \below y\), as desired.
\end{proof}

\begin{lemma}\label{above-single-step-function}%
  A Scott continuous function \(f : D \to E\) is above the single step function
  \(\ssf{d}{e}\) with \(d : D\) compact if and only if \(e \below f(d)\).
\end{lemma}
\begin{proof}
  Suppose that \(\ssf{d}{e} \below f\). Then \(\ssf{d}{e}(d) = e \below f(d)\),
  proving one implication.
  Now assume that \(e \below f(d)\) and let \(x : D\) be arbitrary. To prove
  that \(\ssf{d}{e}(x) \below f(x)\), it suffices that \(e \below f(x)\)
  whenever \(d \below x\). But if \(d \below x\), then
  \(e \below f(d) \below f(x)\) by monotonicity of \(f\).
\end{proof}

\begin{lemma}\label{single-step-function-is-compact}
  If \(d\) and \(e\) are compact, then so is \(\ssf{d}{e}\) in the exponential
  \(E^D\).
\end{lemma}
\begin{proof}
  Suppose that we have a directed family \(\alpha : I \to E^D\) such that
  \(\ssf{d}{e} \below \bigsqcup \alpha\). Then we consider the directed family
  \(\alpha^d : I \to E\) given by \(i \mapsto \alpha_i(d)\).
  We claim that \(e \below \bigsqcup \alpha^d\). Indeed, by
  \cref{above-single-step-function} and our assumption that
  \(\ssf{d}{e} \below \bigsqcup\alpha\) we get
  \(e \below \pa*{\bigsqcup\alpha}(d) = \bigsqcup \alpha^d\).
  Now by compactness of \(e\), there exists \(i : I\) such that
  \(e \below \alpha^d(i) \equiv \alpha_i(d)\) already. But this implies
  \(\ssf{d}{e} \below \alpha_i\) by \cref{above-single-step-function} again,
  finishing the proof.
\end{proof}

\subsubsection{Exponentials with small compact bases}

Fix \(\V\)-dcpos \(D\) and \(E\) with small compact bases
\(\beta_D : B_D \to D\) and \(\beta_E : B_E \to D\) and moreover assume that
\(E\) has suprema for all (not necessarily directed) families indexed by types
in \(\V\). We are going to construct a small compact basis on the exponential
\(E^D\).%

\begin{lemma}\label{sup-of-single-step-functions}
  If \(E\) is sup-complete, then every continuous function
  \(f : D \to E\) is the supremum of the collection of single step functions
  \(\pa*{\ssf{\beta_D(b)}{\beta_E(c)}}_{b : B_D , c : B_E}\) that are below
  \(f\).
\end{lemma}
\begin{proof}
  Note that \(f\) is an upper bound by definition, so it remains to prove that
  it is the least. Therefore suppose we are given an upper bound
  \(g : D \to E\). We have to prove that \(f(x) \below g(x)\) for every
  \(x : D\), so let \(x : D\) be arbitrary. Now
  \(x = \bigsqcup \ddset_{\beta_D} x\), because \(\beta_D\) is a small compact
  basis for \(D\), so by Scott continuity of \(f\) and \(g\), it suffices to
  prove that \(f(\beta_D(b)) \below g(\beta_D(b))\) for every \(b : B_D\).
  So let \(b : B_D\) be arbitrary. Since \(\beta_E\) is a small compact basis
  for \(E\), we have
  \(f(\beta_D(b)) = \bigsqcup \dset_{\beta_E} f(\beta_D(b))\). So to prove
  \(f(\beta_D(b)) \below g(\beta_D(b))\) it is enough to know that
  \(\beta_E(c) \below g(\beta_D(b))\) for every \(c : B_E\) with
  \(\beta_E(c) \below f(\beta_D(b))\).
  But for such \(c : B_E\) we have \(\ssf{\beta_D(b)}{\beta_E(c)} \below f\) and
  therefore \(\ssf{\beta_D(b)}{\beta_E(c)} \below g\) because \(f\) is an upper
  bound of such single step functions, and hence
  \(\beta_E(c) \below g(\beta_D(b))\) by \cref{above-single-step-function}, as
  desired.
\end{proof}

\begin{definition}[Directification]\label{directification}%
  In a poset \(P\) with finite joins, the \emph{directification} of a family
  \(\alpha : I \to P\) is the family \(\bar\alpha : \List(I) \to P\) that maps a
  finite list to the join of its elements.
  It is clear that \(\bar\alpha\) has the same supremum as \(\alpha\), and by
  concatenating lists, one sees that the directification yields a directed
  family, hence the name.
\end{definition}

\begin{lemma}\label{directification-is-compact}
  If each element of a family into a sup-complete dcpo is compact, then so are
  all elements of its directification.
\end{lemma}
\begin{proof}
  The supremum of the empty list is \(\bot\) by
  \cref{least-element-is-compact}, and hence the join of finitely many
  compact elements is compact by~\cref{binary-join-is-compact}.
\end{proof}

Let us write \(\sigma : B_D \times B_E \to E^D\) for the map that takes
\((b,c)\) to the single step function \(\ssf{\beta_D(b)}{\beta_E(c)}\) and
\[\beta : B \colonequiv \List(B_D \times B_E) \to E^D\] for its directification,
which exists because \(E^D\) is \(\V\)-sup-complete as \(E\) is and suprema are
calculated pointwise.

\begin{theorem}\label{exponential-has-small-compact-basis}%
  The map \(\beta\) is a small compact basis for the exponential \(E^D\), where
  \(E\) is assumed to be sup-complete.
\end{theorem}
\begin{proof}
  Firstly, every element in the image of \(\beta\) is compact by
  \cref{directification-is-compact,single-step-function-is-compact}.
  Secondly, for every \(b : B\) and Scott continuous map
  \(f : D \to E\), the type \(\beta(b) \below f\) is small, because
  \(E^D\) is locally small by \cref{exponential-is-locally-small}.
  Thirdly, for every such \(f\), the family
  \[\pa*{\Sigma_{b : B}\pa*{\beta(b) \below f}} \xrightarrow{{\beta}
    \circ {\fst}} E^D\] is directed because \(\beta\) is the
  directification of~\(\sigma\).  Finally, this family has supremum
  \(f\) by~\cref{sup-of-single-step-functions}.
\end{proof}

\subsubsection{Exponentials with small bases}

We now present a variation of~\cref{exponential-has-small-compact-basis} but for
(sup-complete) dcpos with small bases. In fact, we will prove it using
\cref{exponential-has-small-compact-basis} and the theory of retracts
(\cref{Idl-iso-algebraic} in particular).

\begin{definition}[Closure under finite joins]%
  A small basis \(\beta : B \to D\) for a sup-complete poset is \emph{closed
    under finite joins} if we have \(b_\bot : B\) with \(\beta(b_\bot) = \bot\) and a
  map \({\vee} : B \to B \to B\) such that
  \(\beta(b \vee c) = \beta(b) \vee \beta(c)\) for every \(b,c : B\).
\end{definition}

\begin{lemma}\label{close-basis-under-finite-joins}
  If \(D\) is a sup-complete dcpo with a small basis \(\beta : B \to D\), then
  the directification of \(\beta\) is also a small basis for \(D\). Moreover, by
  construction, it is closed under finite joins.
\end{lemma}
\begin{proof}
  Since \(\beta\) is a small basis for \(D\), it follows by
  \cref{locally-small-if-small-basis} that the way-below relation on \(D\) is
  small-valued. Hence, writing \(\bar\beta\) for the directification of
  \(\beta\), it remains to prove that \(\ddset_{\bar\beta} x\) is directed with
  supremum \(x\) for every \(x : D\). But this follows easily from
  \cref{subbasis-lemma}, because \(\ddset_\beta x\) is directed with supremum
  \(x\) and this family is equal to the composite
  \[\pa*{\Sigma_{b : B}\pa*{\beta(b) \ll x}} \xhookrightarrow{b \mapsto [b]}
  \pa*{\Sigma_{l : \List(B)}\pa*{\bar\beta(l) \ll x}} \xrightarrow{{\bar\beta}\,
    \circ \,\fst} D,\]
  where \([b]\) denotes the singleton list.
\end{proof}

\begin{lemma}\label{sup-complete-ideal-completion}
  If \(D\) is a \(\V\)-sup-complete poset with a small basis \(\beta : B \to D\)
  closed under finite joins, then the ideal-completion
  \(\Idl{V}(B,\below)\) is \(\V\)-sup-complete too.
\end{lemma}
\begin{proof}
  Since the \(\V\)-ideal completion is \(\V\)-directed complete, it suffices to
  show that \(\Idl{V}(B,\below)\) has finite joins, because then we can turn an
  arbitrary small family into a directed one via~\cref{directification}.
  As \(\beta : B \to D\) is closed under finite joins, we have \(b_\bot : B\)
  with \(\beta(b_\bot) = \bot\) and we easily see that the subset
  \(\set{b \in B \mid \beta(b) \below \beta(b_\bot)}\) is the least element of
  \(\Idl{V}(B,\below)\).
  Now suppose that we have two ideals \(I , J : \Idl{V}(B,\below)\) and consider
  the subset $K$ defined as
  \[
    K \colonequiv \set{b \in B \mid \exists_{c \in I}\,\exists_{d \in
        J}\,(\beta(b) \below \beta(c \vee d))}.
  \]
  Notice that \(K\) is clearly a lower set.
  Also note that ideals are closed under finite joins as they are directed lower sets.
  Hence, \(b_\bot \in I\) and \(b_\bot \in J\), so that \(b_\bot \in K\) and
  \(K\) is inhabited.
  For semidirectedness, assume \({b_1,b_2} \in K\) so that there exist
  \(c_1,c_2 \in I\) and \(d_1,d_2 \in J\) with
  \(\beta(b_1) \below \beta(c_1 \vee d_1)\) and
  \(\beta(b_2) \below \beta(c_2 \vee d_2)\).
  Then \({b_1 \vee b_2} \in K\), because \({c_1 \vee c_2} \in I\) and
  \({d_1 \vee d_2} \in J\) and
  \(\beta({b_1 \vee b_2}) \below \beta({({c_1 \vee c_2})} \vee {({d_1 \vee
      d_2})})\).
  Hence, \(K\) is a directed lower set.
  We claim that \(K\) is the join of \(I\) and \(J\). First of all,
  \(I\)~and~\(J\) are both subsets of \(K\), so it remains to prove that \(K\)
  is the least upper bound. To this end, suppose that we have an ideal \(L\)
  that includes \(I\) and \(J\), and let \(b \in K\) be arbitrary.
  Then there exist \(c \in I\) and \(d \in J\) with
  \(\beta(b) \below \beta({c \vee d})\).
  But \({I,J} \subseteq L\) and ideals are closed under finite joins, so
  \({c \vee d} \in L\), which implies that \(b \in L\) since \(L\) is a lower
  set.
  Therefore, \(K\) is the least upper bound of \(I\) and \(J\), completing the
  proof.
\end{proof}

\begin{theorem}\label{exponentials-with-small-bases-via-retract}%
  The exponential \(E^D\) of dcpos has a specified small basis if
  \(D\)~and~\(E\) have specified small bases and \(E\) is sup-complete.
\end{theorem}
\begin{proof}
  Suppose that \(\beta_D : B_D \to D\) and \(\beta_E : B_E \to E\) are small
  bases and that \(E\) is sup-complete.  By
  \cref{close-basis-under-finite-joins} we can assume that
  \(\beta_E : B_E \to E\) is closed under finite joins.
  We will write \(\overline{D}\) and \(\overline{E}\) for the
  respective ideal completions \(\Idl{V}{(B_D,\below)}\) and
  \(\Idl{V}(B_E,\below)\).
  Then~\cref{Idl-iso-algebraic} tells us that we have retracts
  \[
    \text{\(\retractalt{D}{\overline{D}}{s_D}{r_D}\) and \(\retractalt{E}{\overline{E}}{s_E}{r_E}\).}
  \]
  Composition yields a retract \[\retract{E^D}{{\overline{E}}^{\overline{D}}}\]
  where \(s(f) \colonequiv s_E \circ f \circ r_D\) and
  \(r(g) \colonequiv r_E \circ g \circ s_D\).
  Now \(\overline{D}\) and \(\overline{E}\) have small compact basis
  by~\cref{Idl-has-small-compact-basis} and \(\overline{E}\) is sup-complete by
  \cref{sup-complete-ideal-completion}. Therefore, the exponential of
  \(\overline{D}\) and \(\overline{E}\) has a small
  basis by \cref{exponential-has-small-compact-basis}.
  Finally, \cref{small-basis-closed-under-continuous-retracts} tells us that the
  retraction \(r\) yields a small basis on \(E^D\), as desired.
\end{proof}

Note how, unlike~\cref{exponential-has-small-compact-basis}, the above theorem
does not give a nice description of the small basis for the exponential when we unfold the definitions. It
may be possible to do so using function graphs, as is done in the classical
setting of effective domain theory in~\cite[Section~4.1]{Smyth1977}, and we leave this as an open question.

An application of the closure results of
\cref{bilimit-has-small-basis,exponential-has-small-compact-basis} is that
Scott's \(D_\infty\) model of the untyped \(\lambda\)-calculus
from~\cref{sec:Scott-D-infty} has a small compact basis.

\begin{theorem}
  Scott's \(D_\infty\) has a small compact basis and in particular is
  algebraic.%
\end{theorem}
\begin{proof}
  By \cref{lifting-has-small-compact-basis} the \(\U_0\)-dcpo \(D_0\) has a
  small compact basis. Moreover, it is not just a \(\U_0\)-dcpo as it has
  suprema for all (not necessarily directed) families indexed by types in
  \(\U_0\), as \(D_0\) is isomorphic to \(\Omega_{\U_0}\).
  Hence, by induction it follows that
  each \(D_n\) is \(\U_0\)-sup-complete.
  Therefore, by induction and \cref{exponential-has-small-compact-basis} we get
  a small compact basis for each \(D_n\).
  Thus, by \cref{bilimit-has-small-basis}, the bilimit \(D_\infty\) has a small
  basis too.
\end{proof}

\section{Concluding remarks}
Taking inspiration from work in category theory by Johnstone and
Joyal~\cite{JohnstoneJoyal1982}, we gave predicatively adequate definitions of
continuous and algebraic dcpos, and discussed issues related to the absence of
the axiom of choice.
We~also presented predicative adaptations of the notions of a basis and the
round ideal completion.
The theory was accompanied by several examples: we described canonical small
compact bases for the lifting and the powerset, and considered the round ideal
completion of the dyadics.
We also showed that Scott's \(D_\infty\) has a small compact basis and is thus
algebraic in particular.

To prove that \(D_\infty\) has a small compact basis, we used that each \(D_n\)
is a \(\U_0\)-sup-lattice, so that we could apply the results
of~\cref{sec:exponentials-with-small-bases}.
\cref{lifting-has-small-compact-basis} tells us that \(\lifting_{\U_0}(\Nat)\)
has a small compact basis too, but to prove that the \(\U_0\)-dcpos in the Scott
model of PCF (see~\cite{deJong2021a}) have small compact bases using the techniques
of~\cref{sec:exponentials-with-small-bases}, we would need
\(\lifting_{\U_0}(\Nat)\) to be a \(\U_0\)-sup-lattice, but it isn't.
However, it is complete for \emph{bounded} families indexed by types in \(\U_0\)
and we believe that is possible to generalise the results
of~\cref{sec:exponentials-with-small-bases} from sup-lattices to bounded
complete posets.
Classically, this is fairly straightforward, but from preliminary considerations
it appears that constructively one needs to impose certain decidability criteria
on the bases of the dcpos. For instance that the partial order is decidable when
restricted to basis elements. Such decidability conditions were also studied
in~\cite{deJong2021b}.
These conditions should be satisfied by the bases of the dcpos in the Scott
model of PCF, but we leave a full treatment of bounded complete dcpos with bases
satisfying such conditions for future investigations.

\section{Acknowledgements}
We are grateful to Andrej Bauer and Dana Scott for independently asking us
whether we could have \(D_\infty\) in predicative univalent foundations, which
led to our reconstruction in~\cref{sec:Scott-D-infty}.
Furthermore, we also thank Andrej Bauer and Vincent Rahli for their detailed
feedback, and Ayberk Tosun for discussions on algebraic dcpos and the results
in~\cref{sec:basis-of-compact-elements} in particular.
Finally, we express our thanks to the anonymous referee for their useful
comments which helped to improve the paper.

The first author was supported by Cambridge Quantum Computing and Ilyas Khan
[Dissertation Fellowship in Homotopy Type Theory]; and The Royal Society [grant
reference URF{\textbackslash}R1{\textbackslash}191055].

\bibliography{references}

\begin{thebibliography}{75}
\expandafter\ifx\csname natexlab\endcsname\relax\def\natexlab#1{#1}\fi
\providecommand{\url}[1]{\texttt{#1}}
\providecommand{\href}[2]{#2}
\providecommand{\path}[1]{#1}
\providecommand{\DOIprefix}{doi:}
\providecommand{\ArXivprefix}{arXiv:}
\providecommand{\URLprefix}{URL: }
\providecommand{\Pubmedprefix}{pmid:}
\providecommand{\doi}[1]{\href{http://dx.doi.org/#1}{\path{#1}}}
\providecommand{\Pubmed}[1]{\href{pmid:#1}{\path{#1}}}
\providecommand{\bibinfo}[2]{#2}
\ifx\xfnm\relax \def\xfnm[#1]{\unskip,\space#1}\fi
\bibitem[{Abramsky and Jung(1994)}]{AbramskyJung1994}
\bibinfo{author}{Abramsky, S.}, \bibinfo{author}{Jung, A.},
  \bibinfo{year}{1994}.
\newblock \bibinfo{title}{Domain theory}, in: \bibinfo{editor}{Abramsky, S.},
  \bibinfo{editor}{Gabray, D.M.}, \bibinfo{editor}{Maibaum, T.S.E.} (Eds.),
  \bibinfo{booktitle}{Handbook of Logic in Computer Science}.
  \bibinfo{publisher}{Clarendon Press}. volume~\bibinfo{volume}{3}, pp.
  \bibinfo{pages}{1--168}.
\bibitem[{Aczel(2006)}]{Aczel2006}
\bibinfo{author}{Aczel, P.}, \bibinfo{year}{2006}.
\newblock \bibinfo{title}{Aspects of general topology in constructive set
  theory}.
\newblock \bibinfo{journal}{Annals of Pure and Applied Logic}
  \bibinfo{volume}{137}, \bibinfo{pages}{3--29}.
\newblock \DOIprefix\doi{10.1016/j.apal.2005.05.016}.
\bibitem[{Ahrens et~al.(2015)Ahrens, Kapulkin and
  Shulman}]{AhrensKapulkinShulman2015}
\bibinfo{author}{Ahrens, B.}, \bibinfo{author}{Kapulkin, K.},
  \bibinfo{author}{Shulman, M.}, \bibinfo{year}{2015}.
\newblock \bibinfo{title}{Univalent categories and the {Rezk} completion}.
\newblock \bibinfo{journal}{Mathematical Structures in Computer Science}
  \bibinfo{volume}{25}, \bibinfo{pages}{1010--1039}.
\newblock \DOIprefix\doi{10.1017/S0960129514000486}.
\bibitem[{Arrieta et~al.(2024)Arrieta, Escard\'o and
  Tosun}]{ArrietaEscardoTosun2024}
\bibinfo{author}{Arrieta, I.}, \bibinfo{author}{Escard\'o, M.H.},
  \bibinfo{author}{Tosun, A.}, \bibinfo{year}{2024}.
\newblock \bibinfo{title}{The patch topology in univalent foundations}.
\newblock \href{http://arxiv.org/abs/2402.03134}{{\tt arXiv:2402.03134}}.
\bibitem[{Bauer and Kavkler(2009)}]{BauerKavkler2009}
\bibinfo{author}{Bauer, A.}, \bibinfo{author}{Kavkler, I.},
  \bibinfo{year}{2009}.
\newblock \bibinfo{title}{A constructive theory of continuous domains suitable
  for implementation}.
\newblock \bibinfo{journal}{Annals of Pure and Applied Logic}
  \bibinfo{volume}{159}, \bibinfo{pages}{251--267}.
\newblock \DOIprefix\doi{10.1016/j.apal.2008.09.025}.
\bibitem[{Benton et~al.(2009)Benton, Kennedy and
  Varming}]{BentonKennedyVarming2009}
\bibinfo{author}{Benton, N.}, \bibinfo{author}{Kennedy, A.},
  \bibinfo{author}{Varming, C.}, \bibinfo{year}{2009}.
\newblock \bibinfo{title}{Some domain theory and denotational semantics in
  {Coq}}, in: \bibinfo{editor}{Berghofer, S.}, \bibinfo{editor}{Nipkow, T.},
  \bibinfo{editor}{Urban, C.}, \bibinfo{editor}{Wenzel, M.} (Eds.),
  \bibinfo{booktitle}{Theorem Proving in Higher Order Logics},
  \bibinfo{publisher}{Springer}. pp. \bibinfo{pages}{115--130}.
\newblock \DOIprefix\doi{10.1007/978-3-642-03359-9_10}.
\bibitem[{Bishop(1967)}]{Bishop1967}
\bibinfo{author}{Bishop, E.}, \bibinfo{year}{1967}.
\newblock \bibinfo{title}{Foundations of Constructive Analysis}.
\newblock \bibinfo{publisher}{McGraw-Hill Book Company}.
\bibitem[{van Collem(2023)}]{vanCollem2023}
\bibinfo{author}{van Collem, S.}, \bibinfo{year}{2023}.
\newblock \bibinfo{title}{Impredicative continuity}.
\newblock
  \bibinfo{howpublished}{\url{https://cs.bham.ac.uk/~mhe/TypeTopology/DomainTheory.BasesAndContinuity.ContinuityImpredicative.html}}.
\newblock \bibinfo{note}{\Agda\ formalisation in \TypeTopology}.
\bibitem[{Coquand et~al.(2003)Coquand, Sambin, Smith and
  Valentini}]{CoquandEtAl2003}
\bibinfo{author}{Coquand, T.}, \bibinfo{author}{Sambin, G.},
  \bibinfo{author}{Smith, J.}, \bibinfo{author}{Valentini, S.},
  \bibinfo{year}{2003}.
\newblock \bibinfo{title}{Inductively generated formal topologies}.
\newblock \bibinfo{journal}{Annals of Pure and Applied Logic}
  \bibinfo{volume}{124}, \bibinfo{pages}{71--106}.
\newblock \DOIprefix\doi{10.1016/s0168-0072(03)00052-6}.
\bibitem[{Coquand and Spiwack(2010)}]{CoquandSpiwack2010}
\bibinfo{author}{Coquand, T.}, \bibinfo{author}{Spiwack, A.},
  \bibinfo{year}{2010}.
\newblock \bibinfo{title}{Constructively finite?}, in:
  \bibinfo{editor}{Lamb\'an, L.}, \bibinfo{editor}{Romero, A.},
  \bibinfo{editor}{Rubio, J.} (Eds.), \bibinfo{booktitle}{Contribuciones
  cient\'ificas en honor de Mirian Andr\'es G\'omez}.
  \bibinfo{publisher}{Universidad de La Rioja}, pp. \bibinfo{pages}{217--230}.
\newblock \URLprefix
  \url{https://dialnet.unirioja.es/servlet/articulo?codigo=3217816}.
\bibitem[{Dockins(2014)}]{Dockins2014}
\bibinfo{author}{Dockins, R.}, \bibinfo{year}{2014}.
\newblock \bibinfo{title}{Formalized, effective domain theory in {Coq}}, in:
  \bibinfo{editor}{Klein, G.}, \bibinfo{editor}{Gamboa, R.} (Eds.),
  \bibinfo{booktitle}{Interactive Theorem Proving},
  \bibinfo{publisher}{Springer}. pp. \bibinfo{pages}{209--225}.
\newblock \DOIprefix\doi{10.1007/978-3-319-08970-6_14}.
\bibitem[{Escard\'o(2019)}]{Escardo2019}
\bibinfo{author}{Escard\'o, M.H.}, \bibinfo{year}{2019}.
\newblock \bibinfo{title}{Introduction to univalent foundations of mathematics
  with {Agda}}.
\newblock \href{http://arxiv.org/abs/1911.00580}{{\tt arXiv:1911.00580}}.
\bibitem[{Escard{\'o} and Knapp(2017)}]{EscardoKnapp2017}
\bibinfo{author}{Escard{\'o}, M.H.}, \bibinfo{author}{Knapp, C.M.},
  \bibinfo{year}{2017}.
\newblock \bibinfo{title}{Partial elements and recursion via dominances in
  univalent type theory}, in: \bibinfo{editor}{Goranko, V.},
  \bibinfo{editor}{Dam, M.} (Eds.), \bibinfo{booktitle}{26th EACSL Annual
  Conference on Computer Science Logic (CSL 2017)}, \bibinfo{publisher}{Schloss
  Dagstuhl--Leibniz-Zentrum f\"ur Informatik}. pp.
  \bibinfo{pages}{21:1--21:16}.
\newblock \DOIprefix\doi{10.4230/LIPIcs.CSL.2017.21}.
\bibitem[{Escard\'o et~al.()}]{TypeTopology}
\bibinfo{author}{Escard\'o, M.H.}, et~al., .
\newblock \bibinfo{title}{{TypeTopology}}.
\newblock
  \bibinfo{howpublished}{\url{https://www.cs.bham.ac.uk/~mhe/TypeTopology/index.html}}.
\newblock \bibinfo{note}{\Agda\ development.
  \url{https://github.com/martinescardo/TypeTopology}, commit
  \href{https://github.com/martinescardo/TypeTopology/tree/60d7bd9}{\texttt{60d7bd9}}}.
\bibitem[{Escard\'o et~al.(Since 2018)}]{TypeTopologyOrdinals}
\bibinfo{author}{Escard\'o, M.H.}, et~al., \bibinfo{year}{Since 2018}.
\newblock \bibinfo{title}{Ordinals in univalent type theory in {Agda}
  notation}.
\newblock
  \bibinfo{howpublished}{\url{https://www.cs.bham.ac.uk/~mhe/TypeTopology/Ordinals.index.html}}.
\newblock
  \bibinfo{note}{\url{https://github.com/martinescardo/TypeTopology/tree/master/source/Ordinals/index.lagda},
  commit
  \href{https://github.com/martinescardo/TypeTopology/tree/60d7bd9/source/Ordinals/index.lagda}{\texttt{60d7bd9}}}.
\bibitem[{Frumin et~al.(2018)Frumin, Geuvers, Gondelman and van~der
  Weide}]{FruminEtAl2018}
\bibinfo{author}{Frumin, D.}, \bibinfo{author}{Geuvers, H.},
  \bibinfo{author}{Gondelman, L.}, \bibinfo{author}{van~der Weide, N.},
  \bibinfo{year}{2018}.
\newblock \bibinfo{title}{Finite sets in homotopy type theory}, in:
  \bibinfo{editor}{Andronick, J.}, \bibinfo{editor}{Felty, A.} (Eds.),
  \bibinfo{booktitle}{CPP'18}, \bibinfo{publisher}{Association for Computing
  Machinery}. pp. \bibinfo{pages}{201--214}.
\newblock \DOIprefix\doi{10.1145/3167085}.
\bibitem[{Gierz et~al.(2003)Gierz, Hofmann, Keimel, Lawson, Mislove and
  Scott}]{GierzEtAl2003}
\bibinfo{author}{Gierz, G.}, \bibinfo{author}{Hofmann, K.H.},
  \bibinfo{author}{Keimel, K.}, \bibinfo{author}{Lawson, J.D.},
  \bibinfo{author}{Mislove, M.}, \bibinfo{author}{Scott, D.S.},
  \bibinfo{year}{2003}.
\newblock \bibinfo{title}{Continuous Lattices and Domains}.
  volume~\bibinfo{volume}{93} of \textit{\bibinfo{series}{Encyclopedia of
  Mathematics and its Applications}}.
\newblock \bibinfo{publisher}{Cambridge University Press}.
\newblock \DOIprefix\doi{10.1017/CBO9780511542725}.
\bibitem[{Grothendieck and Verdier(1972)}]{SGA41}
\bibinfo{author}{Grothendieck, A.}, \bibinfo{author}{Verdier, J.L.},
  \bibinfo{year}{1972}.
\newblock \bibinfo{title}{Prefaisceaux}, in: \bibinfo{editor}{Heidelberg,
  A.D.}, \bibinfo{editor}{Eckmann, B.} (Eds.), \bibinfo{booktitle}{Th\'eorie
  des Topos et Cohomologie Etale des Sch\'emas}, \bibinfo{publisher}{Springer}.
  pp. \bibinfo{pages}{1--217}.
\newblock \DOIprefix\doi{10.1007/BFB0081552}.
\bibitem[{Hart(2020)}]{Hart2020}
\bibinfo{author}{Hart, B.}, \bibinfo{year}{2020}.
\newblock \bibinfo{title}{Investigating Properties of {PCF} in {Agda}}.
\newblock \bibinfo{type}{Final year {MSci} project}. School of Computer
  Science, University of Birmingham.
\newblock \bibinfo{note}{\textsc{url}:
  \url{https://raw.githubusercontent.com/BrendanHart/Investigating-Properties-of-PCF/master/InvestigatingPropertiesOfPCFInAgda.pdf}.
  \Agda\ code available at
  \url{https://github.com/BrendanHart/Investigating-Properties-of-PCF}}.
\bibitem[{Hart(2023)}]{TypeTopologyHart}
\bibinfo{author}{Hart, B.}, \bibinfo{year}{2023}.
\newblock \bibinfo{title}{Formalisation of the {Scott} model of {PCF} in
  {Agda}}.
\newblock
  \bibinfo{howpublished}{\url{https://www.cs.bham.ac.uk/~mhe/TypeTopology/Published.PCF.Lambda.index.html}}.
\newblock
  \bibinfo{note}{\url{https://github.com/martinescardo/TypeTopology/tree/master/source/PCF/Lambda/index.lagda},
  commit
  \href{https://github.com/martinescardo/TypeTopology/tree/60d7bd9/source/PCF/Lambda/index.lagda}{\texttt{60d7bd9}}}.
\bibitem[{Hedberg(1996)}]{Hedberg1996}
\bibinfo{author}{Hedberg, M.}, \bibinfo{year}{1996}.
\newblock \bibinfo{title}{A type-theoretic interpretation of constructive
  domain theory}.
\newblock \bibinfo{journal}{Journal of Automated Reasoning}
  \bibinfo{volume}{16}, \bibinfo{pages}{369--425}.
\newblock \DOIprefix\doi{10.1007/BF00252182}.
\bibitem[{Hyland(1991)}]{Hyland1991}
\bibinfo{author}{Hyland, J.M.E.}, \bibinfo{year}{1991}.
\newblock \bibinfo{title}{First steps in synthetic domain theory}, in:
  \bibinfo{editor}{Carboni, A.}, \bibinfo{editor}{Pedicchio, M.C.},
  \bibinfo{editor}{Rosolini, G.} (Eds.), \bibinfo{booktitle}{Category Theory},
  \bibinfo{publisher}{Springer}. pp. \bibinfo{pages}{131--156}.
\newblock \DOIprefix\doi{10.1007/BFB0084217}.
\bibitem[{Jia et~al.(2015)Jia, Jung, Kou, Li and Zhao}]{JiaEtAl2015}
\bibinfo{author}{Jia, X.}, \bibinfo{author}{Jung, A.}, \bibinfo{author}{Kou,
  H.}, \bibinfo{author}{Li, Q.}, \bibinfo{author}{Zhao, H.},
  \bibinfo{year}{2015}.
\newblock \bibinfo{title}{All cartesian closed categories of quasicontinuous
  domains consist of domains}.
\newblock \bibinfo{journal}{Theoretical Computer Science}
  \bibinfo{volume}{594}, \bibinfo{pages}{143--150}.
\newblock \DOIprefix\doi{10.1016/j.tcs.2015.05.014}.
\bibitem[{Johnstone and Joyal(1982)}]{JohnstoneJoyal1982}
\bibinfo{author}{Johnstone, P.}, \bibinfo{author}{Joyal, A.},
  \bibinfo{year}{1982}.
\newblock \bibinfo{title}{Continuous categories and exponentiable toposes}.
\newblock \bibinfo{journal}{Journal of Pure and Applied Algebra}
  \bibinfo{volume}{25}, \bibinfo{pages}{255--296}.
\newblock \DOIprefix\doi{10.1016/0022-4049(82)90083-4}.
\bibitem[{Johnstone(1977)}]{Johnstone1977}
\bibinfo{author}{Johnstone, P.T.}, \bibinfo{year}{1977}.
\newblock \bibinfo{title}{Topos Theory}.
\newblock \bibinfo{publisher}{Academic Press}.
\newblock \bibinfo{note}{Reprinted by Dover Publications in 2014}.
\bibitem[{Johnstone(2002)}]{Johnstone2002}
\bibinfo{author}{Johnstone, P.T.}, \bibinfo{year}{2002}.
\newblock \bibinfo{title}{Sketches of an Elephant}. volume~\bibinfo{volume}{2}
  of \textit{\bibinfo{series}{Oxford Logic Guides}}.
\newblock \bibinfo{publisher}{Oxford Science Publications}.
\newblock \DOIprefix\doi{10.1093/oso/9780198515982.001.0001}.
\bibitem[{de~Jong(2021a)}]{deJong2021a}
\bibinfo{author}{de~Jong, T.}, \bibinfo{year}{2021}a.
\newblock \bibinfo{title}{The {Scott} model of {PCF} in univalent type theory}.
\newblock \bibinfo{journal}{Mathematical Structures in Computer Science}
  \bibinfo{volume}{31}, \bibinfo{pages}{1270--1300}.
\newblock \DOIprefix\doi{10.1017/S0960129521000153}.
\bibitem[{de~Jong(2021b)}]{deJong2021b}
\bibinfo{author}{de~Jong, T.}, \bibinfo{year}{2021}b.
\newblock \bibinfo{title}{Sharp elements and apartness in domains}, in:
  \bibinfo{editor}{Sokolova, A.} (Ed.), \bibinfo{booktitle}{37th Conference on
  the Mathematical Foundations of Programming Semantics (MFPS 2021)},
  \bibinfo{publisher}{Open Publishing Association}. pp.
  \bibinfo{pages}{134--151}.
\newblock \DOIprefix\doi{10.4204/EPTCS.351.9}.
\bibitem[{de~Jong(2022)}]{deJongThesis}
\bibinfo{author}{de~Jong, T.}, \bibinfo{year}{2022}.
\newblock \bibinfo{title}{Domain Theory in Constructive and Predicative
  Univalent Foundations}.
\newblock Ph.D. thesis. University of Birmingham.
\newblock \URLprefix \url{https://etheses.bham.ac.uk/id/eprint/13401/},
  \href{http://arxiv.org/abs/2301.12405}{{\tt arXiv:2301.12405}}.
\bibitem[{de~Jong(2025)}]{TypeTopologyPaper}
\bibinfo{author}{de~Jong, T.}, \bibinfo{year}{2025}.
\newblock \bibinfo{title}{Index for ``{Continuous and algebraic domains in
  univalent foundations}''}.
\newblock
  \bibinfo{howpublished}{\url{https://www.cs.bham.ac.uk/~mhe/TypeTopology/DomainTheory.Continuous-and-algebraic-domains.html}}.
\newblock
  \bibinfo{note}{\url{https://github.com/martinescardo/TypeTopology/tree/master/source/DomainTheory/Continuous-and-algebraic-domains.lagda},
  commit
  \href{https://github.com/martinescardo/TypeTopology/tree/60d7bd9/source/DomainTheory/Continuous-and-algebraic-domains.lagda}{\texttt{60d7bd9}}}.
\bibitem[{de~Jong and Escard{\'{o}}(2021)}]{deJongEscardo2021a}
\bibinfo{author}{de~Jong, T.}, \bibinfo{author}{Escard{\'{o}}, M.H.},
  \bibinfo{year}{2021}.
\newblock \bibinfo{title}{Domain theory in constructive and predicative
  univalent foundations}, in: \bibinfo{editor}{Baier, C.},
  \bibinfo{editor}{Goubault-Larrecq, J.} (Eds.), \bibinfo{booktitle}{29th EACSL
  Annual Conference on Computer Science Logic (CSL 2021)},
  \bibinfo{publisher}{Schloss Dagstuhl--Leibniz-Zentrum f{\"u}r Informatik}.
  pp. \bibinfo{pages}{28:1--28:18}.
\newblock \DOIprefix\doi{10.4230/LIPIcs.CSL.2021.28}. \bibinfo{note}{expanded
  version with full proofs available on arXiv:
  \href{https://arxiv.org/abs/2008.01422}{\texttt{2008.01422 [math.LO]}}}.
\bibitem[{de~Jong and Escard\'o(2023)}]{deJongEscardo2023}
\bibinfo{author}{de~Jong, T.}, \bibinfo{author}{Escard\'o, M.H.},
  \bibinfo{year}{2023}.
\newblock \bibinfo{title}{On small types in univalent foundations}.
\newblock \bibinfo{journal}{Logical Methods in Computer Science}
  \bibinfo{volume}{19}, \bibinfo{pages}{8:1--8:33}.
\newblock \DOIprefix\doi{10.46298/lmcs-19(2:8)2023}.
\bibitem[{de~Jong et~al.(2023)de~Jong, Kraus, Forsberg and Xu}]{LICS2023}
\bibinfo{author}{de~Jong, T.}, \bibinfo{author}{Kraus, N.},
  \bibinfo{author}{Forsberg, F.N.}, \bibinfo{author}{Xu, C.},
  \bibinfo{year}{2023}.
\newblock \bibinfo{title}{Set-theoretic and type-theoretic ordinals coincide},
  in: \bibinfo{booktitle}{2023 38th Annual ACM/IEEE Symposium on Logic in
  Computer Science (LICS)}, \bibinfo{publisher}{IEEE}.
\newblock \DOIprefix\doi{10.1109/lics56636.2023.10175762}.
  \bibinfo{note}{publicly available as
  \href{https://arxiv.org/abs/2301.10696}{\texttt{arXiv: 2301.10696}}}.
\bibitem[{Jung(1989)}]{Jung1988}
\bibinfo{author}{Jung, A.}, \bibinfo{year}{1989}.
\newblock \bibinfo{title}{Cartesian Closed Categories of Domains}.
  volume~\bibinfo{volume}{66} of \textit{\bibinfo{series}{CWI Tracts}}.
\newblock \bibinfo{publisher}{Centrum voor Wiskunde en Informatica (Centre for
  Mathematics and Computer Science)}.
\newblock \URLprefix \url{https://ir.cwi.nl/pub/13191}.
\bibitem[{Kawai(2017)}]{Kawai2017}
\bibinfo{author}{Kawai, T.}, \bibinfo{year}{2017}.
\newblock \bibinfo{title}{Geometric theories of patch and {Lawson} topologies}.
\newblock \href{http://arxiv.org/abs/1709.06403}{{\tt arXiv:1709.06403}}.
\bibitem[{Kawai(2021)}]{Kawai2021}
\bibinfo{author}{Kawai, T.}, \bibinfo{year}{2021}.
\newblock \bibinfo{title}{Predicative theories of continuous lattices}.
\newblock \bibinfo{journal}{Logical Methods in Computer Science}
  \bibinfo{volume}{17}.
\newblock \DOIprefix\doi{10.23638/LMCS-17(2:22)2021}.
\bibitem[{Knapp(2018)}]{Knapp2018}
\bibinfo{author}{Knapp, C.}, \bibinfo{year}{2018}.
\newblock \bibinfo{title}{Partial Functions and Recursion in Univalent Type
  Theory}.
\newblock Ph.D. thesis. School of Computer Science, University of Birmingham.
\newblock \URLprefix \url{https://etheses.bham.ac.uk/id/eprint/8448/}.
\bibitem[{Kock(1991)}]{Kock1991}
\bibinfo{author}{Kock, A.}, \bibinfo{year}{1991}.
\newblock \bibinfo{title}{Algebras for the partial map classifier monad}, in:
  \bibinfo{editor}{Carboni, A.}, \bibinfo{editor}{Pedicchio, M.C.},
  \bibinfo{editor}{Rosolini, G.} (Eds.), \bibinfo{booktitle}{Category Theory},
  \bibinfo{publisher}{Springer}. pp. \bibinfo{pages}{262--278}.
\newblock \DOIprefix\doi{10.1007/BFB0084225}.
\bibitem[{Kraus et~al.(2017)Kraus, Escard\'o, Coquand and
  Altenkirch}]{KrausEtAl2017}
\bibinfo{author}{Kraus, N.}, \bibinfo{author}{Escard\'o, M.H.},
  \bibinfo{author}{Coquand, T.}, \bibinfo{author}{Altenkirch, T.},
  \bibinfo{year}{2017}.
\newblock \bibinfo{title}{Notions of anonymous existence in {Martin-L\"of Type
  Theory}}.
\newblock \bibinfo{journal}{Logical Methods in Computer Science}
  \bibinfo{volume}{13}.
\newblock \DOIprefix\doi{10.23638/LMCS-13(1:15)2017}.
\bibitem[{Kuratowski(1920)}]{Kuratowski1920}
\bibinfo{author}{Kuratowski, C.}, \bibinfo{year}{1920}.
\newblock \bibinfo{title}{Sur la notion d'ensemble fini}.
\newblock \bibinfo{journal}{Fundamenta Mathematicae} \bibinfo{volume}{1},
  \bibinfo{pages}{129--131}.
\newblock \DOIprefix\doi{10.4064/fm-1-1-129-131}.
\bibitem[{Lidell(2020)}]{Lidell2020}
\bibinfo{author}{Lidell, D.}, \bibinfo{year}{2020}.
\newblock \bibinfo{title}{Formalizing domain models of the typed and the
  untyped lambda calculus in {Agda}}.
\newblock \bibinfo{type}{Master's thesis}. {Department of Computer Science and
  Engineering}, Chalmers University of Technology and University of Gothenburg.
\newblock \URLprefix \url{http://hdl.handle.net/2077/67193}.
\bibitem[{Longley(1995)}]{Longley1995}
\bibinfo{author}{Longley, J.}, \bibinfo{year}{1995}.
\newblock \bibinfo{title}{Realizability Toposes and Language Semantics}.
\newblock Ph.D. thesis. Department of Computer Science, University of
  Edinburgh.
\newblock \URLprefix
  \url{https://www.lfcs.inf.ed.ac.uk/reports/95/ECS-LFCS-95-332}.
\bibitem[{Longley and Normann(2015)}]{LongleyNormann2015}
\bibinfo{author}{Longley, J.}, \bibinfo{author}{Normann, D.},
  \bibinfo{year}{2015}.
\newblock \bibinfo{title}{Higher-Order Computability}.
\newblock Theory and Applications of Computability,
  \bibinfo{publisher}{Springer}.
\newblock \DOIprefix\doi{10.1007/978-3-662-47992-6}.
\bibitem[{Magnusson(1994)}]{Magnusson1995}
\bibinfo{author}{Magnusson, L.}, \bibinfo{year}{1994}.
\newblock \bibinfo{title}{The Implementation of {ALF}}.
\newblock Ph.D. thesis. Department of Computing Science, Chalmers University of
  Technology and University of Gothenburg.
\newblock \URLprefix \url{http://hdl.handle.net/2077/12916}.
\bibitem[{Maietti and Valentini(2004)}]{MaiettiValentini2004}
\bibinfo{author}{Maietti, M.E.}, \bibinfo{author}{Valentini, S.},
  \bibinfo{year}{2004}.
\newblock \bibinfo{title}{Exponentiation of {Scott} formal topologies}.
\newblock \bibinfo{journal}{Electronic Notes in Theoretical Computer Science}
  \bibinfo{volume}{73}, \bibinfo{pages}{111--131}.
\newblock \DOIprefix\doi{10.1016/j.entcs.2004.08.005}.
\bibitem[{Makkai and Par\'e(1989)}]{MakkaiPare1989}
\bibinfo{author}{Makkai, M.}, \bibinfo{author}{Par\'e, R.},
  \bibinfo{year}{1989}.
\newblock \bibinfo{title}{Accessible Categories: The Foundations of Categorical
  Model Theory}. volume \bibinfo{volume}{104} of
  \textit{\bibinfo{series}{Contemporary Mathematics}}.
\newblock \bibinfo{publisher}{American Mathematical Society}.
\newblock \DOIprefix\doi{10.1090/conm/104}.
\bibitem[{Negri(2002)}]{Negri2002}
\bibinfo{author}{Negri, S.}, \bibinfo{year}{2002}.
\newblock \bibinfo{title}{Continuous domains as formal spaces}.
\newblock \bibinfo{journal}{Mathematical Structures in Computer Science}
  \bibinfo{volume}{12}, \bibinfo{pages}{19--52}.
\newblock \DOIprefix\doi{10.1017/S0960129501003450}.
\bibitem[{van Oosten(2008)}]{vanOosten2008}
\bibinfo{author}{van Oosten, J.}, \bibinfo{year}{2008}.
\newblock \bibinfo{title}{Realizability: An Introduction to its Categorical
  Side}. volume \bibinfo{volume}{152} of \textit{\bibinfo{series}{Studies in
  Logic and the Foundations of Mathematics}}.
\newblock \bibinfo{publisher}{Elsevier}.
\newblock \DOIprefix\doi{10.1016/s0049-237X(08)X8001-2}.
\bibitem[{Pattinson and Mohammadian(2021)}]{PattinsonMohammadian2021}
\bibinfo{author}{Pattinson, D.}, \bibinfo{author}{Mohammadian, M.},
  \bibinfo{year}{2021}.
\newblock \bibinfo{title}{Constructive domains with classical witnesses}.
\newblock \bibinfo{journal}{Logical Methods in Computer Science}
  \bibinfo{volume}{17}.
\newblock \DOIprefix\doi{10.23638/LMCS-17(1:19)2021}.
\bibitem[{Phao(1991)}]{Phao1991}
\bibinfo{author}{Phao, W.}, \bibinfo{year}{1991}.
\newblock \bibinfo{title}{Domain Theory in Realizability Toposes}.
\newblock Ph.D. thesis. Department of Computer Science, University of
  Edinburgh.
\newblock \URLprefix
  \url{https://www.lfcs.inf.ed.ac.uk/reports/91/ECS-LFCS-91-171}.
\bibitem[{Plotkin(1977)}]{Plotkin1977}
\bibinfo{author}{Plotkin, G.}, \bibinfo{year}{1977}.
\newblock \bibinfo{title}{{LCF} considered as a programming language}.
\newblock \bibinfo{journal}{Theoretical Computer Science} \bibinfo{volume}{5},
  \bibinfo{pages}{223--255}.
\newblock \DOIprefix\doi{10.1016/0304-3975(77)90044-5}.
\bibitem[{Reus(1999)}]{Reus1999}
\bibinfo{author}{Reus, B.}, \bibinfo{year}{1999}.
\newblock \bibinfo{title}{Formalizing synthetic domain theory}.
\newblock \bibinfo{journal}{Journal of Automated Reasoning}
  \bibinfo{volume}{23}, \bibinfo{pages}{411--444}.
\newblock \DOIprefix\doi{10.1023/A:1006258506401}.
\bibitem[{Reus and Streicher(1999)}]{ReusStreicher1999}
\bibinfo{author}{Reus, B.}, \bibinfo{author}{Streicher, T.},
  \bibinfo{year}{1999}.
\newblock \bibinfo{title}{General synthetic domain theory --- a logical
  approach}.
\newblock \bibinfo{journal}{Mathematical Structures in Computer Science}
  \bibinfo{volume}{9}, \bibinfo{pages}{177--223}.
\newblock \DOIprefix\doi{10.1017/S096012959900273X}.
\bibitem[{Rijke(2022)}]{Rijke2022}
\bibinfo{author}{Rijke, E.}, \bibinfo{year}{2022}.
\newblock \bibinfo{title}{Introduction to homotopy type theory}.
\newblock \href{http://arxiv.org/abs/2212.11082}{{\tt arXiv:2212.11082}}.
  \bibinfo{note}{pre-publication of forthcoming textbook}.
\bibitem[{Rosolini(1986)}]{Rosolini1986}
\bibinfo{author}{Rosolini, G.}, \bibinfo{year}{1986}.
\newblock \bibinfo{title}{Continuity and effectiveness in topoi}.
\newblock Ph.D. thesis. University of Oxford.
\newblock \URLprefix
  \url{https://www.researchgate.net/publication/35103849_Continuity_and_effectiveness_in_topoi}.
\bibitem[{Rosolini(1987)}]{Rosolini1987}
\bibinfo{author}{Rosolini, G.}, \bibinfo{year}{1987}.
\newblock \bibinfo{title}{Categories and effective computations}, in:
  \bibinfo{editor}{Pitt, D.H.}, \bibinfo{editor}{Poign\'e, A.},
  \bibinfo{editor}{Rydeheard, D.E.} (Eds.), \bibinfo{booktitle}{Category Theory
  and Computer Science}, \bibinfo{publisher}{Springer}.
\newblock \DOIprefix\doi{10.1007/3-540-18508-9_17}.
\bibitem[{Sambin(1987)}]{Sambin1987}
\bibinfo{author}{Sambin, G.}, \bibinfo{year}{1987}.
\newblock \bibinfo{title}{Intuitionistic formal spaces --- a first
  communication}, in: \bibinfo{editor}{Skordev, D.G.} (Ed.),
  \bibinfo{booktitle}{Mathematical logic and its applications}.
  \bibinfo{publisher}{Springer}, pp. \bibinfo{pages}{187--204}.
\newblock \DOIprefix\doi{10.1007/978-1-4613-0897-3_12}.
\bibitem[{Sambin(2003)}]{Sambin2003}
\bibinfo{author}{Sambin, G.}, \bibinfo{year}{2003}.
\newblock \bibinfo{title}{Some points in formal topology}.
\newblock \bibinfo{journal}{Theoretical Computer Science}
  \bibinfo{volume}{305}, \bibinfo{pages}{347--408}.
\newblock \DOIprefix\doi{10.1016/S0304-3975(02)00704-1}.
\bibitem[{Sambin et~al.(1996)Sambin, Valentini and
  Virgili}]{SambinValentiniVirgili1996}
\bibinfo{author}{Sambin, G.}, \bibinfo{author}{Valentini, S.},
  \bibinfo{author}{Virgili, P.}, \bibinfo{year}{1996}.
\newblock \bibinfo{title}{Constructive domain theory as a branch of
  intuitionistic pointfree topology}.
\newblock \bibinfo{journal}{Theoretical Computer Science}
  \bibinfo{volume}{159}, \bibinfo{pages}{319--341}.
\newblock \DOIprefix\doi{10.1016/0304-3975(95)00169-7}.
\bibitem[{Scott(1970)}]{Scott1970}
\bibinfo{author}{Scott, D.}, \bibinfo{year}{1970}.
\newblock \bibinfo{title}{Outline of a Mathematical Theory of Computation}.
\newblock \bibinfo{type}{Technical Report} \bibinfo{number}{PRG02}. Oxford
  University Computing Laboratory.
\newblock \URLprefix
  \url{https://www.cs.ox.ac.uk/publications/publication3720-abstract.html}.
\bibitem[{Scott(1972)}]{Scott1972}
\bibinfo{author}{Scott, D.S.}, \bibinfo{year}{1972}.
\newblock \bibinfo{title}{Continuous lattices}, in: \bibinfo{editor}{Lawvere,
  F.W.} (Ed.), \bibinfo{booktitle}{Toposes, Algebraic Geometry and Logic}.
  \bibinfo{publisher}{Springer}. volume \bibinfo{volume}{274} of
  \textit{\bibinfo{series}{Lecture Notes in Mathematics}}, pp.
  \bibinfo{pages}{97--136}.
\newblock \DOIprefix\doi{10.1007/BFB0073967}.
\bibitem[{Scott(1982a)}]{Scott1982a}
\bibinfo{author}{Scott, D.S.}, \bibinfo{year}{1982}a.
\newblock \bibinfo{title}{Domains for denotational semantics}, in:
  \bibinfo{editor}{Nielsen, M.}, \bibinfo{editor}{Schmidt, E.M.} (Eds.),
  \bibinfo{booktitle}{Automata, Languages and Programming},
  \bibinfo{publisher}{Springer-Verlag}. pp. \bibinfo{pages}{577--610}.
\newblock \DOIprefix\doi{10.1007/BFB0012801}.
\bibitem[{Scott(1982b)}]{Scott1982b}
\bibinfo{author}{Scott, D.S.}, \bibinfo{year}{1982}b.
\newblock \bibinfo{title}{Lectures on a mathematical theory of computation},
  in: \bibinfo{editor}{Broy, M.}, \bibinfo{editor}{Schmidt, G.} (Eds.),
  \bibinfo{booktitle}{Theoretical Foundations of Programming Methodology}.
  \bibinfo{publisher}{Springer}. volume~\bibinfo{volume}{91} of
  \textit{\bibinfo{series}{{NATO} Advanced Study Institutes Series}}, pp.
  \bibinfo{pages}{145--292}.
\newblock \DOIprefix\doi{10.1007/978-94-009-7893-5_9}.
\bibitem[{Scott(1993)}]{Scott1993}
\bibinfo{author}{Scott, D.S.}, \bibinfo{year}{1993}.
\newblock \bibinfo{title}{A type-theoretical alternative to {ISWIM}, {CUCH},
  {OWHY}}.
\newblock \bibinfo{journal}{Theoretical Computer Science}
  \bibinfo{volume}{121}, \bibinfo{pages}{411--440}.
\newblock \DOIprefix\doi{10.1016/0304-3975(93)90095-B}.
\bibitem[{Shulman(2016)}]{Shulman2016}
\bibinfo{author}{Shulman, M.}, \bibinfo{year}{2016}.
\newblock \bibinfo{title}{Idempotents in intensional type theory}.
\newblock \bibinfo{journal}{Logical Methods in Computer Science}
  \bibinfo{volume}{12}, \bibinfo{pages}{1--24}.
\newblock \DOIprefix\doi{10.2168/LMCS-12(3:9)2016}.
\bibitem[{Shulman(2019)}]{Shulman2019}
\bibinfo{author}{Shulman, M.}, \bibinfo{year}{2019}.
\newblock \bibinfo{title}{All \((\infty,1)\)-toposes have strict univalent
  universes}.
\newblock \href{http://arxiv.org/abs/1904.07004}{{\tt arXiv:1904.07004}}.
\bibitem[{Smyth(1977)}]{Smyth1977}
\bibinfo{author}{Smyth, M.B.}, \bibinfo{year}{1977}.
\newblock \bibinfo{title}{Effectively given domains}.
\newblock \bibinfo{journal}{Theoretical Computer Science} \bibinfo{volume}{5},
  \bibinfo{pages}{257--274}.
\newblock \DOIprefix\doi{10.1016/0304-3975(77)90045-7}.
\bibitem[{Swan(2019a)}]{Swan2019b}
\bibinfo{author}{Swan, A.W.}, \bibinfo{year}{2019}a.
\newblock \bibinfo{title}{Choice, collection and covering in cubical sets}.
\newblock \URLprefix
  \url{https://www.math.uwo.ca/faculty/kapulkin/seminars/hottestfiles/Swan-2019-11-06-HoTTEST.pdf}.
  \bibinfo{note}{talk at \emph{Homotopy Type Theory Electronic Seminar Talks
  (HoTTEST)}, online}.
\bibitem[{Swan(2019b)}]{Swan2019a}
\bibinfo{author}{Swan, A.W.}, \bibinfo{year}{2019}b.
\newblock \bibinfo{title}{Counterexamples in cubical sets}.
\newblock \URLprefix
  \url{https://logic.math.su.se/mloc-2019/slides/Swan-mloc-2019-slides.pdf}.
  \bibinfo{note}{talk at \emph{Mathematical Logic and Constructivity: The Scope
  and Limits of Neutral Constructivism}, Stockholm, Sweden}.
\bibitem[{Taylor(1999)}]{Taylor1999}
\bibinfo{author}{Taylor, P.}, \bibinfo{year}{1999}.
\newblock \bibinfo{title}{Practical Foundations of Mathematics}.
  volume~\bibinfo{volume}{59} of \textit{\bibinfo{series}{Cambridge Studies in
  Advanced Mathematics}}.
\newblock \bibinfo{publisher}{Cambridge University Press}.
\bibitem[{Troelstra and van Dalen(1988)}]{TroelstraVanDalen1988}
\bibinfo{author}{Troelstra, A.S.}, \bibinfo{author}{van Dalen, D.},
  \bibinfo{year}{1988}.
\newblock \bibinfo{title}{Constructivism in Mathematics}. volume
  \bibinfo{volume}{123} of \textit{\bibinfo{series}{Studies in Logic and the
  Foundations of Mathematics}}.
\newblock \bibinfo{publisher}{Elsevier}.
\bibitem[{Uemura(2019)}]{Uemura2019}
\bibinfo{author}{Uemura, T.}, \bibinfo{year}{2019}.
\newblock \bibinfo{title}{Cubical assemblies, a univalent and impredicative
  universe and a failure of propositional resizing}, in:
  \bibinfo{editor}{Dybjer, P.}, \bibinfo{editor}{Santo, J.E.},
  \bibinfo{editor}{Pinto, L.} (Eds.), \bibinfo{booktitle}{24th International
  Conference on Types for Proofs and Programs (TYPES 2018)},
  \bibinfo{publisher}{Schloss Dagstuhl--Leibniz-Zentrum f\"ur Informatik}. pp.
  \bibinfo{pages}{7:1--7:20}.
\newblock \DOIprefix\doi{10.4230/LIPIcs.TYPES.2018.7}.
\bibitem[{{Univalent Foundations Program}(2013)}]{HoTTBook}
\bibinfo{author}{{Univalent Foundations Program}, T.}, \bibinfo{year}{2013}.
\newblock \bibinfo{title}{Homotopy Type Theory: Univalent Foundations of
  Mathematics}.
\newblock \bibinfo{publisher}{\url{https://homotopytypetheory.org/book}},
  \bibinfo{address}{Institute for Advanced Study}.
\bibitem[{Voevodsky(2011)}]{Voevodsky2011}
\bibinfo{author}{Voevodsky, V.}, \bibinfo{year}{2011}.
\newblock \bibinfo{title}{Resizing rules --- their use and semantic
  justification}.
\newblock \URLprefix
  \url{https://www.math.ias.edu/vladimir/sites/math.ias.edu.vladimir/files/2011_Bergen.pdf}.
  \bibinfo{note}{talk at \emph{18th International Conference on Types for
  Proofs and Programs (TYPES)}, Bergen, Norway}.
\bibitem[{Voevodsky(2015)}]{Voevodsky2015}
\bibinfo{author}{Voevodsky, V.}, \bibinfo{year}{2015}.
\newblock \bibinfo{title}{An experimental library of formalized mathematics
  based on the univalent foundations}.
\newblock \bibinfo{journal}{Mathematical Structures in Computer Science}
  \bibinfo{volume}{25}, \bibinfo{pages}{1278--1294}.
\newblock \DOIprefix\doi{10.1017/S0960129514000577}.

\end{thebibliography}

\end{document}